\documentclass[12pt]{article}

\usepackage{amsfonts,amsmath,amssymb,color,lscape,float,graphicx,tabularx,theorem}
\usepackage[margin=.75in]{geometry}
\usepackage{natbib}
\usepackage{rotating}

\renewcommand{ \baselinestretch}{1.3}

\newtheorem{theorem}{Theorem}
\newtheorem{definition}{Definition}
\newtheorem{lemma}{Lemma}
\newtheorem{remark}{Remark}
\newtheorem{proposition}{Proposition}
\newcommand{\Perp}{\perp\!\!\!\perp}

\newcommand{\qedl}{\hfill $\blacksquare$}
\newcommand{\qed}{\hfill $\blacksquare$}
\newcommand{\qedi}{\hfill $\square$}
\def\var{\textrm{var}}
\def\plim{\textrm{plim}}
\def\E{\mathbb{E}}

\newenvironment{proof}[1][Proof]{\begin{trivlist}
\item[\hskip \labelsep {\bfseries #1}]}{\end{trivlist}}

\begin{document}
\title{Realised quantile-based estimation of the integrated variance\thanks{Podolskij gratefully acknowledges financial support from CREATES funded by the Danish National Research Foundation. We would like to thank Peter Bank, \'{A}lvaro Cartea, Fulvio Corsi, Dick van Dijk, Dobrislav Dobrev, Bruce Lehman, Haikady Nagaraja, Kevin Sheppard (discussant), Neil Shephard and the seminar participants at the London-Oxford Financial Econometrics Study Group at Imperial College, Humboldt-Universit\"{a}t zu Berlin, ECARES of Universit\'{e} Libre de Bruxelles, LSE Department of Statistics, Federal Reverse Board, Washington, Amsterdam Business School, Tinbergen Institute, Rotterdam, the Workshop on Mathematical Finance for Young Researchers, Quantitative Products Laboratory, Berlin, and the SITE Summer Workshop 2008, Stanford for valuable comments and discussions.}}

\author{Kim Christensen\thanks{CREATES, School of Economics and Management, Aarhus University, Building 1322, Bartholins All\'{e} 10, 8000 Aarhus, Denmark. E-mail \texttt{kchristensen@creates.au.dk}.}
\and
Roel Oomen\thanks{Deutsche Bank AG, Winchester House, 1 Great Winchester Street, London EC2N 2DB, UK and affiliated with the Department of Quantitative Economics, the University of Amsterdam, The Netherlands. E-mail: \texttt{roel.ca.oomen@gmail.com}. Phone: +44 (0) 207 54 56989.}
\and
Mark Podolskij\thanks{Podolskij is with ETH Z\"{u}rich, Department of Mathematics, R\"{a}mistrasse 101, CH-8092 Z\"{u}rich, Switzerland and affiliated with CREATES, University of Aarhus, Denmark. E-mail: \texttt{mark.podolskij@math.ethz.ch}.}}

\date{January, 2010}

\maketitle

\begin{abstract}
In this paper, we propose a new jump robust quantile-based realised variance measure of ex-post return variation that can be computed using potentially noisy data. The estimator is consistent for the integrated variance and we present feasible central limit theorems which show that it converges at the best attainable rate and has excellent efficiency. Asymptotically, the quantile-based realised variance is immune to finite activity jumps and outliers in the price series, while in modified form the estimator is applicable with market microstructure noise and therefore operational on high-frequency data. Simulations show that it has superior robustness properties in finite sample, while an empirical application illustrates its use on equity data.


\end{abstract}

\thispagestyle{empty}

\newpage

\section{Introduction} \setcounter{page}{1}

In recent years, our understanding of asset price dynamics has been significantly enhanced by the increasing availability of intra-day financial tick data in conjunction with the rapid development and harnessing of the necessary econometric tools. Realised variance, defined as the sum of squared intra-period returns, has been a key driver in this literature \citep*[e.g.][]{andersen-bollerslev-diebold-labys:01a, barndorff-nielsen-shephard:02a} as it provides a simple yet highly efficient way to consistently estimate the quadratic variation of a price process. However, when faced with the realities of high-frequency data, the use of realised variance is limited in two important ways. First, realised variance is overly sensitive to an inherent feature of high-frequency data, namely market microstructure noise. In fact, it destroys the consistency of the estimator. Second, realised variance is an estimator of the total variation -- the sum of diffusive and jump variation -- and consequently cannot distinguish between these two fundamentally different sources of risk. As emphasized by \citet*{ait-sahalia:04a}, the ability to disentangle jumps from volatility is key for the valuation of derivatives \citep*[e.g.][]{merton:76a, duffie-pan-singleton:00a}, risk measurement and management \citep*[e.g.][]{duffie-pan:01a}, as well as asset allocation \citep*[e.g.][]{jarrow-rosenfeld:84a, liu-longstaff-pan:03a}. \citet*{andersen-bollerslev-diebold:07a} and \citet*{bollerslev-kretschmer-pigorsch-tauchen:09a} also point to its importance for the empirical modeling of asset price dynamics and forecasting of volatility.

In response, two largely separate strands of literature have emerged. One on jump-robust realised variance measures, most notably the influential bi-power variation of \citet*{barndorff-nielsen-shephard:04b} and the threshold estimators of \citet*{jacod:08a,mancini:04b,mancini:09a}. The other on noise-robust realised variance measures, where the main approaches are based on subsampling \citep*[see][]{zhang:06a, zhang-mykland-ait-sahalia:05a}, kernel-based autocovariance adjustments \citep*[see][]{barndorff-nielsen-hansen-lunde-shephard:08a, zhou:96a}, and pre-averaging methods \citep*[see][]{jacod-li-mykland-podolskij-vetter:09a, podolskij-vetter:09a}.\footnote{Other methods include sparse sampling \citep*[e.g.][]{andersen-bollerslev-diebold-labys:00a, bandi-russell:08a},
pre-filtering \citep*{andersen-bollerslev-diebold-ebens:01a, bollen-inder:02a, hansen-large-lunde:08a}, model-based corrections \citep*[e.g.][]{corsi-zumbach-muller-dacorogna:01a}, and wavelet-based methods of \citet*{fan-wang:07a}. Variations and extensions of the realised kernel and subsampling approach can be found in \cite{sun:06a} and \citet*{nolte-voev:07a}, respectively.} Yet, the jump-robust realised variance measures are typically not robust to noise and the noise-robust realised variance measures are typically not robust to jumps. The contribution of this paper is to integrate both issues: we develop a new quantile-based realised variance (QRV) measure, which constitutes the first estimator of integrated variance that, at the same time, is highly efficient and simultaneously robust to jumps and microstructure noise. We present a complete asymptotic theory for this class of estimators, with central limit theorems that hold up in the presence of finite activity jumps. In addition to noise and jump robustness, another appealing feature of our estimator is that it is robust to outliers in the price process. Outliers are unavoidable in high-frequency data (due to, for instance, delayed trade reporting, recording errors, decimal misplacement, etc.) and are often hard to filter out systematically due to their random and irregular nature while the vast quantities of data make ``visual inspection'' impractical. In the paper, we illustrate that this feature of our estimator can be crucial in practical applications.

The construction of the QRV is based on the fundamental relationship between quantiles and the spread of the normal distribution. A stylized example provides the key intuition. For i.i.d. Gaussian data with mean zero and variance $\sigma^2$, we know that the 95\% quantile equals $1.645\sigma$. Thus, with a measurement of a sample quantile, volatility can be estimated by inverting this relationship. Of course, other quantiles can be added to improve the efficiency of the estimator and as long as the selected quantiles are sufficiently far away from the extreme tails we have robustness to outliers. Estimators of this kind have a long history in the statistics literature: they can be traced back to \citet*{pearson:20a} and are further studied by \citet*{mosteller:46a}, \citet*{eisenberger-posner:65a}, and \citet*{david:70a}.

The scope of this paper is, however, more ambitious than the above illustration in that we aim to develop an estimator of the integrated variance under very weak conditions on the underlying process, need to deal with microstructure noise, and consider limits where data are sampled at an increasing rate over a fixed time window. More formally, let $\{X_{i/N}\}_{i=0}^N$ denote a time-series of logarithmic asset price observations and define returns as $\Delta_{i}^{N} X = X_{i / N} - X_{(i - 1) / N}$. The basic idea is to split the sample into $n$ smaller blocks each containing $m$ returns, $D_{i}^{m} X = \left( \Delta_{k}^{N} X \right)_{(i - 1) m + 1 \leq k \leq im}$ for $i = 1, \ldots, n$ with $N = mn$, and then study the behavior of the return quantiles on these blocks as $n$ grows large. In particular, we consider a ($k \times 1$) vector of quantiles $\overline{ \lambda} = (\lambda_{1}, \ldots, \lambda_{k})'$ and define the QRV as:
\begin{equation*}
QRV_N(m,\overline{\lambda},\alpha) \equiv \alpha' QRV_N(m,\overline{\lambda}),
\end{equation*}
where $\alpha$ is a $(k \times 1)$ vector of quantile weights, and $QRV_N(m,\overline{\lambda})$ is a $(k \times 1)$ vector with $j$th entry equal to
\begin{equation*}
QRV_{N} \left( m,\lambda_{j} \right) =  \frac{m}{N} \sum_{i = 1}^{N / m} \frac{q_{i} (m,\lambda_{j})}{\nu_1 \left( m, \lambda_{j} \right)}, \qquad \text{for~} \lambda_{j} \in (1/2,1).
\end{equation*}
Here
\begin{equation*}
q_{i}(m,\lambda) = g^{2}_{ \lambda m} \left( \sqrt{N} D_i^m X \right) + g^{2}_{m - \lambda m + 1} \left( \sqrt{N} D_i^m X \right),
\end{equation*}
is the realised symmetric squared $\lambda$-quantile, $\nu_1 \left( m, \lambda_{j} \right)$ is a normalizing constant that measures the expectation of $q_{i}(m,\lambda)$ under a standard normal, and the function $g_{k} \left( x \right) = x_{ \left( k \right)}$ extracts the $k$th order statistic of $x = \left( x_{1}, \ldots, x_{m} \right)$.

It is quite intuitive that QRV as defined above provides consistent and jump-robust estimates of the integrated variance as $n\to\infty$. As the number of blocks grows, they span an increasingly short interval so that in the limit and under weak assumptions on the price process, each block contains at most one jump and volatility within the block is locally constant. In this scenario, the term $q_{i}(m,\lambda_{j})/\nu_1 \left( m, \lambda_{j} \right)$ constitutes an estimator of the (scaled) return variance over the $i^{th}$ block, which is robust to jumps by the assumption that $\lambda_{max} = \max \{ \overline{ \lambda} \} < 1$. Summing across blocks then naturally yields a consistent estimator for the integrated variance. From the above, it is also clear that the QRV estimator can be formulated based on overlapping blocks. In the paper, we show that such a subsampled version of QRV further improves the efficiency of the estimator. In terms of asymptotic theory, we derive feasible central limit theorems for QRV, which show that for fixed $m$ and in the absence of noise, our estimator converges to the integrated variance at rate $N^{-1 / 2}$. With microstructure noise, our (modified) estimator converges at rate $N^{-1 / 4}$. In both cases, this is known to be the fastest possible rate. We carry on to show that for suitable choice of parameters the asymptotic variance of the QRV is close to the maximum likelihood bound in either situation. Moreover, all our consistency and central limit theorems are shown to be robust to the presence of finite activity jumps.

Implementation of the QRV requires the choice of some ``tuning'' parameters, such as the number of blocks $n$ or block length $m$, the quantiles $\overline{\lambda}$, and the quantile weights $\alpha$. For given block length and quantiles, we show how to select the quantile weights optimally to minimize the asymptotic variance of QRV. In the special case where $m \to \infty$, we can express these optimal weights in a simple way. This proves convenient for the implementation of the QRV, as it turns out we can use these asymptotic weights even for finite $m$ with hardly any loss of efficiency. Similarly, the choice of $m$ and $\overline{\lambda}$ can be based on asymptotic efficiency considerations, but we also emphasize that the robustness of QRV in finite sample is controlled by the joint choice of these parameters as the $(1 - \lambda_{max})m - 1$ largest negative and positive returns are discarded in each block of data. In the paper, we provide detailed guidance on how to choose $m$ and $\overline{\lambda}$ based on both asymptotic theory and practical considerations.

There are a number of recent papers related to our work. \citet*{podolskij-vetter:09a} combine pre-averaging with the bi-power variation measure to also obtain a jump- and noise-robust volatility estimator. In comparison to QRV, their estimator is inefficient and because it is based on bi-power variation it is not robust to outliers and the corresponding central limit theorems are not robust to jumps. \citet*{andersen-dobrev-schaumburg:08a} develop two jump-robust measures named MinRV and MedRV, but they rule out microstructure noise. Interestingly, we show that their estimators can be nested in our class of estimators and by extension we provide an asymptotic theory for the MinRV and MedRV in the presence of microstructure noise. Further, the idea to get rid of jumps through quantiles is somewhat similar in spirit to threshold estimators of the integrated variance \citep*[see, e.g.,][]{mancini:04a,jacod:08a}, where return observations larger than a pre-determined threshold are removed before computing realised variance. Because these estimators use a global threshold to pre-truncate the data there is a risk of retaining jumps when the threshold is set too high or removing an excessive amount of observations when the threshold does not fully encapsulate the returns in high volatility episodes. In both cases the estimator will be biased. In contrast, QRV sets a local threshold that adapts naturally to the magnitude of the returns observations in each block through the choice of quantiles. As such, it is able to be robust to both small and large jumps even in the presence of time-varying volatility. Finally, we note that QRV is also related to the literature on L-statistics \citep*[e.g.][]{vaart:98a}, and it bears some resemblance to the classical idea of trimmed mean estimation.

The remainder of this paper is organized as follows. In Section \ref{sec:QRVintro}, we introduce the QRV estimator and present the associated asymptotic theory. We also conduct an extensive simulation study to gauge its finite sample performance and discuss how this compares to other related estimators. In Section \ref{sec:QRV*}, we develop a modified version of QRV that remains consistent and asympotically efficient with microstructure noise. Section \ref{sec:QRVempirical} contains an empirical application and Section \ref{sec:conclusion} concludes. Proofs and some additional discussion of alternative formulations of QRV can be found in the appendix.

\section{Quantile-based realised variance measurement} \label{sec:QRVintro}

Let $X = \left( X_{t} \right)_{t \geq 0}$ denote the log-price process, defined on a filtered probability space $\bigl( \Omega, \mathcal{F}, \left( \mathcal{F}_{t} \right)_{t \geq 0}, \mathbb{P} \bigr)$ and adapted to the filtration $\left( \mathcal{F}_{t} \right)_{t \geq 0}$. The theory of financial economics states that if $X$ evolves in a frictionless market, then it has to be of semimartingale form \citep*[see][]{back:91a}. In this paper we start from the assumption that $X$ is a pure Brownian semimartingale, i.e. a continuous sample path process of the form:
\begin{equation}
\label{Eqn:X} X_{t} = X_{0} + \int_{0}^{t} a_{u} \text{d}u + \int_{0}^{t} \sigma_{u} \text{d}W_{u}, \quad t \geq 0,
\end{equation}
where $a = \left( a_{t} \right)_{t \geq 0}$ is a predictable locally bounded drift function, $\sigma = \left( \sigma_{t} \right)_{t \geq 0}$ is an adapted c\`{a}dl\`{a}g volatility process and $W = \left( W_{t} \right)_{t \geq 0}$ a standard Brownian motion.

To prove our CLTs, we will work under some stronger assumptions on $\sigma$.\\[-0.25cm]

\noindent \textbf{Assumption (V)} \emph{$\sigma$ does not vanish (V$_1$) and it satisfies the equation:
\begin{equation*}
\tag{V$_2$} \sigma_{t} = \sigma_{0} + \int_{0}^{t} a_{u}^{ \prime} \text{\upshape{d}}u + \int_{0}^{t} \sigma_{u}^{ \prime} \text{\upshape{d}}W_{u} + \int_{0}^{t} v_{u}^{ \prime}
\text{\upshape{d}}B_{u}^{ \prime}, \quad t \geq 0,
\end{equation*}
where $a^{ \prime} = \left( a_{t}^{ \prime} \right)_{t \geq 0}$, $\sigma^{ \prime} = \left( \sigma_{t}^{ \prime} \right)_{t \geq 0}$ and $v^{ \prime} = \left( v_{t}^{ \prime} \right)_{t \geq 0}$ are adapted c\`{a}dl\`{a}g, $B^{ \prime} = \left( B_{t}^{ \prime} \right)_{t \geq 0}$ is a Brownian motion, and $W \Perp B^{ \prime}$ (here $A \Perp B$ means that A and B are stochastically independent)}.\\[-0.25cm]

\noindent This means that $\sigma$ has its own Brownian semimartingale structure. Note the appearance of $W$ in $\sigma$, which allows for leverage effects. If $X$ is a unique strong solution of a stochastic differential equation then, under some smoothness assumptions on the volatility function $\sigma = \sigma(t, X_{t})$, assumption (V$_{2}$) (with $v_{s}' = 0$ for all $s$) is a consequence of Ito's formula. Thus, assumption (V$_{2}$) is fulfilled for many financial models and, even though it is not a necessary condition, it simplifies the proofs considerably. A more general treatment, including the case where $\sigma$ jumps, can be found in \cite{barndorff-nielsen-graversen-jacod-podolskij-shephard:06a}. We rule out these technical details here, as they are not important to our exposition.

In what follows, we also make use of the concept of stable convergence in law.
\begin{definition}
A sequence of random variables, $\left( Z_{n} \right)_{n \in \mathbb{N}}$, converges stably in law with limit $Z$, defined on an appropriate extension of $\bigl( \Omega, \mathcal{F}, \left( \mathcal{F}_{t} \right)_{t \geq 0}, \mathbb{P} \bigr)$, if and only if for every $\mathcal{F}$-measurable, bounded random variable $Y$ and any bounded, continuous function $g$, the convergence $\lim_{n \to \infty} \E \left[Y g \left(Z_{n} \right) \right] = \E \left[ Y g \left(Z \right) \right]$ holds. We write $Z_{n} \overset{d_{s}}{ \to} Z$, if $\left( Z_{n} \right)_{n \in \mathbb{N}}$ converges stably in law to $Z$.
\end{definition}

\noindent Stable convergence implies weak convergence, or convergence in law, which can be defined as above by taking $Y = 1$, see \cite{renyi:63a} or \cite{aldous-eagleson:78a} for more details about the properties of stably converging sequences. The extension of this concept to stable convergence of processes is discussed in \citet*[][pp. 512--518]{jacod-shiryaev:03a}. In our context, we need the stable convergence to transform the infeasible mixed Gaussian central limit theorems (CLT) proved below into feasible ones that can be implemented in practice.

Central to the theory of semimartingales is the quadratic variation process, defined as:
\begin{equation*}
\left[ X \right]_{t} = \underset{n \to \infty}{\plim} \sum_{i = 1}^{n} \left( X_{t_{i}} - X_{t_{i - 1}} \right)^{2},
\end{equation*}
for any sequence of partitions $0 = t_{0} < t_{1} < \ldots < t_{n} = t$ such that $\sup_i\left\{t_{i} - t_{i - 1} \right\} \to 0$ as $n \to \infty$ \citep*[see, e.g.][pp.
66]{protter:04a}. For the process in Eq. \eqref{Eqn:X}, the quadratic variation is equal to integrated variance (IV hereafter), i.e.
\begin{equation*}
\left[ X \right]_{t} = \int_{0}^{t} \sigma_{u}^{2} \text{d}u.
\end{equation*}
As in \citet*{andersen-bollerslev:98a}, and many of the subsequent papers in this area, here the object of econometric interest is the IV.

\subsection{The estimator and its properties} \label{sec:QRV}

From now on, we will work on the unit time interval without loss of generality, i.e. $t\in[0, 1]$. We assume that $X$ is observed at equidistant points $t_i = i / N$, for $i = 0, \ldots, N$. The increments of $X$ -- the continuously compounded returns -- are denoted as:
\begin{equation}
\Delta_{i}^{N} X = X_{i/N} - X_{(i-1)/N},
\end{equation}
for $i = 1, \ldots, N$. We further assume that $N = nm$, where $m$, $n$ are natural numbers. Specifically, we consider $n$ contiguous subintervals or blocks $[(i - 1)/n, i/n]$, each containing $m$ returns, i.e.
\begin{equation}
D_{i}^{m}X = \left( \Delta_{k}^{N} X \right)_{(i - 1) m + 1 \leq k \leq im},
\end{equation}
for $i = 1, \ldots, n$. The idea of building subintervals is quite natural in the current setting and the mathematical intuition for this approach is discussed in \cite{mykland:10a}. In the asymptotic analysis below, we concentrate on the case where $m$ is fixed and $n \to \infty$, but also briefly comment on the case where $n$ is fixed and $m \to \infty$, and $n, m \to \infty$.

Define the function $g_{k}: \mathbb{R}^{m} \to \mathbb{R}$ such that
\begin{equation}
\label{Eqn:g} g_{k} \left( x \right) = x_{ \left( k \right)},
\end{equation}
where $x_{(k)}$ is the $k$th order statistic of $x = \left( x_{1}, \ldots, x_{m} \right)$. Also define the realised (symmetric) squared $\lambda$-quantile on the $i^{th}$ block as
\begin{equation}
\label{Eqn:q} q_{i}(m,\lambda) = g^{2}_{ \lambda m} \left( \sqrt{N} D_i^m X \right) + g^{2}_{m - \lambda m + 1} \left( \sqrt{N} D_i^m X \right),
\end{equation}
where $\lambda m$ is a natural number. Note that the function $q_i(m,\lambda)$ is even in $X$, so its value does not change if we replace $X$ by $-X$. Also note that $\sqrt{N} D_i^m X$ has been normalized to be $O_{p} (1)$.

We are now in a position to introduce the quantile-based realised variance (QRV, hereafter):
\begin{equation} \label{Eqn:QRVvector}
QRV_N(m,\overline{\lambda},\alpha) \equiv \alpha' QRV_N(m,\overline{\lambda}),
\end{equation}
where $\overline{\lambda} = (\lambda_{1}, \ldots, \lambda_{k})'$ is a $(k\times 1)$ vector of quantiles with $\lambda_j \in(1/2,1)$ for $j=1,\ldots,k$, $\alpha$ is a $(k\times 1)$ non-negative weighting vector with $|\alpha|_1 = 1$, and $QRV_N(m,\overline{\lambda})$ is a $(k\times 1)$ vector with $j$th entry equal to
\begin{equation}
\label{Eqn:QRV} QRV_{N} \left( m,\lambda_{j} \right) = \frac{m}{N} \sum_{i = 1}^{N / m} \frac{q_{i} (m,\lambda_{j})}{\nu_1 \left( m, \lambda_{j} \right)}.
\end{equation}
The scaling factor $\nu$ in Eq. \eqref{Eqn:QRV} is given by:
\begin{equation}
\label{Eqn:nu(r,m)} \nu_{r} \left( m, \lambda \right) = \mathbb{E} \left[ \Bigl( | U_{( \lambda m)} |^{2} + | U_{( m - \lambda m + 1)} |^{2} \Bigr)^{r} \right],
\end{equation}
for $r > 0$, where $U_{(\lambda m)}$ is the $(\lambda m)$-th order statistic of an independent standard normal sample $\{ U_{i} \}_{i = 1}^{m}$.

We show below that QRV constitutes a consistent and highly efficient estimator of the IV under very weak conditions on $X$. Moreover, due to its reliance on return quantiles (with $\lambda < 1$), the QRV is asymptotically robust to finite activity jumps. This holds even when $k = 1$ and the estimator is constructed using only a single pair of quantiles. Still, combining multiple pairs of quantiles as in Eq. \eqref{Eqn:QRVvector} improves the efficiency of the QRV, and we can explicitly characterize the optimal quantile weights that minimize its asymptotic variance.

Two further remarks are in order. First, while it is possible to use asymmetric quantiles in Eq. \eqref{Eqn:q}, this is suboptimal in the current setting due to the symmetry of the normal distribution, and we therefore do not consider this case. Moreover, when $q_i(m, \lambda)$ is based on asymmetric quantiles, the resulting CLT becomes infeasible, because the function $q_i(m,\lambda)$ is not even in $X$ anymore \citep*[see][for further discussion]{kinnebrock-podolskij:08a}. Second, to compute QRV the scaling factor $\nu_{r} \left( m, \lambda \right)$ is needed. These values can be obtained to any desired degree of accuracy by straightforward simulation or numerical integration using the joint density function of the order statistics:
\begin{equation*}
\label{Eqn:BivarPDF} f_{U_{ \left( m - \lambda m + 1 \right)}, U_{ \left( \lambda m \right)}} \left( x, y \right) = 1_{ \{ x < y \}}\frac{m!}{( m - \lambda m )!}\frac{\left( \Phi (y) - \Phi (x) \right)^{2 \lambda m - m - 2} \left( (1-\Phi (y))\Phi(x) \right)^{m - \lambda m}}{( 2 \lambda m - m - 2)!( m - \lambda m )!} \phi (x) \phi (y),
\end{equation*}
for $\lambda \in \left( 1 / 2, 1 \right)$, where $\phi$ and $\Phi$ are the standard normal density and distribution functions.

We now present the main limit results of the QRV.

\begin{theorem}
\label{Thm:QRVconsistency} For the process $X$ in Eq. \eqref{Eqn:X}, and $N = mn$ with $m$ fixed, as $N \to \infty$
\begin{equation*}
QRV_N(m,\overline{\lambda},\alpha) \overset{p}{ \to} IV,
\end{equation*}
where $IV = \int_{0}^{1} \sigma_{u}^{2} \text{\upshape{d}}u$.
\end{theorem}
\begin{proof}
see Appendix \ref{app:proofs}\qedl
\end{proof}

\begin{theorem}
\label{Thm:QRVcentrallimit} For the process $X$ in Eq. \eqref{Eqn:X}, with condition (V) satisfied and $N = mn$ with $m$ fixed, as $N \to \infty$
\begin{equation*}
\sqrt{N} \left( QRV_N(m, \overline{ \lambda}, \alpha) -  IV \right) \overset{d_{s}}{\to} \sqrt{\theta(m, \overline{ \lambda}, \alpha)} \int_{0}^{1} \sigma_{u}^{2} \text{\upshape{d}}W_{u}',
\end{equation*}
where
\begin{equation}
\label{Eqn:efficiency}
\theta(m, \overline{ \lambda}, \alpha) = \alpha' \Theta(m, \overline{ \lambda}) \alpha,
\end{equation}
and the $k \times k$ matrix $\Theta(m, \overline{ \lambda} ) = \left( \Theta(m, \overline{ \lambda} )_{sl} \right)_{1 \leq s, l \leq k}$ is given by
\begin{equation*}
\Theta(m, \overline{\lambda})_{ij} = m \frac{ \nu_1(m, \lambda_{i}, \lambda_{j} ) - \nu_1 (m, \lambda_{i} ) \nu_1(m, \lambda_{j} )}{ \nu_1(m,\lambda_{i}) \nu_1 (m, \lambda_{j} )},
\end{equation*}
with
\begin{equation} \label{Eqn:nuij}
\nu_1 (m,\lambda_{i}, \lambda_{j}) = \E [ ( | U_{ \left( m\lambda_i \right)} |^{2} + | U_{ \left( m - m\lambda_i + 1 \right)} |^{2} ) ( | U_{ \left( m\lambda_j \right)} |^{2} + | U_{ \left( m - m\lambda_j + 1 \right)} |^{2} ) ],
\end{equation}
and where $W'$ is another Brownian motion defined on an extension of $\bigl( \Omega, \mathcal{F}, \left( \mathcal{F}_{t} \right)_{t \geq 0}, \mathbb{P} \bigr)$ with $W' \Perp \mathcal{F}$. Because $\sigma$ is independent of $W'$, this implies stable convergence to a mixed normal distribution:
\begin{equation*}
\sqrt{N} \left( QRV_N(m, \overline{ \lambda}, \alpha) -  IV \right) \overset{d_{s}}{ \to} MN \left( 0, \theta(m, \overline{ \lambda}, \alpha) IQ \right),
\end{equation*}
where $IQ = \int_{0}^{1} \sigma_{u}^{4} \text{\upshape{d}}u$ is the integrated quarticity.
\end{theorem}
\begin{proof}
see Appendix \ref{app:proofs} \qedl
\end{proof}

\noindent Theorem \ref{Thm:QRVconsistency} and \ref{Thm:QRVcentrallimit} show that QRV is a consistent estimator of the IV under very weak conditions on the process $X$ and that it converges at rate $N^{-1 / 2}$, the best attainable in this setting. For a given block size $m$ and quantiles $\overline{\lambda}$, the weighting vector $\alpha$ can be chosen optimally to minimize the asymptotic variance of QRV, i.e.
\begin{equation}
\label{Eqn:alpha*m} \alpha^{ \ast} = \frac{ \Theta^{-1} \left( m, \overline{ \lambda} \right) \iota}{ \iota^{ \prime} \Theta^{-1} \left( m, \overline{ \lambda} \right) \iota},
\end{equation}
where $\iota$ is a $(k\times 1)$ vector of ones. The asymptotic efficiency of QRV is characterized by the constant $\theta$ in Eq. \eqref{Eqn:efficiency}, which is equal to $\theta(m, \overline{ \lambda}, \alpha^\ast) = (\iota' \Theta^{-1}(m,\overline{\lambda}) \iota)^{-1}$ when optimal weights are used. In Section \ref{sec:QRVimplementation}, we shall see that $\theta$ takes on values between 3 and 4 when using a single (pair of) quantile around $0.9$ and rapidly approaches the parametric lower bound of 2 as multiple quantiles are used. Thus, QRV can attain the efficiency of the maximum likelihood estimator in the parametric no-jump version of this problem, while still retaining robustness to jumps, when they do appear in the price process. In fact, the consistency and CLT of the QRV are unaffected by the presence of finite activity jumps, as the following proposition highlights.

\begin{proposition} \label{prop:QRVjumprobust}
Theorem \ref{Thm:QRVconsistency} and \ref{Thm:QRVcentrallimit} remain valid for an extension of the process $X$ in Eq. \eqref{Eqn:X} that incorporates finite activity jumps, i.e.
\begin{equation}
\label{Eqn:Xjump} X_{t} = X_{0} + \int_{0}^{t} a_{u} \text{\upshape{d}}u + \int_{0}^{t} \sigma_{u} \text{\upshape{d}}W_{u} + \sum_{i = 1}^{q(t)} J_{i}, \quad t \geq 0,
\end{equation}
where $q = \left( q(t) \right)_{t \geq 0}$ is a finite activity counting process and $J = \left( J_{i} \right)_{i = 1}^{q(t)}$ are non-zero random variables representing the jumps in $X$.
\end{proposition}
The intuition for this result is clear: with $m$ fixed and $N \rightarrow \infty$, the number of observations in each block remains constant but the time interval it spans shrinks so that, in the limit, it contains at most one jump with probability one. Combined with the restriction $\lambda < 1$, which ensures we leave out at least the largest negative and positive return when constructing the QRV, this naturally implies that the estimator is asymptotically immune to finite activity jumps.

To conclude this section, we point out that our quantile-based approach can also be used to estimate the integrated quarticity in a similar fashion. In particular, define
\begin{equation}
\label{Eqn:QIQvector}
QRQ_N(m, \overline{\lambda}, \alpha) \equiv \alpha' QRQ_N(m, \overline{\lambda}),
\end{equation}
where $\overline{\lambda} = (\lambda_{1}, \ldots, \lambda_{k})'$ is a $(k\times 1)$ vector of quantiles with $\lambda_j \in(1/2,1)$ for $j=1,\ldots,k$, $\alpha$ is a $(k\times 1)$ non-negative weighting vector with $|\alpha|_1 = 1$, and $QRQ_N(m,\overline{\lambda})$ is a $(k\times 1)$ vector with $j$th entry equal to:
\begin{equation}
\label{Eqn:QIQ} QRQ_N\left(m, \lambda_{j} \right) = \frac{1}{\nu^{iq}(m,\lambda_{j})} \frac{m}{N} \sum_{i = 1}^{N/m} \left( g^{4}_{ \lambda_{j} m} \left( \sqrt{N} D_i^m X \right) + g^{4}_{m -\lambda_{j} m + 1} \left( \sqrt{N} D_i^m X \right) \right),
\end{equation}
with
\begin{equation*}
\nu^{iq} \left( m, \lambda \right) = \mathbb{E} \left[ | U_{( \lambda m)} |^{4} + | U_{( m - \lambda m + 1)} |^{4}  \right].
\end{equation*}
As in Theorem \ref{Thm:QRVconsistency} we have consistency for $m$ fixed and $N \to \infty$, i.e. $QRQ_N(m, \overline{\lambda}, \alpha) \overset{p}{ \to} IQ $. This estimator can now be used to formulate a feasible CLT for QRV as follows:
\begin{equation*}
\sqrt{N}\frac{ QRV_N(m, \overline{ \lambda}, \alpha_{iv}) - IV}{\sqrt{\theta(m, \overline{\lambda}, \alpha_{iv}) QRQ_N(m, \overline{\lambda}, \alpha_{iq}) }} \overset{d}{ \to} N ( 0, 1).
\end{equation*}

\subsection{A subsampling implementation of QRV}
The QRV estimator developed above is constructed using empirical return quantiles computed over non-overlapping contiguous intervals of data. In this section, we discuss a subsampled version of the QRV, which is still more efficient than its blocked counterpart. In particular, we define:
\begin{equation}
\label{Eqn:QRVsubvector}
QRV_N^{sub}(m,\overline{\lambda},\alpha) \equiv \alpha' QRV_N^{sub}(m,\overline{\lambda}),
\end{equation}
where $\alpha$ and $\overline{\lambda}$ are as above, and $QRV_{N}^{sub}(m,\overline{\lambda})$ is a $(k\times 1)$ vector with $j$th entry equal to:
\begin{equation} \label{Eqn:QRVsub}
QRV_{N}^{sub}(m,\lambda_{j}) = \frac{1}{N-m+1} \sum_{i = 1}^{N-m+1} \frac{q_{i}^{sub} (m,\lambda_{j})}{\nu_1(m, \lambda_{j})},
\end{equation}
and
\begin{equation*}
q_{i}^{sub}(m,\lambda) = g^{2}_{ \lambda m} \left( \sqrt{N} D_{i,m} X \right) + g^{2}_{m - \lambda m + 1} \left( \sqrt{N} D_{i,m} X \right),
\end{equation*}
where $D_{i,m}X = \left( \Delta_{k}^{N} X \right)_{i \leq k \leq i + m-1}$ for $i\geq 1$.

Because the estimator in Eq. \eqref{Eqn:QRVsubvector} is basically a subsampled version of the blocked QRV in Eq. \eqref{Eqn:QRVvector}, it is also consistent for the IV by Theorem \ref{Thm:QRVconsistency}. Moreover, by a triangular inequality argument it is clear that subsampling improves the asymptotic efficiency.

\begin{theorem}
\label{Thm:QRVsubcentrallimit} For the process $X$ in Eq. \eqref{Eqn:X} with condition (V) satisfied and $N = mn$ with $m$ fixed, as $N\to\infty$
\begin{equation} \label{Eqn:t3subCLT}
\sqrt{N} \left( QRV_N^{sub}(m,\overline{\lambda},\alpha) -  IV\right) \overset{d_{s}}{\to} MN \left( 0, \theta^{sub}(m,\overline{\lambda},\alpha) IQ \right),
\end{equation}
where
\begin{equation}
\label{Eqn:subSamVar}
\theta^{sub}(m,\overline{\lambda},\alpha) = \alpha'\Theta^{sub}(m,\overline{\lambda})\alpha,
\end{equation}
and the $k \times k$ matrix $\Theta^{sub}(m, \overline{ \lambda}) = \left( \Theta^{sub}(m, \overline{ \lambda})_{sl} \right)_{1 \leq s, l \leq k}$ is given by
\begin{equation*}
\Theta^{sub}(m,\overline{\lambda})_{ij} = \frac{1}{m} \Theta(m,\overline{\lambda})_{ij} +\frac{2}{ \nu_1(m,\lambda_{i}) \nu_1 (m, \lambda_{j} )} \sum_{k=1}^{m} \mbox{cov}\left(
 | U_{ \left( m\lambda_i \right)}^{(0)} |^{2} + | U_{ \left( m - m\lambda_i + 1 \right)}^{(0)} |^{2} ,  | U_{ \left( m\lambda_j \right)}^{(k)} |^{2} + | U_{ \left( m - m \lambda_j + 1 \right)}^{(k)} |^{2} \right),
\end{equation*}
where $U^{(0)}= \{ U_{i} \}_{i = 1}^{m}$, $U^{(k)}= \{ U_{i} \}_{i = 1+k}^{m+k}$ and $\{ U_{i} \}_{i = 1}^{m+k}$ is an independent standard normal sample. Furthermore, the convergence in Eq. \eqref{Eqn:t3subCLT} is robust to finite activity jumps, i.e. it also holds for the processes defined in Eq. \eqref{Eqn:Xjump}.
\end{theorem}
\begin{proof}
see Appendix \ref{app:proofs} \qedl
\end{proof}
As before, asymptotically optimal weights can be assigned to the quantiles so as to minimize $\theta^{sub}(m,\overline{\lambda},\alpha)$. In Section \ref{sec:QRVimplementation}, we illustrate the efficiency improvement that results from subsampling.

\subsection{QRV with $m \to \infty$}

Up to this point, we have considered the case where $m$ is fixed and $n \to \infty$. In this limit, QRV is consistent under very weak assumptions on the log-price $X$. We now briefly discuss the case where $m \to \infty$. To get the corresponding asymptotic results in this limit, much stronger assumptions need to be imposed on $X$. In particular, when $n$ is fixed a sufficient condition for Theorem \ref{Thm:QRVconsistency} and \ref{Thm:QRVcentrallimit} to hold is that $\sigma$ is constant \citep*[this follows directly from classical order statistic results, see for instance][]{david:70a}. In this case, the asymptotic constants are as given in the following proposition.
\begin{proposition}
\label{prop:QRVmtoinfty} We have
\begin{eqnarray*}
\nu_1(\lambda) &\equiv& \lim_{m\rightarrow\infty}\nu_1(m,\lambda) = 2c_{ \lambda}^{2},\\
\nu_1(\lambda_i,\lambda_j)&\equiv&\lim_{m\rightarrow\infty}\nu_1(m,\lambda_i,\lambda_j) = 4c_{\lambda_i}^{2}c_{\lambda_j}^{2},\\
\Theta(\overline{\lambda})_{ij} &\equiv&\lim_{m\rightarrow\infty}\Theta(m,\overline{\lambda})_{ij} = \lim_{m\rightarrow\infty}\Theta^{sub}(m,\overline{\lambda})_{ij}
= 2\frac{ (1 - \lambda_j ) ( 2 \lambda_i - 1)}{ \phi \bigl( c_{ \lambda_i} \bigr) \phi \bigl( c_{
\lambda_j} \bigr)  c_{ \lambda_i} c_{ \lambda_j} },\\
\theta( \overline{ \lambda}, \alpha) &\equiv& \lim_{m\rightarrow\infty}\theta(m,\overline{\lambda},\alpha) = \lim_{m\rightarrow\infty}\theta^{sub}(m,\overline{\lambda},\alpha) = \alpha' \Theta(\overline{\lambda}) \alpha,
\end{eqnarray*}
with $\lambda_{i} \leq \lambda_{j}$, where $c_{ \alpha}$ and $\phi$ denote the $\alpha$-quantile and density function of the standard normal distribution.
\end{proposition}
\begin{proof}
see Appendix \ref{app:proofs}  \qedl
\end{proof}
These results are interesting for a number of reasons. Firstly, when $m \to \infty$ the asymptotic covariance matrix $\Theta$ can be expressed in closed form, and it is identical for the blocking and subsampling version of the QRV. This, in turn, allows for fast and easy calculation of optimal quantile weights $\alpha^\ast$. In Section \ref{sec:QRVimplementation}, we demonstrate that the use of these weights leads to only a very limited efficiency loss -- even for small $m$ -- which makes it attractive from a practical point of view. Secondly, in certain applications the constant volatility assumption can sometimes be justified (e.g. when sampling over a short horizon or on a suitably deformed time scale), and a single block implementation of QRV (i.e. $n = 1$ and $m$ large) may be preferred purely for the sake of computational simplicity. Interestingly, we will see in Section \ref{sec:QRVimplementation} that the efficiency loss associated with using only a single block is small, particularly when combining multiple quantiles. Finally, the above limit allows us to clarify the relation between the QRV and some related estimators that have appeared in the literature. In particular, with $m \to \infty$, $n = 1$, and constant volatility, QRV corresponds to the estimator of \citet*{david:70a}, the only difference being that the latter estimates the standard deviation instead of the variance. Also, in Appendix \ref{sec:QRVabsolute} we discuss an alternative formulation of QRV based on absolute returns. This estimator nests the MinRV and MedRV of \cite{andersen-dobrev-schaumburg:08a}, and we use the $m \to \infty$ limit to show its equivalence to the QRV based on signed returns introduced above.

As an aside, if we want to relax the constant volatility assumption above, it is still possible to let $m \rightarrow \infty$, but then we also need $n \to \infty$. Moreover, stronger assumptions on the dynamics of $X$ are needed: for consistency we require condition (V) and $m/n \rightarrow 0$, whereas to obtain a CLT we additionally need $a$, $\sigma^{ \prime}$ and $v^{ \prime}$ to satisfy condition (V) and $m^3 / n \rightarrow 0$. We then have the following result:
\begin{equation*}
QRV_N (m, \overline{ \lambda}, \alpha) \overset{p}{ \to} IV,
\end{equation*}
and
\begin{equation*}
\sqrt{N} \left( QRV_N(m,\overline{\lambda},\alpha) -  IV\right) \overset{d_{s}}{\to} MN \left( 0, \theta(\overline{\lambda},\alpha) IQ \right).
\end{equation*}
A formal proof of these results can be found in a separate appendix, i.e. see \citet*{christensen-oomen-podolskij:10b}.

\subsection{Implementation and asymptotic efficiency of QRV} \label{sec:QRVimplementation}

\begin{table}[!ht]
\setlength{\tabcolsep}{0.45cm}
\begin{center}
\caption{Asymptotic efficiency of QRV with single and multiple quantiles}
\label{Table:QRVefficiency}
\begin{tabular}{lccclccclc}
\hline
            & \multicolumn{3}{c}{blocked QRV}     && \multicolumn{3}{c}{subsampled QRV} \\
\cline{2-4}\cline{6-8}
$\overline{\lambda}~~~~~~~~\backslash ~~~~~~~~m$   & $20$ & $40$ & $100$  && $20$ & $40$ & $100$ && $\infty$\\
\hline
\multicolumn{9}{l}{\emph{Panel A: single quantile}}\\
0.80        & 4.24 &   4.29 &   4.31 &&   3.54 &   3.73 &   3.92 &&   4.32\\
0.85        & 3.56 &   3.58 &   3.59 &&   3.02 &   3.14 &   3.27 &&   3.60\\
0.90        & 3.10 &   3.14 &   3.15 &&   2.67 &   2.75 &   2.86 &&   3.16\\
0.95        & 2.88 &   2.99 &   3.07 &&   2.52 &   2.62 &   2.75 &&   3.13\\
0.98        &   -- &     -- &   3.58 &&     -- &     -- &   3.16 &&   3.88\\
\multicolumn{9}{l}{}\\
\multicolumn{9}{l}{\emph{Panel B: multiple quantiles with optimal weights $\alpha^\ast_{(m)}$}}\\
0.80 -- 0.95& 2.40 &   2.41 &   2.42 &&   2.27 &   2.29 &   2.32 &&   2.42\\
0.80 -- 0.98&   -- &     -- &   2.19 &&     -- &     -- &   2.13 &&   2.19\\
\multicolumn{9}{l}{}\\
\multicolumn{9}{l}{\emph{Panel C: multiple quantiles with asymptotically optimal weights $\alpha^\ast_{(\infty)}$}}\\
0.80 -- 0.95& 2.41 &   2.41 &   2.42 &&   2.31 &   2.32 &   2.33 &&   2.42\\
0.80 -- 0.98&   -- &     -- &   2.19 &&     -- &     -- &   2.14 &&   2.19\\
\hline
\end{tabular}
\medskip
\begin{footnotesize}
\parbox{0.98\textwidth}{\emph{Note}. This table reports the asymptotic efficiency constants $\theta(m, \overline{ \lambda}, \alpha)$ for the blocked QRV and $\theta^{sub}(m, \overline{ \lambda}, \alpha)$ for the subsampled QRV as given in Eqs. \eqref{Eqn:efficiency} and \eqref{Eqn:subSamVar} for different values of $m$. The last column reports the limiting value $\theta(\overline{ \lambda}, \alpha)$ as in Proposition \ref{prop:QRVmtoinfty} (which is the same for the blocked and subsampled QRV). Panel A reports the results for a selection of single quantiles. Panels B and C combine these quantiles using exact ``finite $m$'' optimal weights following Theorems \ref{Thm:QRVcentrallimit} and \ref{Thm:QRVsubcentrallimit} and asymptotically optimal weights following Proposition \ref{prop:QRVmtoinfty} respectively.}
\end{footnotesize}
\end{center}
\end{table}

Implementing the QRV requires selection of quantiles $\overline{ \lambda}$, quantile weights $\alpha$, and block-length $m$. The asymptotic theory developed above can be used to determine the optimal quantile weights but, as we will now discuss, it can also be used to guide the choice of $\overline{ \lambda}$ and $m$. We will also use the theory to compare the efficiency of QRV to the leading alternative realised variance measures. Section \ref{sec:QRVsimulations} will follow up this discussion with simulations investigating QRV's finite sample performance and robustness to jumps and outliers.

Table \ref{Table:QRVefficiency} reports the asymptotic efficiency constant $\theta$ of the blocked and subsampled QRV estimator for different $m$ and $\overline{\lambda}$. A number of important observations regarding the efficiency and the preferred implementation of the estimator can be made.

First, considering Panels A and B, we see that QRV is a highly efficient estimator particularly when using multiple quantiles. For instance, when combining the five quantiles listed in Panel A, QRV has an asymptotic efficiency of around 2.2. This compares favorably to the leading jump-robust bi-power variation measure of \cite{barndorff-nielsen-shephard:04b} for which the corresponding figure is $\pi^2/4+\pi-3 \simeq 2.61$. Moreover, by including additional quantiles in the construction of QRV we can push its efficiency arbitrarily close to 2, so that in the limit it attains the ML efficiency of realised variance \citep*{jacod:94a, jacod-protter:98a, barndorff-nielsen-shephard:02a} while still retaining robustness to jumps. Comparing the blocked QRV with the subsampled QRV we confirm an efficiency gain associated with subsampling, albeit the benefit is modest particularly when using multiple quantiles.\footnote{The efficiency gain associated with subsampling can be very substantial when we consider higher powers of quantiles, for instance when estimating integrated quarticity. Also with microstructure noise, the efficiency gain is of a factor 2 -- 3. Because we start Section \ref{sec:QRV*} immediately with a subsampled version of our noise robust-estimator this is not explicitly discussed.}

\begin{figure}[t!]
\begin{center}
\caption{Optimal quantile weights and scaling factors for varying block size $m$.}
\label{Figure:QRVasymptotic}
\begin{tabular}{cc}
Panel A: optimal weights $\alpha$ &
Panel B: scaling factors $\nu_1$ \\
\includegraphics[height=8cm,width=0.45\textwidth]{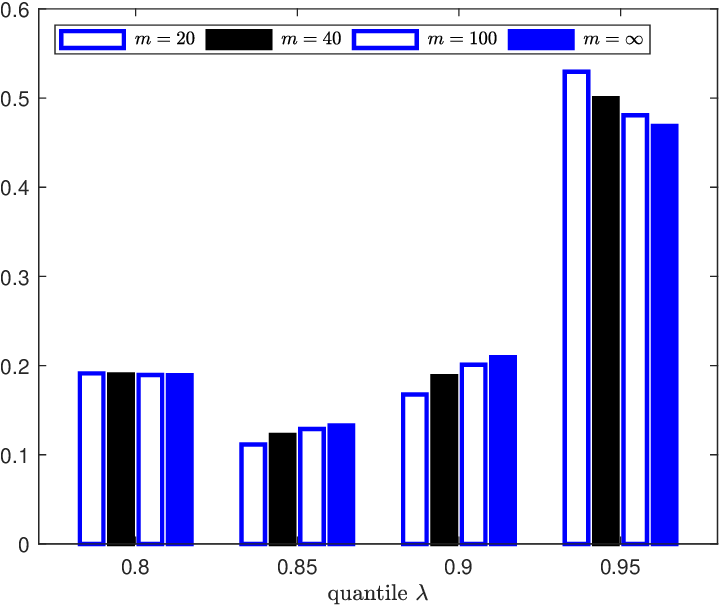} &
\includegraphics[height=8cm,width=0.45\textwidth]{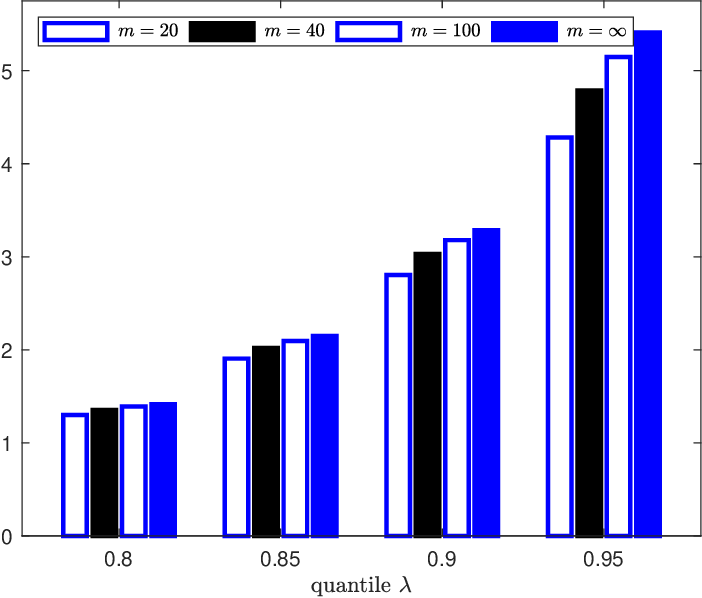} \\
\end{tabular}
\begin{footnotesize}
\parbox{0.98\textwidth}{\emph{Note}. This figure reports the optimal quantile weights $\alpha$ (Panel A) as in Eq. \eqref{Eqn:alpha*m} and scaling factors $\nu_1$ (Panel B) as in Eq. \eqref{Eqn:nu(r,m)} for QRV using blocking with $\overline{ \lambda} =\{0.80,0.85,0.90,0.95\}$ and varying $m$.}
\end{footnotesize}
\end{center}
\end{figure}

Second, comparing Panels B and C, we see that the use of limiting ``$m \to \infty$'' optimal quantile weights following Proposition \ref{prop:QRVmtoinfty} instead of exact ``finite-$m$'' optimal weights following Theorem \ref{Thm:QRVcentrallimit} and \ref{Thm:QRVsubcentrallimit} leads to only a marginal deterioration in efficiency. This is very attractive from a practical viewpoint, for it means that we can simply use the limiting closed-form expression for $\Theta$ to obtain reliable and near-optimal quantile weights instead of calculating its exact finite-sample counterpart for each and every $m$ and $\overline{\lambda}$ we may consider. As an aside, we point out that the use of asymptotic scaling factors $\nu_1$ should be avoided as it will induce potentially serious biases in the estimator even for moderately large $m$. Figure \ref{Figure:QRVasymptotic} further illustrates these effects by plotting the optimal quantile weights and scaling factors for different values of $m$.

Third, from Panel A we observe that the choice of quantile can be informed by efficiency considerations. It is quite intuitive that quantiles close to the mode of the return distribution are not very informative about the variance of the process. At the same time, quantiles far into the tail region tend to be erratic. The optimal choice of quantile balances this tradeoff to extract the maximum amount of information regarding the spread of the distribution. From Panel A we see that, depending on the choice of $m$, the optimal quantile lies in the region $0.90 - 0.95$. Of course, with $k>1$, quantiles outside this region may be added to exploit the covariance structure of the order statistics. On the choice of block length, we see that there are modest efficiency gains to be had by choosing it small (it can be shown that $\theta$ is monotonically increasing in $m$). However, when multiple quantiles are used, the gain is negligible.

From the above discussion it is clear that the asymptotic theory is very helpful in guiding the choice of quantiles and block length. In particular, we can determine the optimal (mix of) quantiles and their weights by maximizing asymptotic efficiency of the resulting estimator. In addition, the theory indicates a weak preference for choosing a small $m$ and the use of subsampling. In practice, there are two other important considerations in addition to asymptotic efficiency, namely (i) QRV's ability to estimate the IV in finite sample and (ii) QRV's robustness to jumps and outliers. For QRV to accurately measure the IV, $\sigma$ is required to be locally constant on each subinterval. Thus, also from this viewpoint, a small $m$ or equivalently a large $n$ is desirable. On the other hand, jump robustness is controlled by the \emph{joint} choice of $m$ and $\lambda_{max} = \max \{ \overline{ \lambda} \}$: over a block of length $m$, QRV is robust to $(1-\lambda_{max})m-1$ positive and negative jumps. These observations suggest that the optimal choice of $\overline{\lambda}$ and $m$ should satisfy the following three conditions (i) $m$ is sufficiently small to ensure good ``locality'' of the estimator, (ii) $\lambda$ includes a quantile in the highly informative tail region, and (iii) $(1-\lambda_{max})m$ is sufficiently large to ensure robustness against jumps and outliers. Sections \ref{sec:QRVsimulations} and \ref{sec:QRVempirical} will further illustrate the above using simulations and empirical analysis.

\subsection{Finite sample performance and jump robustness}\label{sec:QRVsimulations}

The results so far indicate that the asymptotic efficiency of QRV is excellent. The simulations below are designed to gauge the finite sample performance of the estimator. We pay particular attention to bias, efficiency, and robustness to jumps and outliers. We also compare the performance of the QRV to recently developed alternative estimators.

To simulate the log-price $X$, we adopt the following model:
\begin{equation}
\label{Eqn:Xsim} \text{d}X_{t} = \sigma_{t} \text{d} W_t, \qquad t \in \left[ 0, 1 \right],
\end{equation}
where $W$ is a standard Brownian motion and the dynamics of $\sigma_{t}$ are as specified below. The baseline scenario is a constant volatility Brownian motion (``BM''), i.e.
\begin{equation}
\sigma^2_t  = 0.0391.\label{Eqn:BM}
\end{equation}
To assess QRV's ability to handle time-varying volatility, we use a Heston-type stochastic volatility (``SV'') model
\begin{equation}
\text{d} \sigma^2_t = (0.3141 - 8.0369\sigma^2_t) \text{d}t+\sigma_t\sqrt{0.1827}
\text{d}B_t,\label{Eqn:Heston}
\end{equation}
where $B$ is another Brownian motion with $B \Perp W$. To gauge the impact of leverage, we also simulate from Eq. \eqref{Eqn:Heston} with $\text{d}W_t \text{d}B_t = -0.75 \text{d}t$ (``SV-LEV''). Finally, we consider two more variance specifications that are both capable of generating erratic and highly volatile sample paths. The first is a model proposed by \cite{ait-sahalia:96a} that incorporates stochastic elasticity of variance and non-linear drift (``SEV-ND''), i.e.
\begin{eqnarray}
\text{d} \sigma^2_t & = & (-0.554+21.32\sigma^2_t-209.3\sigma^4_t+0.005\sigma^{-2}_t) \text{d}t +
\sqrt{0.017\sigma^2_t+53.97\sigma_t^{5.76}} \text{d}B_t.\label{Eqn:SEVND}
\end{eqnarray}
with $B \Perp W$. The second is a two-factor stochastic volatility model (``SV2F-LEV'') analyzed in \citet*{chernov-gallant-ghysels-tauchen:03a}, i.e.
\begin{eqnarray}
\sigma_t^2 &= &\text{s-}\exp(-1.2+0.04f^{(1)}_t+1.5f^{(2)}_t),\label{Eqn:SV2F}\\
\text{d}f^{(1)}_t & = &-0.000137f^{(1)}_t \text{d}t+\text{d}B_t^{(1)},\notag\\
\text{d}f^{(2)}_t & = &-1.386f^{(2)}_t \text{d}t+(1+0.25f^{(2)}_t)\text{d}B_t^{(2)},\notag
\end{eqnarray}
where $\text{d}W_t \text{d}B_t^{(1)} = \text{d}W_t \text{d}B_t^{(2)} = -0.3 \text{d}t$ and $\text{s-}\exp$ denotes a ``spliced'' exponential function as specified and discussed in \citet*{chernov-gallant-ghysels-tauchen:03a}.

The above stochastic volatility models cover a wide range of dynamic specifications and thus provide a good testing ground for QRV. The parameter values for the BM, SV, and SEV-ND models in Eqs. (\ref{Eqn:BM}--\ref{Eqn:SEVND}) are taken from the empirical study by \cite{bakshi-ju-ou-yang:06a} whereas the parameters for the SV2F-LEV model in Eq. \eqref{Eqn:SV2F} are taken from \citet*{huang-tauchen:05a}. It should be noted that while these studies typically calibrate the models from daily data to an annual horizon, here we simulate the processes over the unit interval so that, effectively, we compress a year's worth of variation into a single day. As a result, we end up simulating highly erratic volatility paths, which serves to challenge the QRV estimator to the extreme. As an illustration, see Panel A of Figure \ref{Figure:QRVrobust} for a simulated return and variance series.

To study the robustness of QRV, we also simulate from the BM model and add jumps. In particular, we add a fixed number of $n_J$ Gaussian jumps at random points in the sample with a combined jump variation $v_J$ measured as a fraction of the IV. We consider four scenarios: $\{n_J,v_J\} = \{1,\frac{1}{4}\}$, i.e. one large jump accounting for $20\%$ of total variation, $\{n_J,v_J\} = \{ 5, \frac{1}{4} \}$, i.e. five medium jumps accounting for $20\%$ of total variation, $\{n_J,v_J\} = \{10, \frac{1}{4} \}$, i.e. ten small jumps accounting for $20\%$ of total variation, and $\{ n_J, v_J \} = \{5, \frac{1}{2}\}$, i.e. five large jumps accounting for a third of total variation. Additionally, we consider a scenario where the price series is contaminated by ``outliers''. Such spurious and deviant price observations are commonly encountered in high-frequency data due to, for instance, delayed trade reporting, misplaced decimal points, data errors, etc. (see Section \ref{sec:QRVempirical} for some examples). In our simulations, we position a single outlier of random size at a random point in the series, ensuring that its variation accounts for $20\%$ of total variation. As an illustration, see Panel B of Figure \ref{Figure:QRVrobust} for a simulated price series with a jump and outlier added.

\begin{figure}[t!]
\begin{center}
\caption{SV and jump simulation.}
\label{Figure:QRVrobust}
\begin{tabular}{cc}
\footnotesize{Panel A: squared return series from SEV-ND} &
\footnotesize{Panel B: price series from BMJ \& BM-outlier} \\
\includegraphics[height=8cm,width=0.45\textwidth]{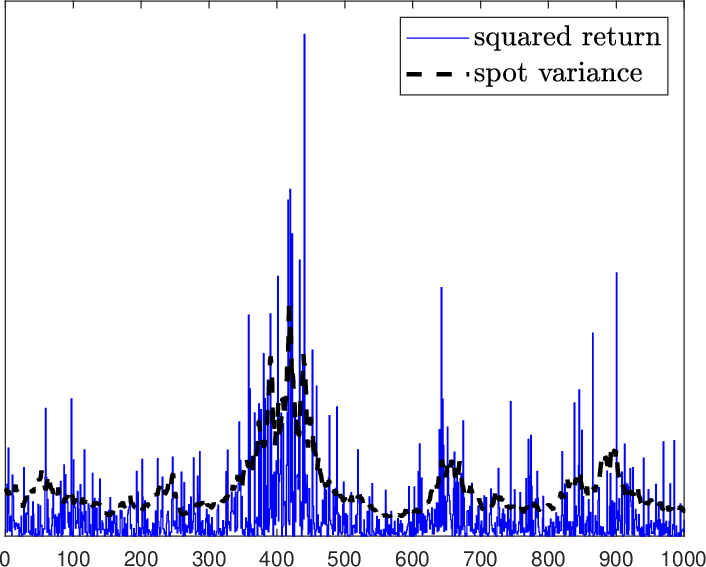} &
\includegraphics[height=8cm,width=0.45\textwidth]{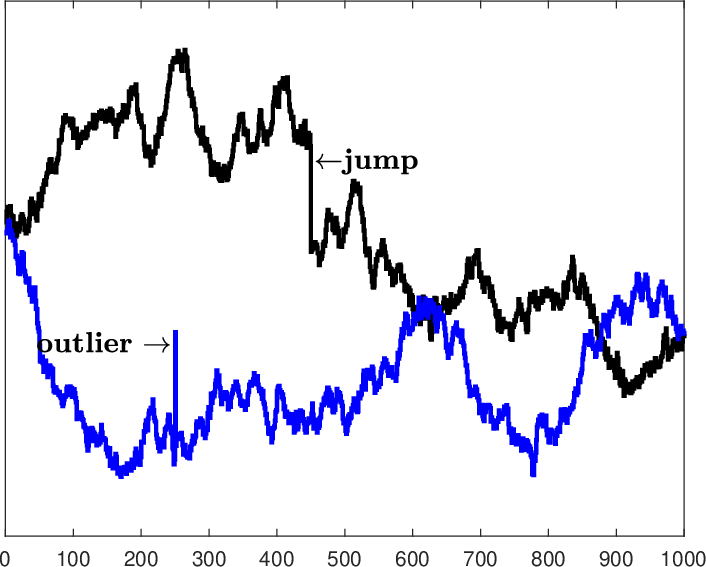}
\end{tabular}
\begin{footnotesize}
\parbox{0.98\textwidth}{\emph{Note}. Panel A plots a time-series of $N = 1,000$ squared returns simulated from the SEV-ND model over the unit interval. Panel B plots a time-series of $N=1,000$ prices simulated from the BM+jump and BM+outlier model over the unit time interval.}
\end{footnotesize}
\end{center}
\end{figure}

To simulate the process in Eq. \eqref{Eqn:Xsim}, we use an Euler discretization scheme and set $N = 1,000$. QRV is computed as in Eq. \eqref{Eqn:QRVvector} and \eqref{Eqn:QRVsubvector} using four pairs of quantiles $\overline{\lambda} = \{0.80,0.85;0.90;0.95\}$, asymptotic weights derived from Proposition \ref{prop:QRVmtoinfty}, and three different choices of block length, namely $m = \{20,40,100\}$ or equivalently $n = \{50,25,10\}$. For comparison, we also compute realised variance (RV) in addition to three recently proposed jump-robust estimators, i.e. bi-power variation (BPV) of \citet*{barndorff-nielsen-shephard:04b}, threshold realised variance (TRV) of \citet*{jacod:08a, mancini:04b, mancini:09a}, and MedRV of \citet*{andersen-dobrev-schaumburg:08a}:
\begin{eqnarray*}
RV_N &=& \sum_{i =1}^N |\Delta_i^N X|^2,\notag\\
BPV_N &=& \frac{\pi}{2} \sum_{i = 2}^N |\Delta_i^N X||\Delta_{i-1}^N X|,\notag\\
TRV_N &=& \sum_{i =1}^N |\Delta_i^N X|^21_{\{|\Delta_i^N X|<cN^{-\overline{\omega}}\}},\qquad \text{for}\quad\overline{\omega}\in (0,1/2),\notag\\
MedRV_N &=& \frac{\pi}{6-4\sqrt{3}+\pi}\frac{N}{N-2}\sum_{i = 2}^{N-1} \textrm{median} (|\Delta_{i-1}^N X|,|\Delta_i^N X|,|\Delta_{i+1}^N X|)^2.
\end{eqnarray*}

\begin{sidewaystable}[!ht]
\setlength{\tabcolsep}{0.25cm}
\begin{center}
\caption{Performance of QRV with stochastic volatility, jumps, and outliers}\smallskip
\label{Table:QRVrobust}
\begin{tabular}{lrrrcrrrcrrrr}
\hline & \multicolumn{3}{c}{blocked QRV $(m,n)$} && \multicolumn{3}{c}{subsampled QRV $(m,n)$} && \multicolumn{4}{c}{benchmarks} \\
\cline{2-4} \cline{6-8} \cline{10-13}
model & $(20,50)$ & $(40,25)$ & $(100,10)$ & & $(20,50)$ & $(40,25)$ & $(100,10)$ & & RV & BPV& TRV & MedRV\\
\hline
\multicolumn{6}{l}{\emph{Panel A: ``bias'' measure $\mathbb{E} (\widehat{IV}/IV)$}}\\
BM                             &     1.00 &     1.00 &     1.00 & &     1.00 &     1.00 &     1.00 & &     1.00 &     1.00 &     1.00 &     1.00 \\
SV                             &     1.00 &     0.99 &     0.98 & &     1.00 &     0.99 &     0.98 & &     1.00 &     1.00 &     1.00 &     1.00 \\
SV-LEV                         &     1.00 &     0.99 &     0.98 & &     1.00 &     0.99 &     0.98 & &     1.00 &     1.00 &     1.00 &     1.00 \\
SEV-ND                         &     1.00 &     0.99 &     0.99 & &     1.00 &     0.99 &     0.98 & &     1.00 &     1.00 &     1.00 &     1.00 \\
SV2F-LEV                       &     1.00 &     1.00 &     0.99 & &     1.00 &     0.99 &     0.98 & &     1.00 &     1.00 &     1.00 &     1.00 \\
BMJ$(n_J=1,v_J=\frac{1}{4})$   &     1.00 &     1.00 &     1.00 & &     1.00 &     1.00 &     1.00 & &     1.25 &     1.03 &     1.01 &     1.00 \\
BMJ$(n_J=5,v_J=\frac{1}{4})$   &     1.02 &     1.02 &     1.02 & &     1.02 &     1.02 &     1.02 & &     1.25 &     1.06 &     1.05 &     1.02 \\
BMJ$(n_J=10,v_J=\frac{1}{4})$  &     1.04 &     1.04 &     1.03 & &     1.04 &     1.04 &     1.03 & &     1.25 &     1.08 &     1.12 &     1.03 \\
BMJ$(n_J=5,v_J=\frac{1}{2})$   &     1.03 &     1.02 &     1.02 & &     1.03 &     1.02 &     1.02 & &     1.50 &     1.09 &     1.05 &     1.02 \\
BM-outlier                     &     1.01 &     1.01 &     1.01 & &     1.01 &     1.01 &     1.01 & &     1.25 &     1.21 &     1.02 &     1.33 \\
\multicolumn{6}{l}{}\\
\multicolumn{6}{l}{\emph{Panel B: ``efficiency'' measure \upshape{var}$(\sqrt{N}(\widehat{IV} - IV)/\sqrt{IQ})$}}\\
BM                             &     2.41 &     2.42 &     2.42 & &     2.33 &     2.38 &     2.49 & &     2.00 &     2.60 &     2.00 &     2.96 \\
SV                             &     2.42 &     2.41 &     2.37 & &     2.37 &     2.51 &     3.28 & &     2.01 &     2.61 &     2.01 &     2.96 \\
SV-LEV                         &     2.40 &     2.39 &     2.33 & &     2.36 &     2.49 &     3.27 & &     2.02 &     2.60 &     2.02 &     2.94 \\
SEV-ND                         &     2.37 &     2.35 &     2.29 & &     2.33 &     2.46 &     3.29 & &     2.00 &     2.61 &     2.00 &     2.95 \\
SV2F-LEV                       &     2.38 &     2.36 &     2.32 & &     2.34 &     2.51 &     3.52 & &     1.99 &     2.59 &     1.99 &     2.93 \\
BMJ$(n_J=1,v_J=\frac{1}{4})$   &     2.44 &     2.44 &     2.44 & &     2.36 &     2.40 &     2.51 & &   127.74 &     3.66 &     2.20 &     2.99 \\
BMJ$(n_J=5,v_J=\frac{1}{4})$   &     3.02 &     2.54 &     2.52 & &     2.77 &     2.49 &     2.59 & &    27.87 &     3.80 &     3.51 &     3.29 \\
BMJ$(n_J=10,v_J=\frac{1}{4})$  &     3.16 &     2.68 &     2.61 & &     2.90 &     2.61 &     2.69 & &    15.53 &     3.84 &     4.88 &     3.41 \\
BMJ$(n_J=5,v_J=\frac{1}{2})$   &     4.63 &     2.60 &     2.52 & &     3.81 &     2.52 &     2.59 & &   104.66 &     5.24 &     3.48 &     4.06 \\
BM-outlier                     &     2.46 &     2.47 &     2.46 & &     2.38 &     2.42 &     2.53 & &   127.22 &    89.24 &     2.89 &   237.02 \\
\hline
\end{tabular}
\begin{scriptsize}
\parbox{\textwidth}{\emph{Note}. This table reports the bias and efficiency measure for QRV (and RV, BPV, TRV, and MedRV for comparison) under various model specifications for $N=1,000$. The bias measure in Panel A is equal to 1 for an unbiased IV estimator. The efficiency measure in Panel B takes on a minimum attainable value of 2 for the MLE.}
\end{scriptsize}
\end{center}
\end{sidewaystable}

To implement TRV, we set $\overline{\omega} = 0.47$ and $c = 6\sqrt{IV}$, where IV is estimated using BPV. This parameter choice is in line with \citet*{ait-sahalia-jacod:09a} and ensures that, in our simulation setup, TRV is unbiased in the absence of jumps (alternatively, we could have lowered $c$, gaining robustness to jumps but introducing a downward bias under SV).

Over 100,000 independent simulation runs, we compute a ``bias'' measure $\mathbb{E} ( \widehat{IV} / IV )$ and an ``efficiency'' measure $\var(\sqrt{N}(\widehat{IV} -IV)/\sqrt{IQ})$, where $\widehat{IV} = \{QRV_N, RV_N, BPV_N, TRV_N, MedRV_N\}$. If the estimator is unbiased we expect the bias statistic to be one. Moreover, from the relevant asymptotic results we know the efficiency statistic should be $2$ for RV and TRV, $2.6$ for BPV, $3.0$ for MedRV, and around $2.4$ and $2.3$ for our implementation of the blocked and subsampled QRV, respectively.

From the results in Table \ref{Table:QRVrobust} several interesting patterns emerge. First consider the scenarios without jumps. With model BM, all estimators perform as expected. They are unbiased and their efficiency measure is close to what the asymptotic distribution theory predicts, indicating that it affords a good approximation to finite sample performance. When introducing stochastic volatility through model SV, we find that QRV is biased downwards when few blocks are selected. However, this bias is small for $m=100~/~n=10$ and negligible for $m=20~/~n=50$. Leverage (SV-LEV) does not have a noticeable impact on any of the results. Under the SEV-ND and SV2F-LEV variance specifications, both generating high volatility-of-volatility, the QRV estimator still performs well provided that a sufficient number of subintervals are selected. With a highly erratic volatility path, we need to use large $n$ or small $m$ to ensure good locality of the estimator in line with the discussion above. Finally, comparing the blocked QRV to the subsampled QRV, we see that they perform similarly in terms of bias but that the efficiency of the subsampled QRV deteriorates when $m$ increases. The intuition for this is that the subsampling procedure places less weight on the observations in the first and last block than it places on all other observations. In the asymptotic analysis this effect disappears as $n\to\infty$, but in finite sample it can adversely affect the efficiency of the estimator, particularly in the presence of stochastic volatility. Importantly, the blocking implementation of QRV does not suffer from this and may thus be preferred in situations where $n$ is relatively small.

Next, consider the scenarios with jumps. As expected, QRV enjoys superior robustness to jumps. The small bias we observe (not exceeding 4\% across all scenarios considered) can be explained by noting that the jumps added to the process distort the original ordering of the diffusive returns. This in turn biases the empirical return quantiles used to construct QRV.\footnote{To further clarify intuition for this, consider the following example. Suppose we have a ranked sequence of diffusive returns $\{r_{(1)},r_{(2)},\ldots,r_{(m)}\}$ from the BM model. With $\lambda = 0.95$ we can estimate IV unbiasedly using $r_{(5)}$ and $r_{(95)}$ as described above. Now suppose a positive jump $J$ is added to, say, $r_{(60)}$. If the jump is sufficiently large, the ordered returns sequence becomes $\{ r_{(1)}, r_{(2)}, \ldots, r_{(59)}, r_{(61)}, \ldots, r_{(m)}, r_{(60)} + J \}$ and the ``realised'' $\lambda = 0.95$ quantile is now $r_{(96)}$. Thus, as jumps are added to the price process, the original ordering of returns can be disrupted, leading to a small upward bias in QRV. This bias, however, is only weakly related to the size of the jump.} Importantly, however, this effect is largely independent of the jump size so that QRV maintains excellent robustness in finite sample. Also, with outliers simulated as described above, we see that QRV is virtually unaffected both in terms of bias and in terms of efficiency. Turning to the benchmark estimators, it is well known that in the current setting RV estimates total variation, i.e. $(1+v_J)IV$, explaining the bias and low efficiency when evaluated against IV. BPV is asymptotically immune to jumps, but biased in finite sample: for the BMJ model considered here $\mathbb{E} (BPV_N/IV) \simeq 1+2\sqrt{n_Jv_J/N}$ which can be substantial when jumps become more frequent or volatile. Also, with an outlier in the price series -- effectively constituting two consecutive jumps of opposite sign -- the key assumption underlying the robustness of BPV and MedRV is violated leaving both estimators severely biased. Finally, TRV's performance with large jumps and outliers is comparable to that of QRV, but deteriorates when jumps become smaller and more frequent. TRV requires one to specify a uniform threshold for the sample which makes it sensitive to small jumps if it is set too high and sensitive to stochastic volatility if it is set too low. In this respect, a nice feature of QRV is that it effectively sets a block specific ``threshold'' which naturally adapts to the magnitude of returns through the choice of quantiles.

\section{QRV with market microstructure noise} \label{sec:QRV*}

It has long been recognized that market microstructure effects in high-frequency data -- such a bid-ask bounce and non-synchronous trading -- distort the statistical properties of returns \citep*[e.g.][]{epps:79a, fisher:66a, niederhoffer-osborne:66a} and are detrimental to RV as an estimator of the IV, see, e.g., \cite{zhou:96a}. In this section, we develop a modified version of the QRV that is robust to noise and delivers consistent estimates of the IV.

On a filtered probability space $\bigl( \Omega, \mathcal{F}, \left( \mathcal{F}_{t} \right)_{t \geq 0}, \mathbb{P} \bigr)$, we consider the noisy diffusion model
\begin{equation}
\label{Eqn:Y} Y_{i / N} = X_{i / N} + u_{i /N},
\end{equation}
for $i = 0, 1, \ldots, N$. Here, the ``efficient'' price $X$ is a Brownian semimartingale as in Eq. \eqref{Eqn:X}. The microstructure noise $u$ is an i.i.d. process, independent of $X$, with
\begin{equation*}
\mathbb{E} \left( u_{i/N} \right) = 0, \qquad \mathbb{E} \left( u_{i/N}^{2} \right) = \omega^{2}.
\end{equation*}
The process $Y$ in Eq. \eqref{Eqn:Y} is constructed as follows. Suppose that $X$ is defined on a filtered probability space $\bigl( \Omega^0, \mathcal{F}^0, \left( \mathcal{F}_{t}^0 \right)_{t \geq 0}, \mathbb{P}^0 \bigr)$. We define a second probability space $(\Omega^1, \mathcal F^1, (\mathcal F_t^1)_{t\geq 0}, \mathbb{P}^1)$, where $\Omega^{1}$ denotes $\mathbb{R}^{[0,1]}$ and $\mathcal{F}^{1}$ the product Borel-$\sigma$-field on $\Omega^{1}$. Next, let $\mathbb{Q}$ be a probability measure on $\mathbb{R}$ (the marginal law of $u$). For any $t\geq 0$, $\mathbb{P}_t^{1}=\mathbb{Q}$ and $\mathbb{P}^1$ denotes the product $\otimes_{t \in [0,1]} \mathbb{P}_t^{1}$. The filtered probability space $(\Omega, \mathcal F, (\mathcal F_t)_{t\geq 0}, \mathbb{P})$, on which we define the process $Y$, is given as
\begin{equation*}
\left.\begin{array}{l}
\Omega = \Omega^{0} \times \Omega^{1}, \qquad \mathcal{F} = \mathcal{F}^{0} \times \mathcal{F}^{1}, \qquad \mathcal{F}_{t} = \bigcap_{s>t} \mathcal{F}_s^{0} \times \mathcal{F}_s^{1}, \\
\mathbb{P} = \mathbb{P}^{0} \otimes \mathbb{P}^{1}.
\end{array}\right\}
\end{equation*}
The i.i.d. assumption on $u$ is a natural starting point to analyze the noise case and is widely used in the literature. Moreover, it has some empirical support at moderate sampling frequencies \citep*[see, e.g.,][for a further discussion of this assumption]{hansen-lunde:06a,diebold-strasser:13a}. As before, the object of econometric interest is the IV of $X$, with the additional challenge that inference is now based on noisy high-frequency data.

The modified QRV measure we develop below uses pre-averaged data, which gives it robustness to noise. A similar approach to noise reduction is studied in \citet*{podolskij-vetter:09a} and applied to the BPV estimator. Although the main contribution of this paper is the development of QRV as a highly efficient jump- and outlier-robust measure of the IV, we point out that the treatment of our estimator in the presence of noise goes well beyond \citet*{podolskij-vetter:09a}. First, the procedure developed here makes much more efficient use of the data. Specifically, in a constant volatility setting, the asymptotic variance of QRV can be as low as $8.5\sigma^3\omega$, which is very close to the lower bound of the ML estimator under parametric assumptions and substantially more efficient than the estimator used by \citet*{podolskij-vetter:09a}, which has an asymptotic variance of around $26 \sigma^3 \omega$. Second, our asymptotic theory holds under much weaker assumptions on the noise distribution. Third, we prove that the CLT of the noise-corrected QRV is robust to finite activity jumps, which is something that is not possible for the BPV. Fourth, even though there is no explicit formula available for the conditional variance in the CLT, we provide a consistent estimator of this quantity, which permits a feasible CLT.

\subsection{Construction of the estimator} \label{sec:QRV*construction}

Choose a natural number $K = K(N)$ with
\begin{equation}
\label{Eqn:km} K = c N^{1 / 2} + o\left( N^{  1/4}  \right),
\end{equation}
for some constant $c > 0$, and consider a weight function $h$ on $[0,1]$, which is continuous, piecewise continuously differentiable having a piecewise Lipschitz derivative $h'$ with $h(0) = h(1)= 0$ and that satisfies $\int_0^1 h^2 (s) \text{d}s > 0$. A typical example, that is used in our simulations in Section \ref{sec:QRV*simulations}, is $h(x) = x \wedge (1 - x)$.

Define the return-like statistic
\begin{equation}
\label{Eqn:Ybaerep}
\overline{Y}_j^{N}=  \sum_{i = 1}^{K - 1} h \Big( \frac{i}{K} \Big) \Delta_{j + i}^{N} Y,
\end{equation}
and also set $\psi_1= \int_0^1 (h'(x))^2 \text{d}x$ and $\psi_2= \int_0^1 h^2(x) \text{d}x$.
\begin{remark} \rm
In practice, it is better to use the Riemann approximations
\begin{equation*}
\psi_1^n = K\sum_{j=1}^{K} \left( h \Big( \frac{j}{K} \Big)-h \Big( \frac{j-1}{K} \Big)\right)^2~, \qquad
\psi_2^n =  \frac{1}{K}\sum_{j=1}^{K-1}  h^2 \Big( \frac{j}{K} \Big)
\end{equation*}
of $\psi_1$ and $\psi_2$ to improve the finite sample properties, because $\psi_1^n$ and $\psi_2^n$ are the "true" constants that appear in the computations.
\end{remark}
Next, select a sub-sequence using data observed in the interval $[i/N, (i+m(K-1))/N]$:
\begin{equation*}
\label{Eqn:Dnoise} \overline{D}_i^N Y = \{ \overline{Y}_{i + \left( j - 1 \right) (K-1)}^{N} \}_{j = 1}^{m}, \qquad \text{for}\quad i = 0,1,\ldots,N-m(K-1)
\end{equation*}
and compute
\begin{equation*}
q^\ast_i(m,\lambda )= g^2_{\lambda m} \left( N^{1/4} \overline{D}_i^N Y \right) + g^2_{m-\lambda m+1} \left( N^{1/4} \overline{D}_i^N Y \right).
\end{equation*}
The noise-corrected QRV measure (QRV$^\ast$ hereafter) is now defined as:
\begin{equation}
QRV^\ast_N (m, \overline{ \lambda}, \alpha) \equiv \alpha' QRV^\ast_N (m, \overline{ \lambda}),
\end{equation}
where $\overline{\lambda}$ and $\alpha$ are as above, and the $j$th element of $QRV^\ast_N (m, \overline{ \lambda})$ is given by:
\begin{equation}
QRV^\ast_N (m,\lambda_{j} ) = \frac{1}{c\psi_2 (N-m(K-1)+1)} \sum_{i = 0}^{N - m(K-1) }\frac{q^\ast_i(m, \lambda_{j})}{\nu_1(m,\lambda_{j} )}.
\end{equation}
Note that the constant $K$ controls the stochastic order of the term $\overline{Y}_j^{N}$, since
\begin{eqnarray}
\label{Eqn:uorder}
\overline{u}_{j}^{N} = O_p \left( \sqrt{ \frac{1}{K}} \right), \qquad \overline{X}_{j}^{N} = O_p \left( \sqrt{ \frac{K}{N}} \right).
\end{eqnarray}
Thus, when $K$ is chosen as in Eq. \eqref{Eqn:km} the stochastic orders of the quantities in Eq. \eqref{Eqn:uorder} are balanced (this implies the best rate of convergence), and under mild conditions we have that
\begin{equation*}
N^{1/4} \overline{Y}_j^{N} | \mathcal{F}_{j / N} \stackrel{a}{ \sim} N \left( 0, c\psi_2 \sigma_{j / N}^{2} + \frac{\psi_1}{c} \omega^{2} \right).
\end{equation*}
This demonstrates the rationale of the filtering procedure underlying the construction of QRV$^\ast$, namely while  $N^{1 / 4} \overline{Y}_{j}^{N}$ is affected by the noise through $\omega^2$, it behaves like $\sqrt{N} (X_{i / N} - X_{ \left( i - 1 \right) / N})$.

\subsection{Asymptotic properties}
Our first result shows the consistency of QRV$^{ \ast}$ (after a proper bias correction).
\begin{theorem}
\label{Thm:QRV*consistency} Assume that $m$ is a fixed number and $\mathbb{E} \left( u_{i}^{4} \right) < \infty $. As $N \to \infty $, it holds that
\begin{equation*}
QRV^\ast_N (m, \overline{ \lambda}, \alpha) - \frac{\psi_1}{c^{2} \psi_2}\omega^{2}\overset{p}{ \to} IV.
\end{equation*}
\end{theorem}
\begin{proof}
see Appendix \ref{app:proofs} \qedl
\end{proof}
In practice, we can form consistent estimates of $\omega^{2}$, e.g. $\widehat{\omega}^2 = \frac{1}{2N} \sum_{i=1}^N |Y_{i/N} - Y_{(i-1)/N}|^2$ as in \citet*{bandi-russell:06a}, $\widehat{\omega}^2 = - \frac{1}{N - 1} \sum_{i=1}^{N-1} (Y_{(i+1)/N} - Y_{i / N})(Y_{i / N} -
Y_{(i - 1) / N})$ as in \citet*{oomen:06a}, or with the parametric MA(1)-based maximum likelihood estimator of \citet*{ait-sahalia-mykland-zhang:05a}. As a consequence, we have the convergence
\begin{equation*}
QRV^\ast_N (m, \overline{ \lambda}, \alpha) - \frac{\psi_1}{c^{2} \psi_2} \widehat{ \omega}^{2} \overset{p}{ \to} IV.
\end{equation*}
This result is robust to the presence of finite activity jumps. Also note that because $\widehat{ \omega}^{2}$ is a $\sqrt N$-estimator of $\omega^{2}$, it will not influence the CLT of the slower converging QRV$^\ast$.

To prove the CLT, it is useful to introduce some further notation.
\begin{definition} \label{def2}
For $x\in \mathbb{R}$, $u\in [0,1]$, $l=1, \ldots, m$  and $\lambda_1$, $\lambda_2$  we define the quantity
\begin{equation}\label{Eqn:def2function}
f_{m,l,x,u} (\lambda_1, \lambda_2)= \text{\upshape{cov}} \Big(g^2_{\lambda_1 m} (S) + g^2_{m-\lambda_1 m+1} (S), g^2_{\lambda_2 m} (T) + g^2_{m-\lambda_2 m+1} (T) \Big)~,
\end{equation}
where $S=(S_1,\ldots, S_m)^T$, $T=(T_1,\ldots, T_m)^T$ are centered and jointly normal with
\begin{itemize}
\item[(i)] $S_i\bot S_j$, $T_i\bot T_j$ for all $i \not= j$.
\item[(ii)] $\text{\upshape{var}}(S_i)= \text{\upshape{var}}(T_i)=c \psi_2  x^2 + \frac{\psi_1}{c}  \omega^2$  for all $i$.
\item[(iii)] $\text{\upshape{cov}}(S_{i+l-1},T_i)=c w_{h}(u) x^2 + \frac{1}{c} w_{h'} (u)
    \omega^2$ for all $i$.
\item[(iv)] $\text{\upshape{cov}}(S_{i+l},T_i)=c w_{h}(1-u) x^2 + \frac{1}{c} w_{h'} (1-u)
    \omega^2$ for all $i$.
\item[(v)] $\text{\upshape{cov}}(S_i,T_j)=0$ for all $|i+l-j-1|>1$.
\end{itemize}
Here the function $w_h(u)$ is defined by
\begin{equation*}
w_{h} \left( u \right) = \int_{0}^{1 - u} h \left( y \right) h \left( y + u \right) \text{\upshape{d}}y.
\end{equation*}
When $\lambda =\lambda_1=\lambda_2$ we use the notation $f_{m,l,x,u} (\lambda )=f_{m,l,x,u} (\lambda_1, \lambda_2)$.
\end{definition}
Notice that $h'$ (the derivative of $h$) exists almost everywhere, so the quantity $w_{h'}$ makes sense.

\begin{theorem}
\label{Thm:QRV*centrallimit} Assume that $m$ is a fixed number, $\mathbb{E} \left( u_{i}^{8} \right)<\infty$, the marginal distribution $\mathbb{Q}$ of $u$ is symmetric around $0$ and that condition (V) is satisfied. As $N \rightarrow \infty$
\begin{equation*}
N^{1 / 4} \left( QRV^{ \ast}_{N} (m,\overline{\lambda},\alpha)- \frac{\psi_1}{c^{2} \psi_2} \widehat{\omega}^2 - IV \right) \overset{d_{s}}{ \to} MN \biggl(0, \frac{2}{c\psi_2^2} \Sigma_m (\lambda_1, \ldots, \lambda_k) \biggr),
\end{equation*}
where $\Sigma_m(\lambda_1, \ldots, \lambda_k)=(\Sigma_m(\lambda_1, \ldots, \lambda_k)_{sl})_{1\leq s,l\leq k}$ is given by
\begin{equation*}
\Sigma_m(\lambda_1, \ldots, \lambda_k)_{sl}= \frac{1}{ \nu_{1,m}(\lambda_s) \nu_{1,m}(\lambda_l) } \sum_{l=1}^m \int_0^1 \int_0^1 f_{m,l,\sigma_t,u} (\lambda_s, \lambda_l) \text{\upshape{d}}t \text{\upshape{d}}u.
\end{equation*}
Furthermore, this convergence is robust to the presence of finite activity jumps.
\end{theorem}
\begin{proof}
see Appendix \ref{app:proofs} \qedl
\end{proof}
A couple of points are worth highlighting. First, the rate of convergence in Theorem \ref{Thm:QRV*centrallimit} is $N^{-1/4}$, which is known to be optimal in the noisy diffusion model \citep*{gloter-jacod:01a,gloter-jacod:01b}. Second, \cite{jacod-li-mykland-podolskij-vetter:09a} show that when the IV is estimated using a ``sum-of-squares'' estimator based on filtered data:
\begin{equation*}
\frac{1}{c \psi_2 (N - K)} \sum_{j = 1}^{N - K} | \overline{Y}_{j}^{ N} |^{2},
\end{equation*}
the lowest attainable variance for the choice $h(x) = x \wedge (1 - x)$ is roughly $8.5 \sigma^{3} \omega$ assuming $\sigma$ is constant (the variance of the ML estimator is $8 \sigma^{3} \omega$). Consequently, the lower bound for the variance of QRV$^\ast$ is also $8.5\sigma^{3} \omega$ (note that for a suitable choice of parameters the realised kernel of \cite{barndorff-nielsen-hansen-lunde-shephard:08a} can attain a variance of $8.002 \sigma^{3} \omega$). Finally, even though there is no explicit expression for the conditional variance in the CLT, it is nonetheless possible to estimate it from the data.
\begin{proposition} \label{prop:QRV*feasible}
Assume that $m$ is fixed and $\mathbb{E} \left( u_{i}^{8} \right)<\infty $. As $N \to \infty $, it holds that
\begin{eqnarray*}
\frac{1}{c \psi_2^2 (K-1) (N-3m(K-1)+3)} \sum_{i = m(K-1)-1}^{N - 2m(K-1) + 1}\frac{q^\ast_i(m,\lambda_s)}{\nu_1(m,\lambda_s)}
\sum_{j=i-m(K-1)+1}^{i+m(K-1)-1} \frac{q^\ast_j(m,\lambda_l) - q^\ast_{i+m(K-1)}(m,\lambda_l)}{\nu_1(m,\lambda_l)} \nonumber
\\[1.5 ex] \overset{p}{ \to} \frac{2}{c\psi_2^2} \Sigma_m(\lambda_1, \ldots, \lambda_k)_{sl}.
\end{eqnarray*}
\end{proposition}
\begin{proof}
see Appendix \ref{app:proofs}\qedl
\end{proof}
Using the estimator from Proposition \ref{prop:QRV*feasible}, we obtain a feasible CLT for QRV$^\ast$ in the exact same manner as discussed in Section \ref{sec:QRV} for QRV.

\subsection{Finite sample performance and noise robustness} \label{sec:QRV*simulations}

The simulations below are designed to illustrate the performance of QRV$^\ast$ in the presence of market microstructure noise, comment on reasonable choice of the pre-averaging window width $K$, and make a comparison to alternative estimators. For ease of exposition, we mean QRV$^\ast$ to include the bias correction term $ \psi_1\widehat{\omega}^2/(\psi_2c^2)$ throughout the remainder of this paper. To simulate the ``efficient'' price process, we use the BM model as in Eqs. (\ref{Eqn:Xsim} -- \ref{Eqn:BM}) and add i.i.d. noise as in Eq. \eqref{Eqn:Y}. To ensure our simulation setup is realistic, we base our choice of parameters on a comprehensive set of summary statistics of global equity trade data as reported in Appendix \ref{app:noisestats}. We set the number of high-frequency return observations $N = \{1,000;10,000\}$ representing typical small-to-mid and large-cap stocks. The level of microstructure noise is set to $\omega^2 = \gamma^2 IV / N$, where $\gamma^2=\{0.25;2.50;10\}$. From Table \ref{Table:noisestats}, we see that this covers average, high, and extreme levels of noise. Note that the noise is normalized with respect to the IV, and that the so-called noise ratio $\gamma$ \citep*[see][]{oomen:06a} has a natural interpretation in relation to the bias of the RV, as $\mathbb{E} (RV) = IV (1 + 2 \gamma^2)$. To implement QRV$^\ast$, we use the quantiles as before, i.e. $\overline{\lambda} = \{0.80;0.85;0.90;0.95\}$, set $m = 40$, estimate $\omega^2$ as in \citet*{oomen:06a}, and vary $K$ between 1 and 25. To provide a benchmark for our results, we compute the multi-scale RV (MSRV) of \cite{zhang:06a}:
\begin{equation}
\label{Eqn:MSRV} MSRV_N(q) = \sum_{j = 1}^{q} \frac{a_j}{j} \sum_{h = 0}^{j - 1}\gamma_{h,j}(0),
\end{equation}
where $q$ denotes the number of subsamples and
\begin{equation*}
\gamma_{h, q} = \sum_{i = 1}^{N} (Y_{iq+h}-Y_{(i-1)q+h})^2,\quad\text{and}\quad a_{j} = (1-1/q^2)^{-1}\left(\frac{j}{q^2}h(j/q)-\frac{j}{2q^3}h^\prime(j/q)\right),
\end{equation*}
for $h(x) =12(x-1/2)$. In the simulations, we use the optimal number of subsamples, which can be chosen as $q_Z^\ast = c^\ast \sqrt{N}$, where
\begin{equation}
\label{Eqn:cstarMSRV} c^\ast = \arg \min_{c} \left\{ 2\frac{52}{35}c IQ+\frac{48}{5}c^{-1}\omega^2(IV+\omega^2/2)+48 c^{-3}\omega^4 \right\}.
\end{equation}
Both MSRV and QRV$^\ast$ are consistent for the IV and converge at rate $N^{- 1 / 4}$, but the MSRV is not robust to jumps or outliers. The same is true for the two-scale RV of \cite{zhang-mykland-ait-sahalia:05a} and the realised kernel of \cite{barndorff-nielsen-hansen-lunde-shephard:08a}. Because the performance of these estimators is very similar in the current setting, we concentrate on MSRV to conserve space.

\begin{figure}[t!]
\begin{center}
\caption{Performance of QRV$^\ast$ in the presence of noise.}
\label{Figure:QRV*}
\begin{tabular}{cc}
\footnotesize{Panel A: log MSE of QRV$^\ast$ with $N = 1,000$} &
\footnotesize{Panel B: log MSE of QRV$^\ast$ with $N = 10,000$} \\
\includegraphics[height=8cm,width=0.45\textwidth]{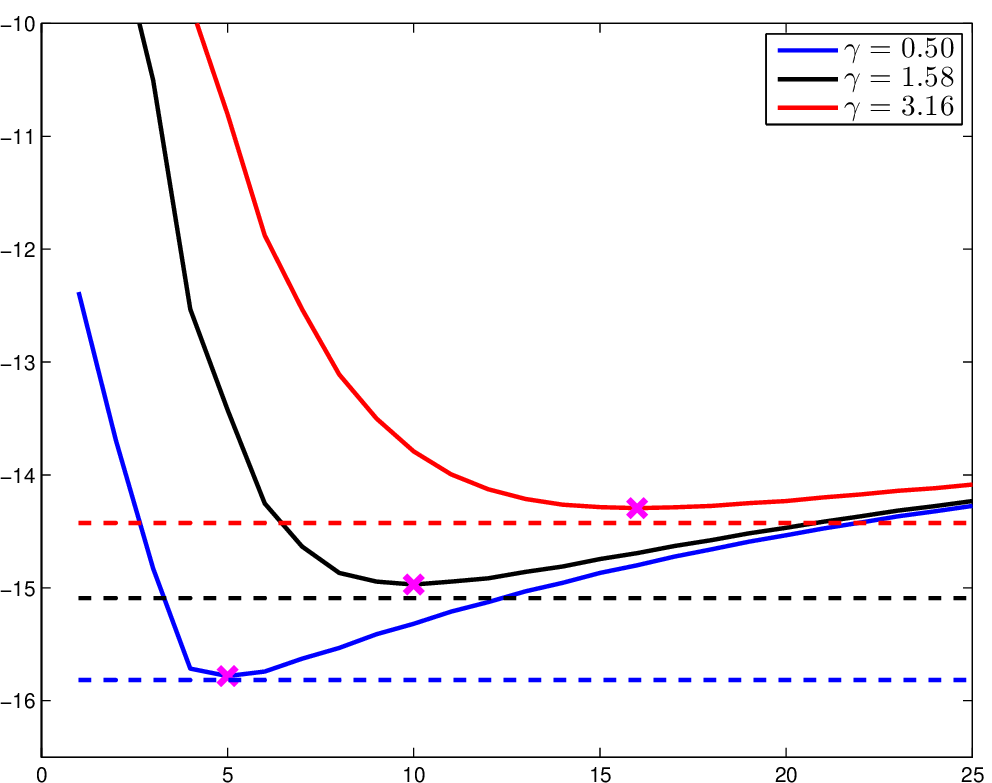} &
\includegraphics[height=8cm,width=0.45\textwidth]{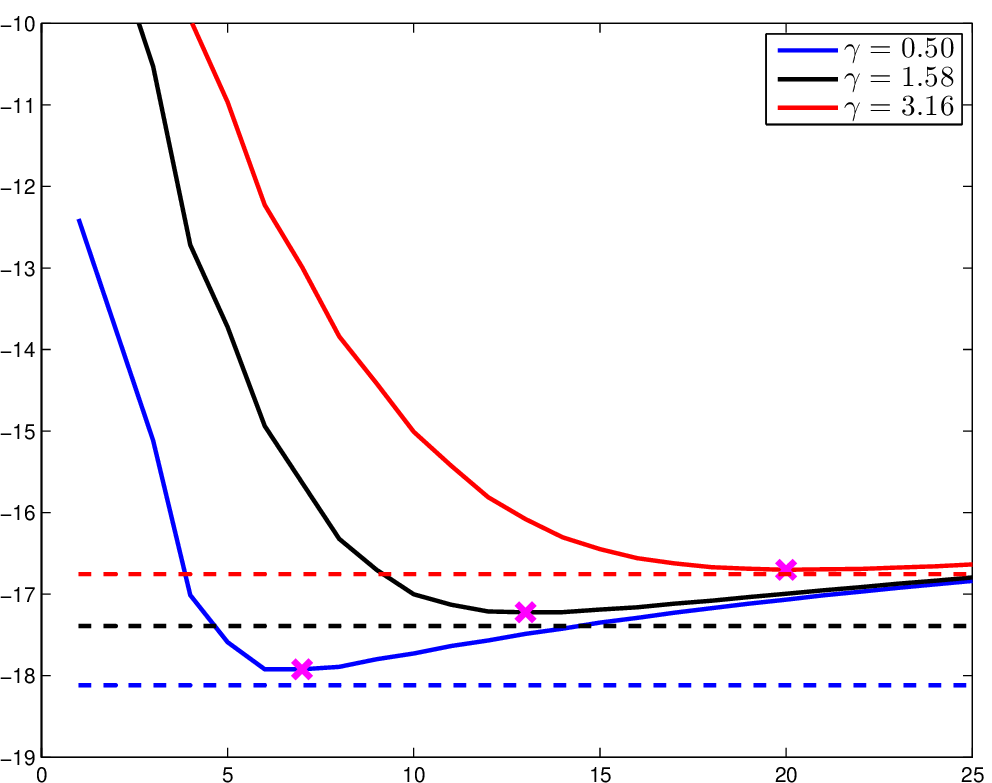} \\
\end{tabular}
\begin{footnotesize}
\parbox{0.98\textwidth}{\emph{Note}. This figure plots the (log) MSE of the bias-corrected QRV$^\ast$ for sample size $N=1,000$ (Panel A) and $N = 10,000$ (Panel B) and various levels of microstructure noise $\gamma$. The crosses indicate the minimum MSE, and thus identify the optimal choice of $K$. The dashed horizonal lines indicate the log MSE of MSRV using optimal number of subsamples.}
\end{footnotesize}
\end{center}
\end{figure}

Figure \ref{Figure:QRV*} plots the log MSE of QRV$^\ast$ as a function of $K$ for the simulations described above. The dashed horizontal lines indicate the performance of MSRV. As expected, the MSE minimizing choice of $K$ increases in $\gamma$ and $N$: the optimal choice of pre-averaging window width balances the noise reduction it achieves at the cost of efficiency loss. Perhaps most importantly, we see that the performance of QRV$^\ast$ is comparable to that of MSRV across all scenarios considered. While QRV$^\ast$ is only slightly inferior to MSRV in terms of efficiency, it comes of course with the benefit of being robust to jumps and outliers. The empirical application below will further illustrate this point. Regarding the optimal choice of $K$ we can make the following observations. Although we have no explicit asymptotic guidance available, it is clear that the choice of window width can be informed by simulations as conducted here. In particular, for given $N$ and noise ratio $\gamma$ -- the first quantity is readily available and the second can be estimated straightforwardly from the data -- the optimal value for $K$ can be read off a simulated MSE curve like the one in Figure \ref{Figure:QRV*}. Also note that the MSE loss function is highly asymmetric in $K$, and so a conservative choice of pre-averaging window is generally preferred. This also helps to reduce the effects of price discreteness often encountered in high-frequency data (see Section \ref{sec:QRVempirical}) and makes the estimator less sensitive to potential violations of the i.i.d. noise assumption.

\section{Empirical illustration}\label{sec:QRVempirical}

In this section, we apply the QRV estimator to a variety of equity data. The aim here is to illustrate the practical implementation of QRV and highlight some of its empirical properties. We use clean low-frequency data over long horizons as well as noisy high-frequency data over short horizons and find that in both cases the performance of the QRV is good compared to RV and its microstructure noise robust counterparts. To facilitate the discussion and interpretation of our results, we express all estimates as annualized standard deviations throughout this section.

\subsection{QRV with ``clean'' low-frequency data}

The use of QRV, like any other RV measure, is not merely limited to high-frequency data over short horizons, but can also be applied to low-frequency data over longer horizons. In the latter case the impact of market microstructure noise is benign and can be ignored for all practical purposes. See, for instance, \citet*{schwert:89a} who calculates monthly RVs using daily data or, more recently, \citet*{andersen-bollerslev-diebold-wu:06a} who study quarterly RVs and realised betas calculated from daily data.

\begin{figure}[t!]
\begin{center}
\caption{QRV with daily Dow Jones Industrial Average index data.}
\label{Figure:DJ30}
\begin{tabular}{cc}
\footnotesize{Panel A: QRV and RV against time} &
\footnotesize{Panel B: QRV (vert) vs RV (horz)} \\
\includegraphics[height=8cm,width=0.45\textwidth]{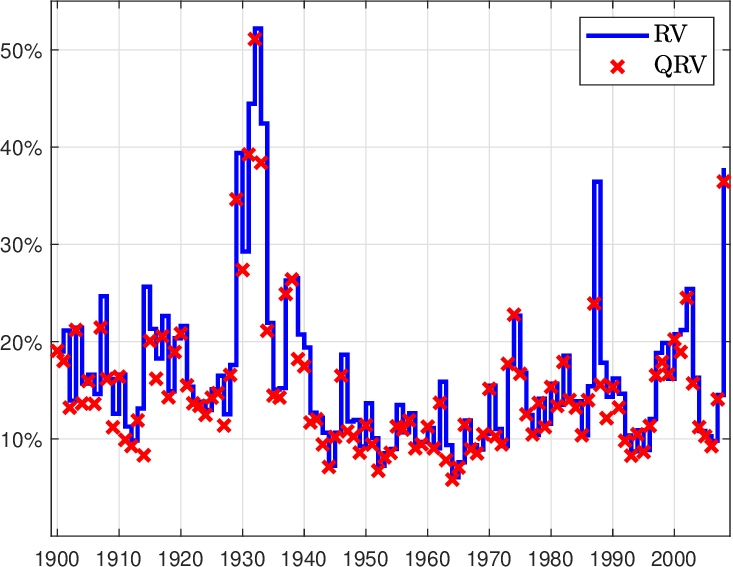} &
\includegraphics[height=8cm,width=0.45\textwidth]{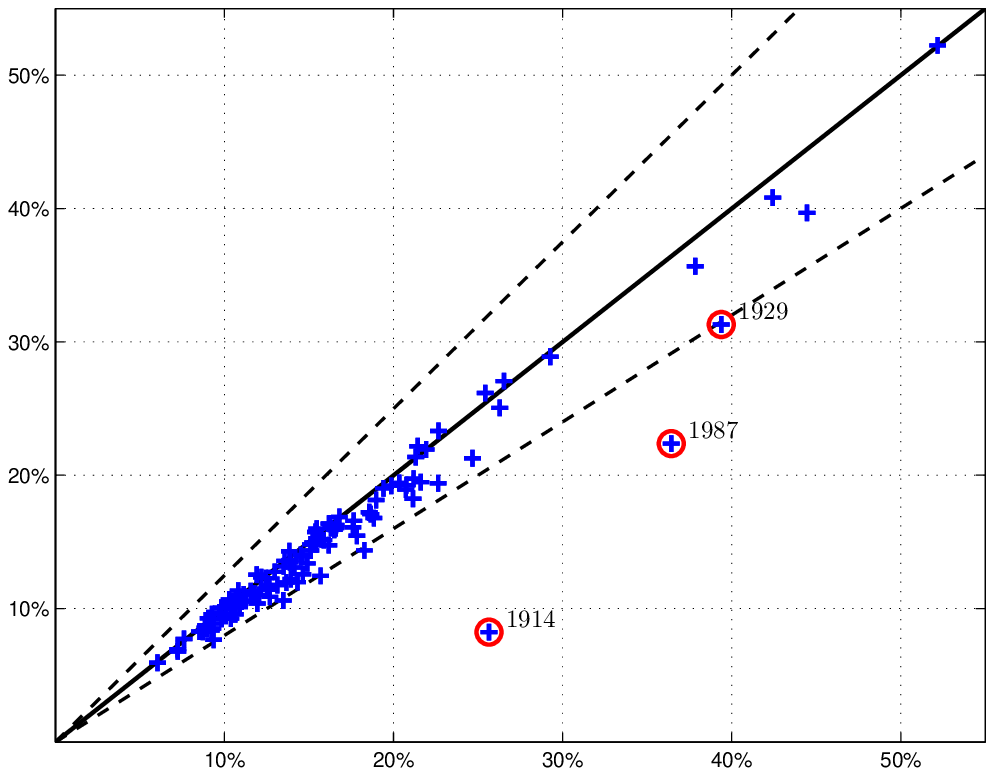}
\end{tabular}
\begin{footnotesize}
\parbox{0.98\textwidth}{\emph{Note}. Year-by-year QRV and RV estimates from daily DJIA index data over the period January 2, 1900 through December 30, 2008. QRV is calculated using $\overline{ \lambda} = \{0.80;0.85;0.90;0.95 \}$ and $m = 60$. The estimates are reported as annualized standard deviation. The dashed lines in Panel B mark the region where RV and QRV differ by more than 25\%.}
\end{footnotesize}
\end{center}
\end{figure}

As our first illustration, we look at daily data for the Dow Jones Industrial Average (DJIA) stock index over the sample period January 1900 through December 2008, i.e. 27,306 daily observations spanning more than a century.\footnote{Source: Dow Jones Indexes, http://www.djindexes.com/} For each year in the sample, we estimate the ex-post return variation using QRV and RV from these daily data (i.e. $N\simeq 250$). To implement QRV, we use $\overline{ \lambda} = \{0.80, 0.85, 0.90, 0.95 \}$ with $m = 60$ and use the subsampling implementation. As volatility is widely documented to be very persistent, a quarterly block length provides sufficient locality. At the same time, with $(1-\lambda_{\max})m - 1 = 2$, we are robust to up to four jumps or two outliers per quarter. We found, however, that our results were insensitive to reasonable alternative choices of quantiles and block length. Also, the subsampled and blocked implementation of QRV yield very similar estimates.

In Figure \ref{Figure:DJ30}, we plot the time series of variance estimates in Panel A and a cross plot of RV (on the horizontal axis) versus QRV (on the vertical axis) in Panel B. We can see that the instances where QRV deviates substantially from RV all correspond to years with extreme market movements. For example, in 1914, the DJIA closed on July 30, 1914 at 71.42 and reopened more than 4 months later on December 14, 1914 at 56.76, reflecting a 20\% drop in value. In 1929, the start of the great depressions, the DJIA index fell 13.5\% on October 28, another 11.7\% the next day, only to rebound by 12.3\% on October 30. Similarly, in 1987 the stock market crashed again, experiencing a daily return of $-22.6\%$ on October 19 which, even with the RV estimate of 38\% for that year, constitutes a nine-standard deviation event. All this illustrates the robustness of QRV to jumps. In the remaining years, the QRV estimates are close to those of RV (with a sample correlation exceeding 0.99) indicating good efficiency and absence of any systematic biases. As an aside, note that BPV and MedRV are not robust to the jump scenario experienced in 1929 with three consecutive large returns.

Over the full 108 year sample period, we calculate an average annualized volatility estimate of $18.05\%$ for RV and $16.75\%$ for QRV, which suggests that roughly $14\%$ of total variation can be attributed to jumps. If we leave out the three years discussed above, these figures drop to $17.4\%$ and $16.5\%$, respectively, indicating that about $9\%$ of total variation is due to jumps. Interestingly, using BPV \citet*{andersen-bollerslev-diebold:07a} estimate the jump contribution to total variation at 14.4\% using 5 minute S\&P500 futures data over the period $1990 - 2002$, \citet*{huang-tauchen:05a} estimate the contribution at 7\% for 5 minute S\&P500 cash data from $1997-2002$ and 4.5\% for 5 minute S\&P500 futures data from $1982 - 2002$, while \citet*{corsi-reno:12a} estimate the contribution at around 10\% using 5 minute S\&P500 futures data from $1990 - 2007$. In a related study, \citet*{eraker-johannes-polson:03a} estimate continuous time jump diffusion models using daily S\&P500 index returns over the period $1980 - 1999$ and measure the jump contribution between 8.2\% and 14.7\% depending on the model considered. While a direct comparison of these results is difficult (as different data, sampling frequency, horizon,
and econometric techniques are used), it does illustrate that our estimates are plausible and in line with the extant literature.

\subsection{QRV with ``noisy'' high-frequency data}

Our second illustration is based on noisy high-frequency data with volatility computed over short daily horizons. We study Apple Inc. (AAPL) trade data over the period May 1, 2006 through December 30, 2008, which were extracted from the NYSE TAQ database. We only include trades from the primary exchange (i.e. NASDAQ) and aggregate records with the same millisecond precision time-stamp into one observation using the volume-weighted average trade price. To illustrate the robustness of QRV to jumps and outliers, we do \emph{not} filter the data on qualifiers. The final dataset contains a record of 35,419,565 observations over 672 trading days, an average of about 52,708 trades per day.

Despite the deep liquidity of AAPL, its trade data are inherently noisy due to presence of bid-ask spread bounce. This is confirmed in Panel A of Figure \ref{Figure:preavg}, where we find substantial autocorrelation in returns. Also, with trade data at this frequency, price discreteness is a concern: from Panel B we see that the vast majority of return observation are either zero or plus/minus one tick, and virtually all observations are less than 8 ticks in magnitude. For these reasons, it is clearly inappropriate to apply the standard QRV. Instead, we use our noise-robust QRV$^\ast$. To implement this estimator we use, as before, four pairs of quantiles $\overline{ \lambda} = \{0.80,0.85,0.90,0.95\}$, set $K = 15$, and $m=240$. This choice of parameters ensures substantial robustness to jumps and outliers with $(1-\lambda_{\max})m - 1 = 11$. Moreover, with an effective block length of about 30 minutes for an average day (i.e. $15\times240/52708\times 390$), the estimator is sufficiently ``local'' for it to pick up time variation in volatility. Also, because the data is sampled in event time, return volatility is homogenized to a large extent and this makes the results very insensitive to the choice of $m$. Finally, the choice of $K$ is guided by simulations as in Section \ref{sec:QRV*simulations}: unreported results show that the optimal $K$ is around 10 for representative values of AAPL sample size and noise level (i.e. $N\approx50,000$ and $\gamma \approx 0.5$). As mentioned above, a conservative choice of $K$ is advised to account for price discreteness and other features of the data not captured by the BM plus i.i.d. noise model. Besides, given the large amount of data available here, efficiency is less of a concern. For these reasons we set $K=15$, but it turns out that virtually identical results are obtained with $K=10$. The effect of pre-averaging is nicely illustrated in Panels C and D of Figure \ref{Figure:preavg}. The pronounced serial correlation observed in raw returns is virtually eliminated for the pre-averaged data. At the same time, price discreteness is heavily reduced and the return distribution is now much closer to Gaussian.

\begin{figure}[t!]
\begin{center}
\caption{Summary statistics of ``noisy'' AAPL trade data.}
\label{Figure:preavg}
\begin{tabular}{cc}
\footnotesize{Panel A: autocorrelation of raw returns}   & \footnotesize{Panel B: histogram of raw returns} \\
\includegraphics[height=8cm,width=0.45\textwidth]{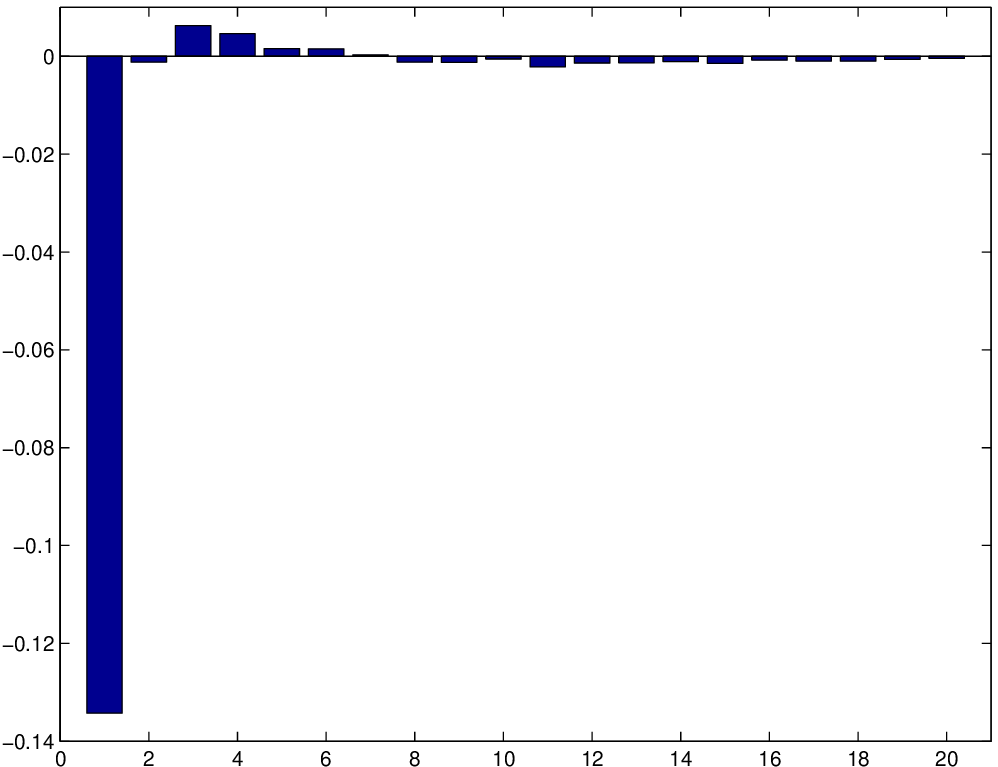} &
\includegraphics[height=8cm,width=0.45\textwidth]{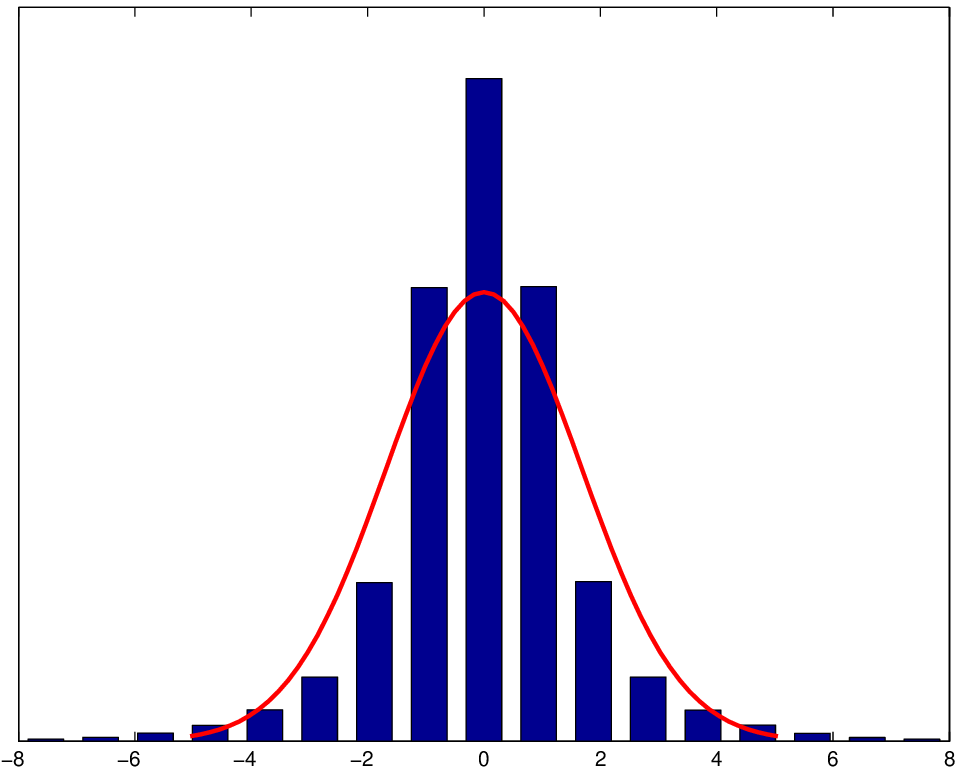} \\
\footnotesize{Panel C: autocorrelation of pre-averaged returns ($K = 15$)} & \footnotesize{Panel D: histogram of pre-averaged returns ($K = 15$)} \\
\includegraphics[height=8cm,width=0.45\textwidth]{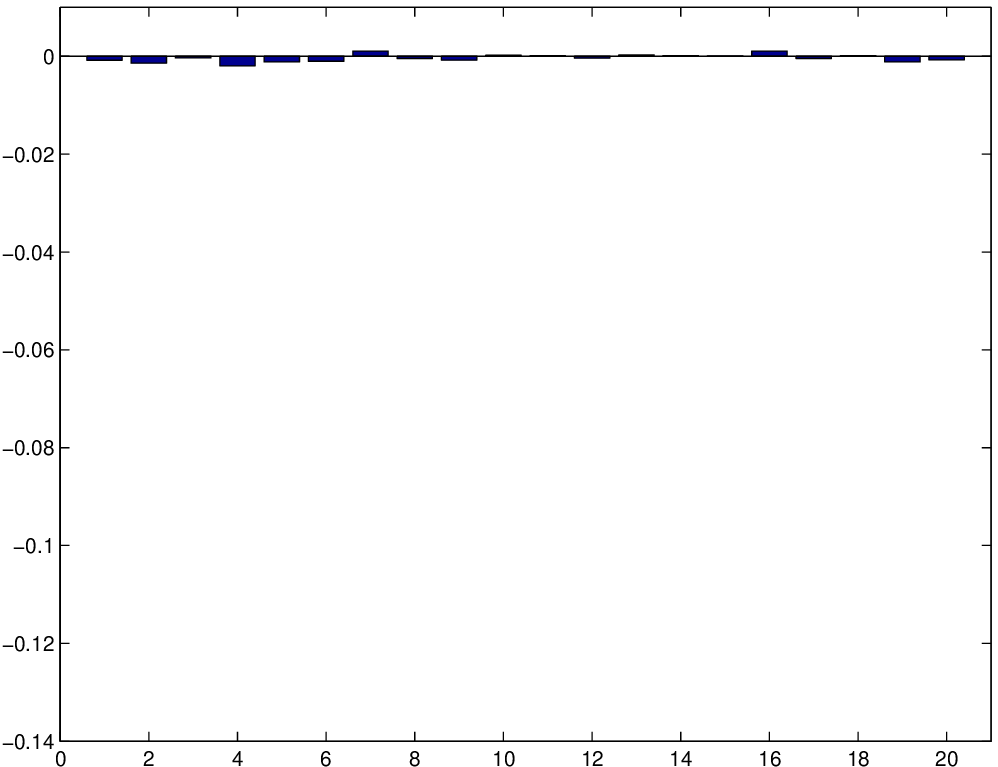}&
\includegraphics[height=8cm,width=0.45\textwidth]{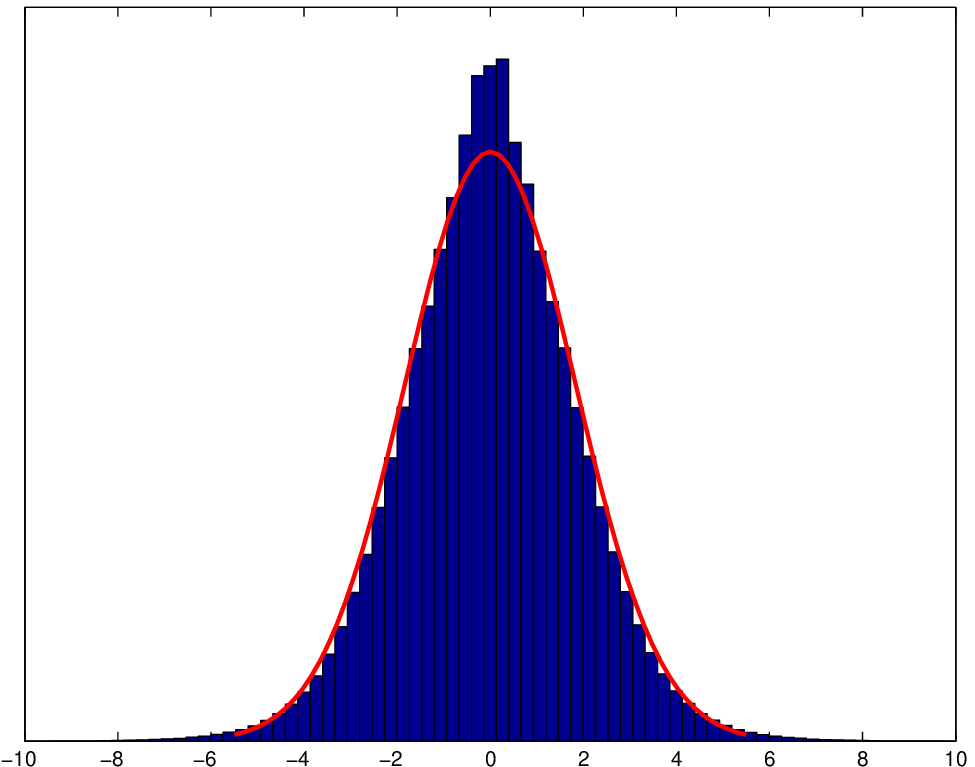}
\end{tabular}
\begin{footnotesize}
\parbox{0.98\textwidth}{\emph{Note}. Summary statistics of raw and pre-averaged AAPL trade data over period May 1, 2006 through December 30, 2008. Returns are expressed in basis points.}
\end{footnotesize}
\end{center}
\end{figure}

Figure \ref{Figure:tradeAAPL} draws a time series (in Panel A) and scatter plot (in Panel B) of QRV$^\ast$ compared to MSRV of \citet*{zhang:06a}.\footnote{The MSRV is implemented with an optimal bandwidth as in Eq. \eqref{Eqn:cstarMSRV}, estimated for each day in the sample separately. $\omega^{2}$ is estimated from the first-order autocovariance of the trade data, whereas IV and IQ estimates are calculated by subsampling the equivalent of 5 minute data in trade time (i.e. each subsample consists of 79 price observations per day). The average optimal bandwidth is 3.28, with a minimum of 2, a maximum of 9, and $q^\ast = 3$ for more than three out of every four days.} As in the previous illustration, we find a close alignment between the two estimates with a few noticeable exceptions. We highlight two here. On May 11, 2007 our sample of unfiltered trade data contains numerous ``out-of-sync'' records as indicated by the qualifier ``Z'' in the TAQ data.\footnote{\citet{nasdaq:08a} describes this qualifier as ``Sold Out of Sequence is used when a trade is printed (reported) out of sequence and at a time different from the actual transaction time.''} The time series of highly erratic transaction prices is displayed in Panel C of Figure \ref{Figure:tradeAAPL}, from which it is evident that standard volatility estimators will be heavily distorted. On the raw data, MSRV estimates daily volatility at $39.3\%$, which more than halves to $19.2\%$ after removing the spurious observations. In sharp constrast, QRV gives reliable estimates both on the raw and cleaned data, i.e. $20.0\%$ and $18.0\%$ respectively. The second example is October 11, 2007. See Panel D of Figure \ref{Figure:tradeAAPL} for the price path: the unfiltered data contains a number of trades that are seemingly executed at prices well below fair market value. Closer inspection of the data reveals that several -- but not all -- of these observations were flagged by the NASDAQ as suspicious (using a ``yellow flag'' indicated by qualifier ``Y''\footnote{\citet{nasdaq:08a} describes this qualifier as ``Market Centers will have the ability to identify regular trades being reported during specific events as out of the ordinary by appending a new sale condition code Yellow Flag (``Y'') on each transaction reported.''}) or re-opening prints (as indicated by qualifier ``O'' or ``5''). On the raw data, MSRV estimates volatility at $4,183.29\%$ which drops to $67.88\%$ after removing the spurious data records. QRV, on the other hand, again enjoys remarkable robustness to these outliers and estimates volatility at $63.59\%$ on the raw data and at $62.90\%$ on the cleaned data.

\begin{figure}[t!]
\begin{center}
\caption{QRV$^\ast$ with ``noisy'' AAPL trade data.}
\label{Figure:tradeAAPL}
\begin{tabular}{cc}
\footnotesize{Panel A: QRV and MSRV against time} & \footnotesize{Panel B: QRV vs MSRV} \\
\includegraphics[height=8cm,width=0.45\textwidth]{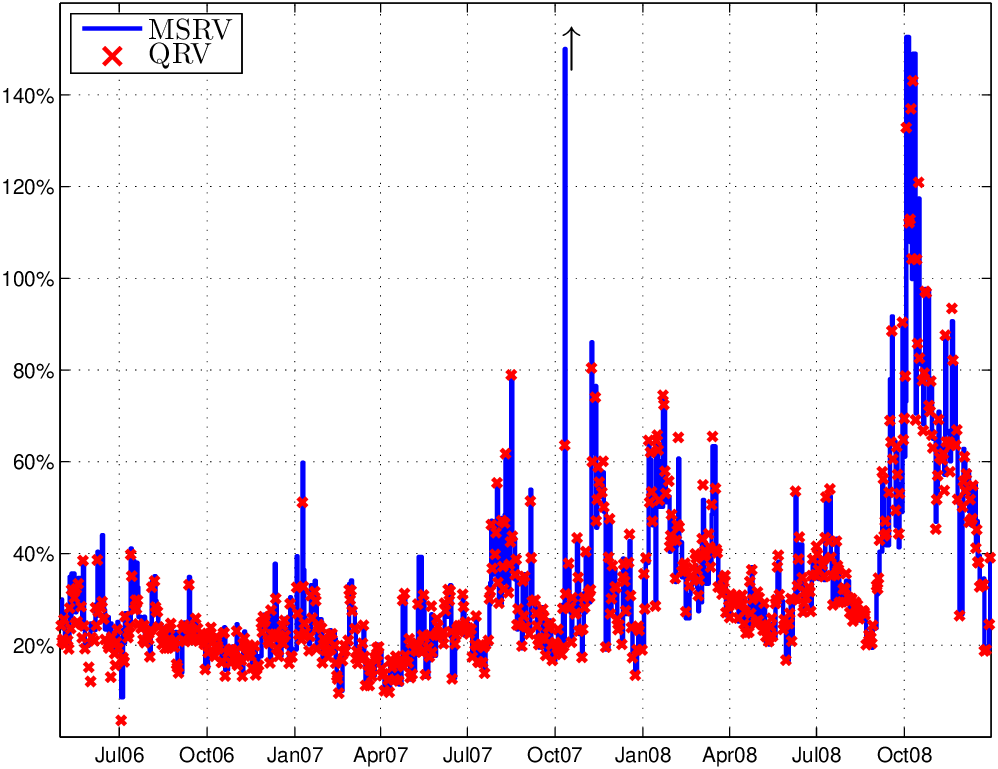} &
\includegraphics[height=8cm,width=0.45\textwidth]{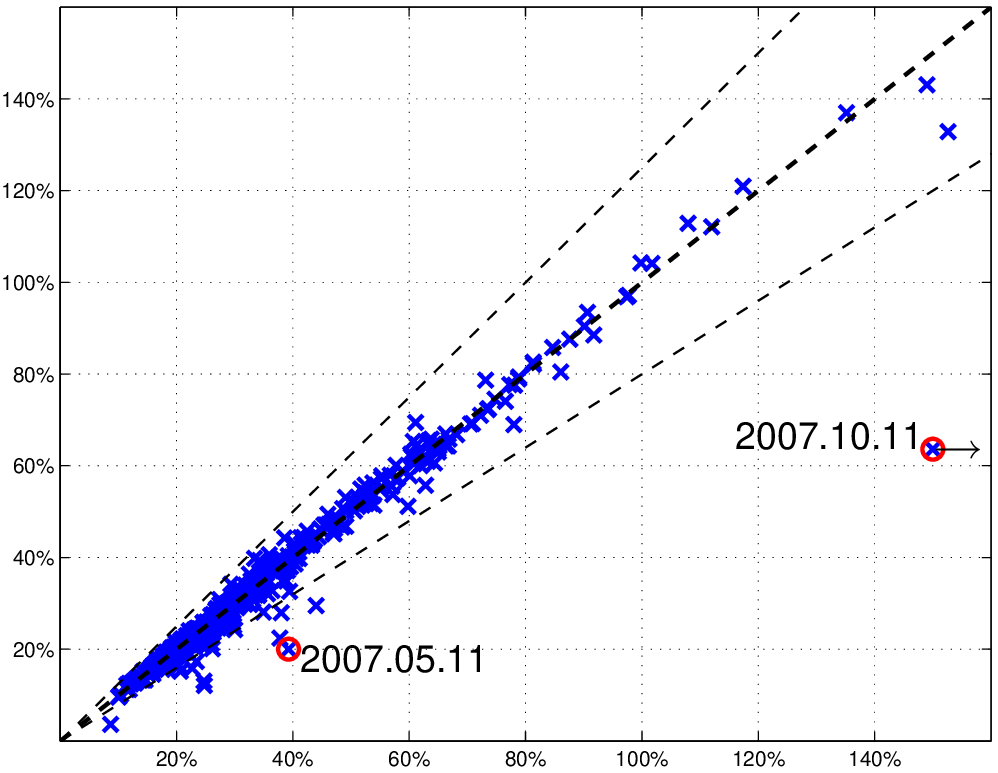} \\
\footnotesize{Panel C: trade prices 2007.05.11 (9:50 -- 10:05)} & \footnotesize{Panel D: trade prices 2007.10.11 (9:45 -- 16:00)}\\
\includegraphics[height=8cm,width=0.45\textwidth]{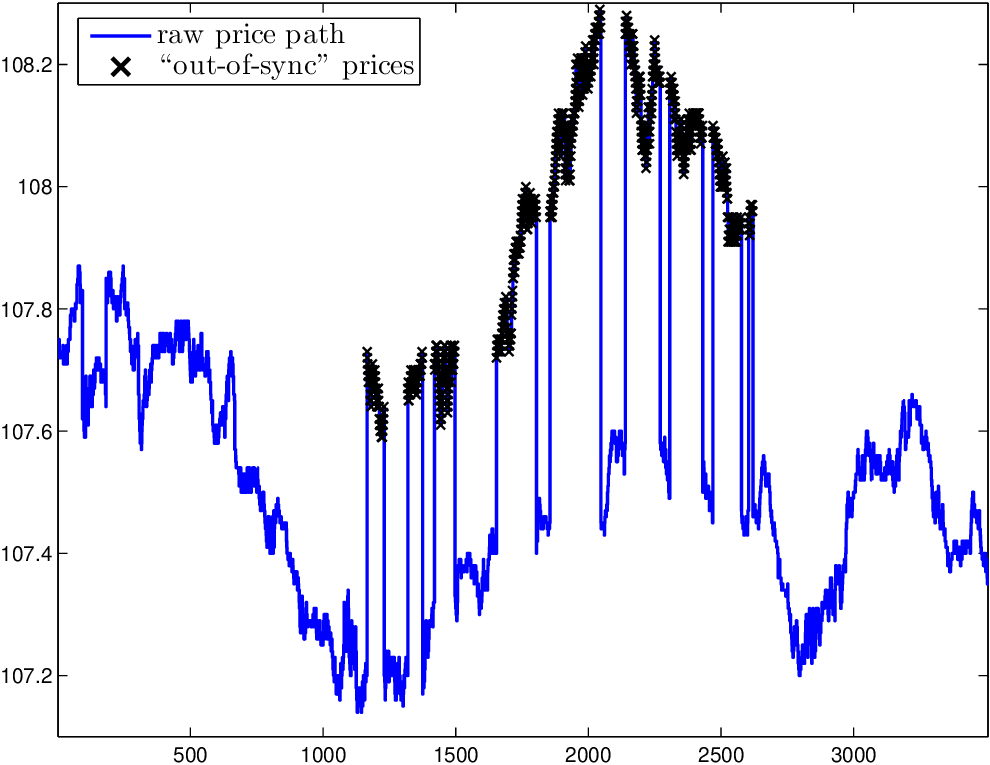} &
\includegraphics[height=8cm,width=0.45\textwidth]{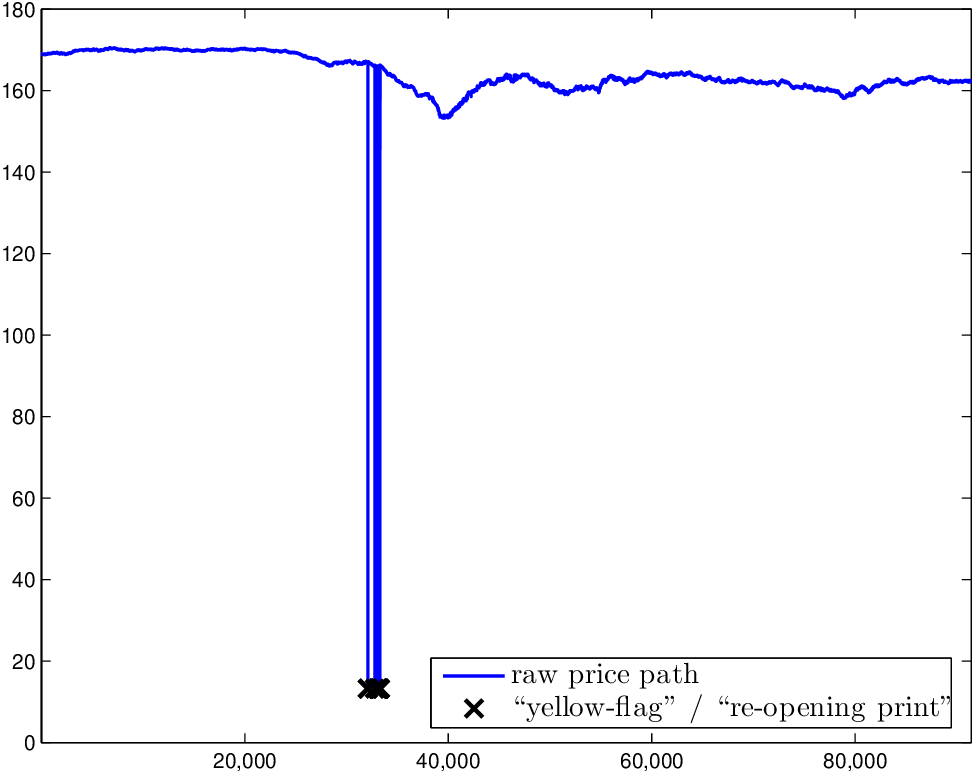}
\end{tabular}
\begin{footnotesize}
\parbox{0.98\textwidth}{\emph{Note}. Day-by-day QRV and MSRV estimates from unfiltered tick-by-tick trade prices of AAPL over the period May 1, 2006 through December 30, 2008. QRV$^\ast$ is calculated using $\overline{ \lambda} = \{0.80;0.85;0.90;0.95\}$, $m=240$, and $K = 15$. MSRV is calculated using the optimal number of subsamples as in Eq. \eqref{Eqn:cstarMSRV}. Estimates are expressed as annualized standard deviation. The thin dashed lines in Panel B mark the region where MSRV and QRV differ by more than 25\%. The MSRV estimate for 2007.10.11 is 4,183.29\%, but to preserve the scaling of the graph it is capped at 150\% (as indicated by the arrow).}
\end{footnotesize}
\end{center}
\end{figure}

Averaging over all days in the sample (excluding October 11, 2007), we compute a return volatility of $37.32\%$ with MSRV and $37.15\%$ with QRV$^\ast$. In stark contrast to our previous illustration with daily data and the studies by \citet*{andersen-bollerslev-diebold:07a, corsi-reno:12a, eraker-johannes-polson:03a, huang-tauchen:05a} amongst others, we now find that less than 1\% of total variation can be attributed to jumps! This is rather surprising particularly because we are considering here a single stock instead of a well diversified index like the S\&P500 and cover a period that includes the exceptionally turbulent 2008. Moreover, because some of the larger deviations between QRV and MSRV in our sample are due to further spurious data points, we fail to identify a single instance, where a true jump in the price is observed at this frequency. Hence, this raises the question whether previously identified jumps in the literature are in fact actual jumps in the price, which become difficult to identify at ultra-high frequency due to for instance market microstructure noise or, on the other hand, whether they are simply a consequence of (sparse) sampling and vanish when moving to tick-by-tick event time. We argue the latter and provide anecdotal evidence, which suggests that price jumps identified at a low sampling frequency may in fact be bursts of volatility. This logic is also what underlies the family of jump tests developed by \citet*{ait-sahalia-jacod:09a}, namely, true jumps can only be identified by increasing the sampling frequency to the limit.

Consider April 24, 2007: an important day for AAPL with the SEC dropping tax fraud charges against CEO Steve Jobs\footnote{See http://www.sec.gov/news/press/2007/2007-70.htm.}, and a quarterly earnings announcement to be released on the following day. Panel A of Figure \ref{Figure:JumpVolatility} plots the intra-day price path at a conventional $5$ minute frequency, and we observe a temporary 2\% drop in share price with an instant recovery around 13:30. The widely used BPV jump test statistic is highly significant at $-6.8$ (RV is $41.2\%$ and BPV $32.4\%$) and unequivocally identifies this day to contain at least one jump. Contrast this with Panel B of Figure \ref{Figure:JumpVolatility}, where we plot the AAPL share price over the one hour window from 12:45 to 13:45 at ultra-high trade frequency. Now, the largest trade-by-trade move is only $\$0.13$ at a price of around $\$92$ or merely $14$ basis points. At the same time, more than one-third of the total daily volume was traded in this one hour window. On these data the MSRV and QRV$^\ast$ estimates closely agree at $30.4\%$ and $29.8\%$ respectively. Taken together, the hypothesis that there was a burst of volatility around 13:30, which is mistakenly identified as a jump at lower frequency, seems the more plausible one. This argument can be supported further by noting that the New York Stock Exchange, which actively makes markets in AAPL, contractually obliges their designated market makers to maintain a fair and orderly market which, as highlighted in NYSE Rule 104, implies ``\emph{the maintenance of price continuity with reasonable depth}''.\footnote{NYSE Rule 104 on the dealings and responsibilities of a designated market marker (DMM) states ``\emph{The function of a member acting as a DMM on the Floor of the Exchange includes the maintenance, in so far as reasonably practicable, of a fair and orderly market on the Exchange in the stocks in which he or she is so acting. The maintenance of a fair and orderly market implies the maintenance of price continuity with reasonable depth, to the extent possible consistent with the ability of participants to use reserve orders, and the minimizing of the effects of temporary disparity between supply and demand. In connection with the maintenance of a fair and orderly market, it is commonly desirable that a member acting as DMM engage to a reasonable degree under existing circumstances in dealings for the DMM's own account when lack of price continuity, lack of depth, or disparity between supply and demand exists or is reasonably to be anticipated.}'' (source: http://rules.nyse.com)} See \citet*{hasbrouck:07a} for further discussion on this.

The suggestion that price jumps are less common than previously thought does not in any way negate the importance of QRV or other jump-robust RV measures. Firstly, QRV is robust not only to jumps but also to outliers. As we have seen above, outliers do occur in high-frequency data and cannot always be filtered out perfectly based on trade qualifiers (this is particularly true for less recent non-US data). They are therefore a challenge to the econometrician and here the value of QRV is indisputable. Secondly, numerous scenarios remain where true price jumps can be observed, for instance over intra-day auctions, lunch breaks or when circuit breakers or volatility auctions are activated.\footnote{For example, the German trading platform Xetra has daily intra-day auction between 13:00 -- 13:17 for the various segments. The Tokyo (Hong Kong) stock exchange shuts down from 11:00 -- 12:30 (12:30 -- 14:30). The NYSE uses a circuit breaker where a 10\% intra-day move in the DJIA triggers a market wide trading halt of up to one hour.} Also, the market maker's precise obligation changes by exchange and regardless of this, its influence is not without limit so that with exceptional news releases or large market orders in an illiquid market, jumps can of course still occur. Finally, jumps can be commonplace when moving beyond equity markets governed by a market maker charged with the obligation to maintain price continuity. Consider, for example, electricity markets, where storage is costly, power plants can fail, and sudden unpredictable swings in demand for energy -- often for exogenous reasons -- can lead to substantial temporary price spikes \citep*[see, e.g.,][]{bessembinder-lemmon:02a,longstaff-wang:04a}.

\begin{figure}[t!]
\begin{center}
\caption{Jumps or burst of volatility?}
\label{Figure:JumpVolatility}
\begin{tabular}{cc}
\footnotesize{Panel A: price path at 5 minute frequency} & \footnotesize{Panel B: price path at trade frequency (12:45 -- 13:45)}\\
\includegraphics[height=8cm,width=0.45\textwidth]{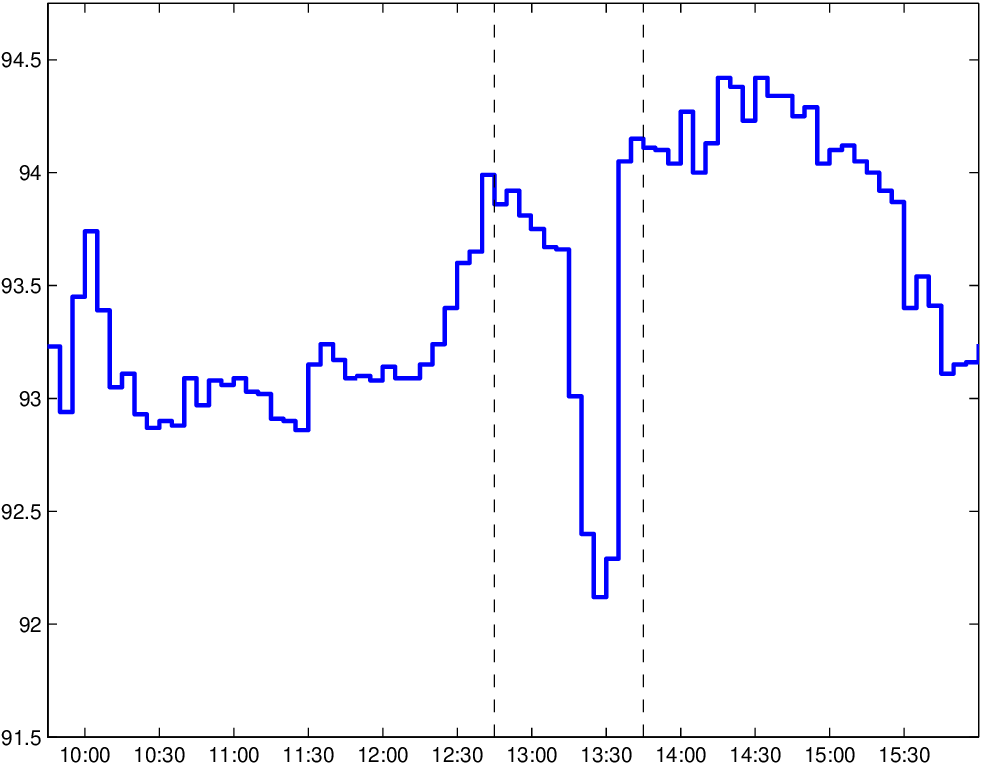} &
\includegraphics[height=8cm,width=0.45\textwidth]{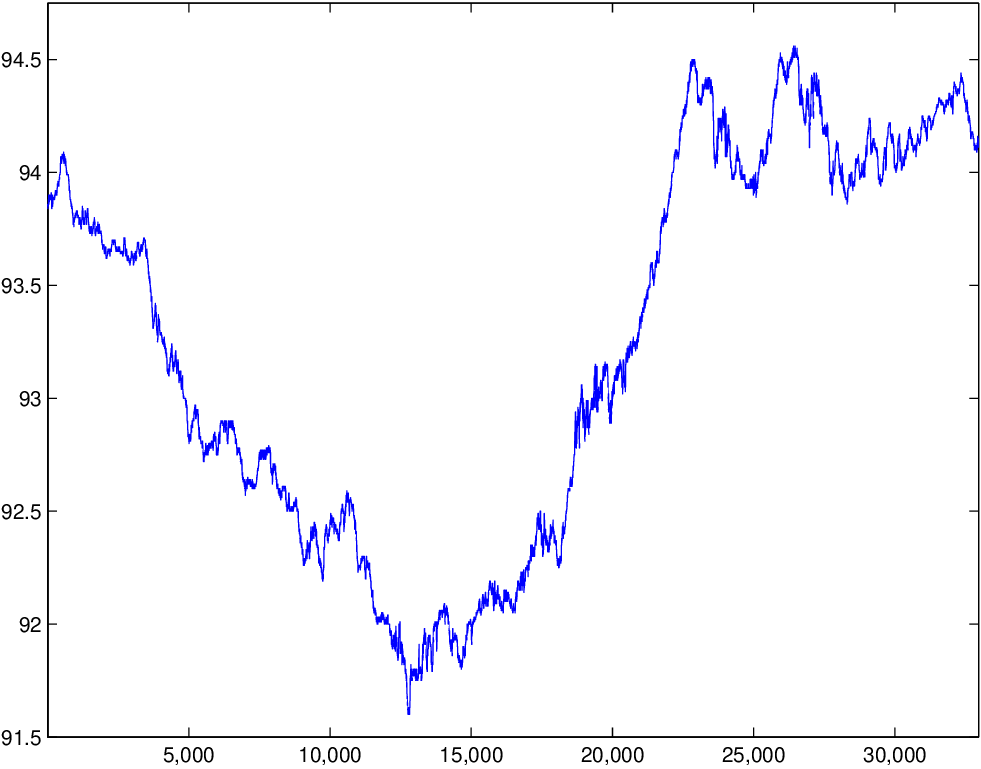}
\end{tabular}
\begin{footnotesize}
\parbox{0.98\textwidth}{\emph{Note}. Panel A plots the price of AAPL on April 24, 2007 at a 5 minute frequency (using last tick interpolation on trade data). Panel B plots the price at trade frequency from 12:45 to 13:45.}
\end{footnotesize}
\end{center}
\end{figure}

\section{Concluding remarks} \label{sec:conclusion}

In this paper we develop a new quantile-based realised variance measure that is consistent for the integrated variance and robust to jumps and outliers. A modified version, based on pre-averaged data, is also introduced and we show that in the presence of microstructure noise it retains consistency and attains the best possible convergence rate of $N^{-1/4}$. Importantly, our estimator is highly efficient making it the first estimator of integrated variance in the literature that is, at the same time, efficient, and robust to both jumps, outliers and microstructure noise. From a practical viewpoint, the estimator is easy to implement and is relatively insensitive to the particular choice of tuning parameters. Extensive simulations and empirical applications illustrate the excellent performance of our estimator.

The methodology outlined in this paper can be extended into various directions. For instance, it is possible to develop a joint distribution theory for RV and QRV allowing the construction of a formal jump test in the spirit of \citet*{barndorff-nielsen-shephard:06a}. Also, it is possible to modify QRV to produce jump and noise robust estimates of the integrated quarticity, a key quantity when making inference about integrated variance and testing for jumps. With these tools available, it may then be interesting to revisit some of the empirical work on non-parametric jump tests \citep*[e.g.][]{ait-sahalia-jacod:09a, ait-sahalia-jacod:09b, barndorff-nielsen-shephard:06a, christensen-podolskij:09a, fan-wang:07a, jiang-oomen:08a, lee-mykland:08a} and, inspired by our empirical findings, reevaluate the role that jumps play in financial equity price dynamics \citep*[e.g.][]{andersen-bollerslev-diebold:07a, eraker-johannes-polson:03a, huang-tauchen:05a}. Recent work by \citet*{ait-sahalia-jacod:09a}, \citet*{barndorff-nielsen-shephard-winkel:06a} and \citet*{woerner:06a} has shown that bi-power variation is not only robust to finite activity jumps (as considered in this paper) but also to certain infinite activity jump specifications. An investigation of the properties of QRV in such a scenario might be of interest and allow for further comparison to alternative jump robust estimators. Finally, in this paper we maintained the assumption of i.i.d. noise but this may be relaxed to allow for dependent noise \citep*[as studied by, for instance,][]{ait-sahalia-mykland-zhang:11b,barndorff-nielsen-hansen-lunde-shephard:08a,jacod-li-mykland-podolskij-vetter:09a}. All the above is well beyond the scope of the current paper and will be left for future research.

\clearpage

\appendix

\section{Summary statistics of trade data} \label{app:noisestats}

\begin{table}[!ht]
\setlength{\tabcolsep}{0.50cm}
\begin{center}
\caption{Summary statistics of trade data (2008) \label{Table:noisestats}}\medskip
\begin{tabular}{p{4cm}rrrcrrr}
\hline
            & \multicolumn{3}{c}{\# of observations $N$} && \multicolumn{3}{c}{noise ratio $\gamma$} \\
\cline{2-4}\cline{6-8}
universe  & Q5 & Q50 & Q95    && Q5 & Q50 & Q95   \\
\hline
\multicolumn{8}{l}{\emph{Panel A: US}}\\
S\&P600                   &       157 &       751 &     2,417 &&     0.10 &     0.34 &     0.73 \\
S\&P400                   &       604 &     1,749 &     4,710 &&     0.12 &     0.36 &     0.76 \\
S\&P500                   &     1,477 &     4,174 &    12,355 &&     0.14 &     0.37 &     0.93 \\
S\&P100                   &     2,945 &     7,338 &    20,707 &&     0.17 &     0.40 &     1.06 \\
DJ30                      &     4,701 &     9,562 &    23,686 &&     0.22 &     0.45 &     0.97 \\
\multicolumn{8}{l}{}\\
\multicolumn{8}{l}{\emph{Panel B: Europe}}\\
DJ Stoxx Small 200        &       158 &       772 &     2,225 &&     0.25 &     0.59 &     1.15 \\
DJ Stoxx Mid 200          &       352 &     1,419 &     3,689 &&     0.30 &     0.63 &     1.16 \\
DJ Stoxx Large 200        &       999 &     3,634 &    11,169 &&     0.34 &     0.66 &     1.28 \\
DJ Stoxx50                &     3,161 &     6,975 &    15,860 &&     0.40 &     0.71 &     1.40 \\
\multicolumn{8}{l}{}\\
\multicolumn{8}{l}{\emph{Panel C: Asia-pacific}}\\
S\&P ASX200               &       199 &       744 &     2,957 &&     0.30 &     0.74 &     1.75 \\
S\&P Topix 150            &       370 &     1,070 &     2,639 &&     0.34 &     1.03 &     3.59 \\
Hang Seng                 &       465 &     1,260 &     4,090 &&     0.39 &     0.88 &     2.26 \\
\multicolumn{8}{l}{}\\
\multicolumn{8}{l}{\emph{Panel D: Emerging markets (BRIC)}}\\
Ibovespa (Brazil)         &       261 &     1,130 &     5,617 &&     0.32 &     0.66 &     1.21 \\
DJ Titans 10 (Russia)     &       543 &     6,066 &    22,230 &&     0.58 &     1.03 &     1.29 \\
DJ BRIC 50 (India)        &       726 &     2,098 &     4,987 &&     0.16 &     0.49 &     0.90 \\
DJ BRIC 50 (China)        &     1,177 &     2,328 &     5,197 &&     0.60 &     1.17 &     2.96 \\
\hline
\end{tabular}
\medskip
\begin{footnotesize}
\parbox{0.98\textwidth}{\emph{Note}. This table reports the 5th, 50th, and 95th percentile of the number of observations (i.e. number of trades per day) $N$ and the noise ratio $\gamma^2 = \omega^2/(IV/N)$ computed across all names in each universe and all days over the period Jan 2, 2008 through Dec 31, 2008. The index constituents of January 2009 are used.}
\end{footnotesize}
\end{center}
\end{table}

To motivate the choice of parameters used for the simulations in Section \ref{sec:QRV*simulations}, Table \ref{Table:noisestats} reports summary statistics of the number of intra-day trade price observations $N$ and the level of microstructure noise as measured by the noise ratio $\gamma^2 = N\omega^2/IV$ \citep*[see][]{oomen:06a} for various stock universes. The data is taken from Reuters DataScope Tick History and covers the period January 2, 2008 through December 30, 2008. For the US and Europe, the selection covers small-caps (S\&P600, DJ Stoxx Small 200), mid-caps (S\&P400, DJ Stoxx Mid 200), large-caps (S\&P500, S\&P100, DJ Stoxx Large 200) and blue-chips (DJ30, DJ Stoxx50). For the Asia-pacific region and emerging markets, the universes cover large-caps only. The tick data are filtered on trade conditions, only trades from the primary exchange are included\footnote{For the US in particular, and Europe to lesser extent, this depresses the average number of trades per day as large fractions of volume are executed on competing exchanges and multi-lateral trading platforms.}, and trades with identical millisecond precision time-stamp are aggregated into one observation with a volume-weighted average trade price. All this is done to maximize the cleanliness of the data. To compute the noise ratio (for each stock and day in the sample), the microstructure noise variance $\omega^2$ is estimated as the negative of the first-order autocovariance of trade returns following \citet{oomen:06a}. The IV is proxied by an ad-hoc implementation of the realised (Bartlett) kernel of \citet*{barndorff-nielsen-hansen-lunde-shephard:08a} with a bandwidth parameter equal to five as suggested by \citet*{gatheral-oomen:10a}. The summary statistics in Table \ref{Table:noisestats} are computed across all names in the universe and days in the sample. This is justified as $N$ and $\gamma$ are reasonable stable over the sample period.

\clearpage

\section{The QRV with absolute returns: an alternative formulation} \label{sec:QRVabsolute}

In the main text, we defined QRV based on signed returns using symmetrized quantiles. Here we give an alternative formulation of the QRV based on the quantiles of absolute returns and show how these specifications are related. We base our exposition on the subsampled version of the QRV, as it can be linked to some other estimators proposed in the literature, but it should be clear that we could equally well work with a blocking QRV.

Define
\begin{equation*}
q_{i}^{abs,sub} (m,\lambda) = g^{2}_{ \lambda m} \left( \sqrt{N} |D_{i, m} X| \right), \qquad \text{for } \lambda \in [0,1).
\end{equation*}
As we are now dealing with absolute returns, there is no need to symmetrize $q_{i}^{abs,sub} (m,\lambda)$. We write
\begin{equation}
QRV_{N}^{abs,sub} \left( m, \overline{ \lambda}, \alpha \right) \equiv \alpha' QRV_{N}^{abs,sub} ( m, \overline{ \lambda}),
\end{equation}
with the $j$th element of $QRV_{N}^{abs,sub} ( m, \overline{ \lambda} )$ given by:
\begin{equation}
QRV_{N}^{abs,sub} (m,\lambda_{j}) =  \frac{1}{N - m} \sum_{i = 1}^{N - m} \frac{q_{i}^{abs,sub} (m,\lambda_{j})}{\nu_1^{abs}(m, \lambda_{j})}, \qquad \text{for } \lambda_{j} \in [0,1),
\end{equation}
and $j = 1, \ldots, k$, where $\nu_{r}^{abs} \left( m, \lambda \right)$ is defined as:
\begin{equation}\label{Eqn:nu-abs}
\nu_{r}^{abs} \left( m, \lambda \right) = \mathbb{E} \left[ \Bigl( |U| _{( \lambda m)}  \Bigr)^{2r} \right].
\end{equation}
The consistency and central limit theorems derived for QRV based on signed returns, extend directly to the case where we use absolute returns by replacing $\nu_r(m,\lambda)$ in Eq. \eqref{Eqn:nu(r,m)} with $\nu_{r}^{abs}(m,\lambda)$ above and $\nu_1(m,\lambda_i,\lambda_j)$ in Eq. \eqref{Eqn:nuij} with
\begin{equation}\label{Eqn:nuij-abs}
\nu_{1}^{abs} \left(m, \lambda_i, \lambda_j \right) = \mathbb{E} \left[ \Bigl( |U| _{( \lambda_i m)}  \Bigr)^{2} \Bigl( |U|_{( \lambda_j m)}  \Bigr)^{2} \right].
\end{equation}
The corresponding asymptotic constants for the $m\to\infty$ limit of Proposition \ref{prop:QRVmtoinfty} are now:
\begin{eqnarray*}
\nu_1^{abs}(\lambda) &\equiv& \lim_{m\rightarrow\infty}\nu_1^{abs} (m,\lambda) = d_{ \lambda}, \\
\nu_1^{abs}(\lambda_i,\lambda_j)&\equiv&\lim_{m\rightarrow\infty}\nu_1^{abs}(m,\lambda_i,\lambda_j) = d_{\lambda_i} d_{\lambda_j}, \\
\Theta^{abs}(\overline{\lambda})_{ij} &\equiv&\lim_{m\rightarrow\infty}\Theta^{abs}(m,\overline{\lambda})_{ij} = \frac{ \lambda_i (1 - \lambda_j ) }{ p \bigl( d_{ \lambda_i} \bigr) p \bigl( d_{ \lambda_j} \bigr)  d_{ \lambda_i} d_{ \lambda_j}},
\end{eqnarray*}
with $\lambda_i<\lambda_j$, where $d_{ \alpha}$ and $p$ denote the $\alpha$-quantile and density function of the $\chi_1^2$-distribution. With noise, the results presented for signed returns also extend to the formulation based absolute returns. The only substantive change would be the corresponding redefinition of Eq. \eqref{Eqn:def2function}, i.e.
\begin{equation*}
f_{m,l,x,u}^{abs} (\lambda_1, \lambda_2)= \text{\upshape{cov}} \Big(g^2_{\lambda_1 m} (|S|) , g^2_{\lambda_2 m} (|T|) \Big).
\end{equation*}
It is quite intuitive that $QRV_{N}^{sub}$ and $QRV_{N}^{abs, sub}$ are closely related, when the quantiles are chosen as $\lambda_i^{abs} = 2\lambda_i-1$. While it is hard to formalize this intuition, in unreported simulations we find that the performance of these estimators is indistinguishable. In the $m\to\infty$ limit, however, we can be more precise and prove their equivalence. Consider $\overline{ \lambda} = \lambda$ for $\lambda \in (1/2,1)$. Then it follows that $U_{(\lambda m)} \overset{p}{ \to} c_{ \lambda}$ as $m \to \infty$, where $c_{ \lambda}$ is the $\lambda$-quantile of $N(0,1)$. On the other hand, $|U|_{((2 \lambda - 1) m)} \overset{p}{ \to} c_{ \lambda}$, and since $c_{ \lambda} > 0$ for $\lambda \in (1 / 2, 1)$, $|U_{( \lambda m)}|$ and $|U|_{((2 \lambda - 1) m)}$ have the same asymptotic behavior. By symmetry of the normal distribution, it also holds that $| U_{(m-\lambda m+1)} | \overset{p}{ \to} c_{ \lambda}$, and thus $(|U_{( \lambda m)}| + | U_{(m-\lambda m+1)} |) / 2$ and $|U|_{((2 \lambda - 1) m)}$ also have the same asymptotic behaviour. It now follows that the expressions for the asymptotic variance of the two estimators are identical for $m \rightarrow \infty$, when the quantiles are chosen this way.

There are, however, two important ways in which the QRV estimator based on absolute returns is distinguished from its counterpart on signed returns. First, by a suitable choice of $\lambda^{abs}$ it is possible to discard only a single observation per block (e.g. $m=100$ with $\lambda^{abs} = 0.99$). This allows for the use of observations in the extreme tail, at the cost of losing robustness to outliers. Second, the estimator based on absolute returns nests the MinRV and MedRV proposed by \cite{andersen-dobrev-schaumburg:08a} as special cases. Specifically, with $m=2$ and $\lambda^{abs} = \frac{1}{2}$ we have MinRV and with $m=3$ and $\lambda^{abs} = \frac{2}{3}$ we have MedRV.

\clearpage

\section{Proofs} \label{app:proofs}
In this part of the paper, we state the proofs of the theorems given in the main text. Throughout, we use the approximation
\begin{equation*}
\Delta_{i}^{N} X \approx \sigma_{ \frac{i - 1}{N}} \Delta_{i}^{N} W.
\end{equation*}
Thus, to prove our asymptotic results we first replace $\Delta_{i}^{N} X$ with $\sigma_{ \frac{i - 1}{N}} \Delta_{i}^{N} W$ and then show that the error caused by this approximation is asymptotically negligible.

Let us fix some notations. We set
\begin{equation*}
\beta_{i}^{n} = \sqrt{N} \left( \sigma_{ \frac{i - 1}{n}} \Delta_{k}^{N} W \right)_{(i - 1) m + 1 \leq k \leq im},
\end{equation*}
and define
\begin{equation*}
w_{i}^{(n, m)} ( \lambda) = g_{ \lambda m}^{2} ( \beta_{i}^{n} ) + g_{m - \lambda m + 1}^{2} ( \beta_{i}^{n} ).
\end{equation*}
Before we start to prove the main results, we state a simple Lemma.

\begin{lemma}
\label{Lem:gdiff} The function $g_{k}$ defined in \eqref{Eqn:g} has the following properties:
\begin{enumerate}
   \item $g_{k}$ is continuous.
   \item $g_{k}$ is differentiable on the set $\left\{x \in \mathbb{R}^{m} \mid x_{i} \not = x_{j}, 1 \leq i < j \leq m \right\}$, that is
   \begin{equation*}
   \frac{1}{ \epsilon} \left[ g_{k} \left( x + \epsilon y \right) - g_{k} \left( x \right) \right] \to y_{k^{*}} \qquad \epsilon \searrow 0,
   \end{equation*}
   where $y \in \mathbb{R}^m$ and
   \begin{equation*}
   k* = i \text{ with } x_{i} = x_{ \left( k \right)}.
   \end{equation*}
\end{enumerate}
\end{lemma}
In the following we assume without loss of generality that $a$, $\sigma$, $a'$, $\sigma'$ and $v'$ are bounded \citep*[for details see e.g. Section 3 in][]{barndorff-nielsen-graversen-jacod-podolskij-shephard:06a}. Moreover, the constants used in the proofs will all be denoted by $C$.

\begin{proof}[Proof of Theorem \ref{Thm:QRVconsistency}]
Here we show the consistency of $QRV_N(m,\overline{\lambda},\alpha)$. Since $\sum_{j=1}^k \alpha_j=1$ it is sufficient to prove Theorem \ref{Thm:QRVconsistency} for $\overline{\lambda} = \lambda \in (1/2, 1)$ and $\alpha =1$ (i.e. $k=1$). In this case we use the notation $QRV_N(m,\lambda)=QRV_N(m,\overline{\lambda},\alpha)$. First, we define:
\begin{align*}
\xi_{i}^{n} &= \nu_{1, m}^{-1} \left( \lambda \right) w_{i}^{ \left( n, m \right)} \left( \lambda \right), \\[0.25cm]
U_{n} &= \frac{1}{n} \sum_{i = 1}^{n} \xi_{i}^{n}.
\end{align*}
Note that:
\begin{equation*}
\E \left[ \xi_{i}^{n} \mid \mathcal{F}_{ \frac{i - 1}{n}} \right] = \sigma^{2}_{ \frac{i - 1}{n}},
\end{equation*}
so
\begin{equation}
\label{Eqn:ApI} \frac{1}{n} \sum_{i = 1}^{n} \E \left[ \xi_{i}^{n} \mid \mathcal{F}_{ \frac{i - 1}{n}} \right] \overset{p}{ \to} \int_{0}^{1} \sigma_{u}^{2} \text{d}u.
\end{equation}
Now, by setting
\begin{equation}
\label{Eqn:eta} \eta_{i}^{n} = \xi_{i}^{n} - \E \left[ \xi_{i}^{n} \mid \mathcal{F}_{ \frac{i - 1}{n}} \right],
\end{equation}
we get:
\begin{equation*}
\E \left[ | \eta_{i}^{n} |^{2} \mid \mathcal{F}_{ \frac{i - 1}{n}} \right] = \frac{ \nu_{2, m} \left( \lambda \right) - \nu_{1, m}^2 \left( \lambda \right)}{ \nu_{1, m}^{2} \left( \lambda \right)} \sigma^{4}_{ \frac{i - 1}{n}}.
\end{equation*}
Therefore,
\begin{equation*}
\frac{1}{n^{2}} \sum_{i = 1}^{n} \E \left[ | \eta_{i}^{n} |^{2} \mid \mathcal{F}_{ \frac{i - 1}{n}} \right] \overset{p}{ \to} 0.
\end{equation*}
Hence, the assertion $U_{n} \overset{p}{ \to} \int_{0}^{1} \sigma_{u}^{2} \text{d}u$ follows directly from \eqref{Eqn:ApI}. Now, we are left to prove that
\begin{equation}
QRV_{N}(m,\lambda) - U_{n} \overset{p}{ \to} 0.
\end{equation}
Note that
\begin{equation*}
QRV_{N}(m,\lambda) - U_{n} = \frac{ \nu_{1, m}^{-1} ( \lambda)}{n} \sum_{i = 1}^{n} \zeta_{i}^{n},
\end{equation*}
where
\begin{equation*}
\zeta_{i}^{n} = q_{i}^{ \left( n, m \right)} \left( \lambda \right) - w_{i}^{ \left( n, m \right)}
\left( \lambda \right).
\end{equation*}
Now, we use the decomposition
\begin{equation*}
\zeta_{i}^{n} = \zeta_{i}^{n} \left( 1 \right) + \zeta_{i}^{n} \left( 2 \right),
\end{equation*}
where $\zeta_{i}^{n} \left( k \right)$, $k = 1, 2$, are given by
\begin{align}
\label{Eqn:zeta1} \zeta_{i}^{n} \left( 1 \right) &= g_{ \lambda m}^{2} \left( \sqrt{N} D_i^m X \right) - g_{ \lambda m}^{2} \left( \beta_{i}^{n} \right), \\[0.25cm]
\label{Eqn:zeta2} \zeta_{i}^{n} \left( 2 \right) &= g_{m - \lambda m + 1}^{2} \left( \sqrt{N} D_i^m X \right) - g_{m - \lambda m + 1}^{2} \left( \beta_{i}^{n} \right).
\end{align}
In the following, we show that
\begin{equation}
\label{Eqn:prob1} \frac{ \nu_{1, m}^{-1} ( \lambda)}{n} \sum_{i = 1}^{n} \zeta_{i}^{n} \left( 1 \right) \overset{p}{ \to} 0.
\end{equation}
The corresponding result for $\zeta_{i}^{n} \left( 2 \right)$ can be proven similarly. We begin with the following Lemma.
\begin{lemma}
\label{Lem:lemsim} For $x \in \mathbb{R}^{m}$, we define a norm $|| x || = \sum_{k = 1}^{m} \left| x_{k} \right|$. Then we have
\begin{equation} \label{Eqn:app}
\frac{1}{n} \sum_{i = 1}^{n} \E \left[ || \sqrt{N} D_i^m X - \beta_{i}^{n} ||^{2} \right] \to 0.
\end{equation}
\end{lemma}
\begin{proof}[Proof of Lemma \ref{Lem:lemsim}]
The boundedness of the drift function $a$ and $|| x ||^{2} \leq m \sum_{k = 1}^{m} \left| x_{k} \right|^{2}$ yield
\begin{align*}
\E \left[ || \sqrt{N} D_i^m X - \beta_{i}^{n} ||^{2} \right] &\leq m C \left( \frac{1}{n} + N \sum_{k = 1}^{m} \E \left[ \int_{ \frac{i - 1}{n} + \frac{k - 1}{N}}^{ \frac{i - 1}{n} + \frac{k}{N}} | \sigma_{u} - \sigma _{ \frac{i - 1}{n}} |^{2} \text{d}u \right] \right) \\[0.25cm]
&= m C \left( \frac{1}{n} + N \E \left[ \int_{ \frac{i - 1} {n}}^{ \frac{i}{n}} | \sigma_{u} - \sigma_{ \frac{i - 1}{n}} |^{2} \text{d}u \right] \right).
\end{align*}
Hence, the left side of \eqref{Eqn:app} is smaller than
\begin{equation*}
m C \left( \frac{1}{n} + m \E \int_{0}^{1}  | \sigma_{u} - \sigma_{ \left[ nu \right] / n} |^{2}  \text{d}u \right).
\end{equation*}
Because $\sigma$ is bounded and c\`{a}dl\`{a}g (and so $\sigma_{u} - \sigma_{ \left[ nu \right] / n} \overset{p}{ \to} 0$ with an exception of countably many $u \in [0,1]$), the assertion of Lemma \ref{Lem:lemsim} is proved by Lebesgue's theorem. \qed
\end{proof}

\noindent Next, we set
\begin{equation*}
m_{A} \left( \epsilon \right) = \sup\{ | g_{ \lambda m}^{2} \left( x \right) - g_{ \lambda m}^{2} \left( y \right) | : || x - y || \leq \epsilon,|| x || \leq A \}.
\end{equation*}
For all $\epsilon \in \left( 0, 1 \right]$ and $A > 1$, we obtain the estimate
\begin{align*}
\zeta_{i}^{n} \left( 1 \right) &\leq C \Bigl( m_{A} \left( \epsilon \right) + A^{2} 1_{ \{ || \sqrt{N} D_i^m X - \beta_{i}^{n} || > \epsilon \}} \\[0.25cm]
&+ \left( g_{ \lambda m}^{2} \left( \sqrt{N} D_i^m X \right) + g_{ \lambda m}^{2} \left( \beta_{i}^{n} \right) \right) \left( 1_{ \{ || \sqrt{N} D_i^m X || > A \}} + 1_{ \{ || \beta_{i}^{n} || > A \}} \right) \Bigr) \\[0.25cm]
&\leq C \left( m_{A} \left( \epsilon \right) + \frac{ A^{2} || \sqrt{N} D_i^m X - \beta_{i}^{n} ||^{2}}{ \epsilon^{2}} + \frac{ \left( || \sqrt{N} D_i^m X || + || \beta_{i}^{n} || \right)^{3}}{A} \right).
\end{align*}
The boundedness of $a$ and $\sigma$ and Burkholder's inequality imply that
\begin{equation} \label{incrineq}
\E \left[ || \sqrt{N} D_i^m X ||^{p} \right] + \E \left[ || \beta_{i}^{n} ||^{p} \right] \leq C_{p},
\end{equation}
for all $p \geq 0$. This means that
\begin{equation*}
\frac{1}{n} \sum_{i = 1}^{n} \E \left[ | \zeta_{i}^{n} \left( 1 \right) | \right] \leq C \left( m_{A} \left( \epsilon \right) + \frac{1}{A} + \frac{A^{2}}{n \epsilon^{2}} \sum_{i = 1}^{n} \E \left[ || \sqrt{N} D_i^m X - \beta_{i}^{n} ||^{2} \right] \right).
\end{equation*}
Because $m_{A} \left( \epsilon \right) \to 0$ as $\epsilon \to 0$ for every $A$, we obtain by Lemma
\ref{Lem:lemsim}
\begin{equation}
\label{Eqn:pn1} \frac{1}{n} \sum_{i = 1}^{n} \E \left[ | \zeta_{i}^{n} \left( 1 \right) | \right] \to 0,
\end{equation}
by first choosing $A$ large, then $\epsilon $ small and finally $n$ large. Then \eqref{Eqn:prob1} holds and the proof of Theorem \ref{Thm:QRVconsistency} is complete. \qed \\[-0.50cm]
\end{proof}

\begin{proof}[Proof of Theorem \ref{Thm:QRVcentrallimit}]
We proceed with a three-stage proof of Theorem \ref{Thm:QRVcentrallimit}. First, we prove a CLT for the sequence $\bar{U}_{n} = (\bar{U}_{n}^1, \ldots, \bar{U}_{n}^k)$ with
\begin{equation*}
\bar{U}_{n}^j = \sqrt{ \frac{m}{n}} \sum_{i = 1}^{n} \eta_{i}^{n}(j), \qquad \eta_{i}^{n}(j) = \nu_{1, m}^{-1} \left( \lambda_j \right) \left(
w_{i}^{ \left( n, m \right)} \left( \lambda_j \right) - \E \left[ w_{i}^{ \left( n, m \right)} \left( \lambda_j \right) \mid \mathcal{F}_{ \frac{i -1}{n}} \right]\right).
\end{equation*}
More precisely, we show that $\bar{U}_{n} \overset{d_{s}}{\to} MN \left( 0, \Theta (m,\overline{\lambda}) IQ \right)$. Next, we again consider the case $k=1$, $\overline{\lambda} = \lambda$. The second step is to define a new sequence:
\begin{equation*}
U_{n}' = \nu_{1, m}^{-1} ( \lambda) \sqrt{ \frac{m}{n}} \sum_{i = 1}^{n} \left( q_{i} \left(m, \lambda \right) - \E \left[ q_{i} \left(m, \lambda \right)
\mid \mathcal{F}_{ \frac{i - 1}{n}} \right] \right),
\end{equation*}
and show the result
\begin{equation*}
U_{n}' - \bar{U}_{n} \overset{p}{ \to} 0.
\end{equation*}
Finally, in part III, we prove the convergences:
\begin{align*}
& \sqrt{ \frac{m}{n}} \sum_{i = 1}^{n} \left( \nu_{1, m}^{-1} ( \lambda) \E \left[ q_{i} \left(m, \lambda \right) \mid \mathcal{F}_{ \frac{i - 1}{n}} \right] - \E \left[ \xi_{i}^{n} \mid \mathcal{F}_{ \frac{i - 1}{n}} \right] \right) \overset{p}{ \to} 0, \\[0.25cm]
& \sqrt{ \frac{m}{n}} \left( \sum_{i = 1}^{n} \E \left[ \xi_{i}^{n} \mid \mathcal{F}_{ \frac{i - 1}{n}} \right] - \int_{0}^{1} \sigma_{u}^{2} \text{d}u \right) \overset{p}{ \to} 0.
\end{align*}
Clearly, the afore-mentioned steps imply the assertion of Theorem \ref{Thm:QRVcentrallimit}.
\begin{proof}[Proof of part I]
Notice that:
\begin{equation*}
\frac{m}{n} \sum_{i = 1}^{n} \E \left[  \eta_{i}^{n} (j) \eta_{i}^{n} (l) \mid \mathcal{F}_{ \frac{i - 1}{n}} \right] \overset{p}{ \to} m \frac{ \nu_{1} \left( m, \lambda_j, \lambda_l \right) - \nu_{1} \left(m, \lambda_j \right) \nu_{1}\left(m, \lambda_l \right)}{ \nu_{1} \left(m, \lambda_j \right) \nu_{1}\left(m, \lambda_l \right)} \int_{0}^{1} \sigma_{u}^{4} \text{d}u.
\end{equation*}
for any $1\leq j,l\leq k$. Moreover, since $W \overset{d}{=} -W$ and $w_{i}^{ \left( n, m \right)} \left( \lambda_j \right)$ is an even functional in $W$, we have
\begin{equation*}
\E \left[ \eta_{i}^{n} (j) \Delta_{i}^{n} W \mid \mathcal{F}_{ \frac{i - 1}{n}} \right] = 0, \qquad 1\leq j\leq k.
\end{equation*}
Next, let $H = \left( H_{t} \right)_{t \in \left[ 0, 1 \right]}$ be a bounded martingale on $\bigl( \Omega, \mathcal{F}, \left( \mathcal{F}_{t} \right)_{t \geq 0}, \mathbb{P} \bigr)$, which is orthogonal to $W$ (i.e., with quadratic covariation $\left[ W, H \right] = 0$, almost surely). By the Clark's representation theorem (see, e.g., Karatzas and Shreve, 1998, Appendix E) we obtain
\begin{equation*}
\eta_{i}^{n} (j) = \sigma_{\frac{i-1}{n}}^2 \int_{\frac{i-1}{n}}^{\frac{i}{n}} F_s^n(j) dW_s
\end{equation*}
for some predictable process $F_s^n(j)$. Then,
\begin{equation*}
\sqrt{ \frac{m}{n}} \sum_{i = 1}^{n} \E \left[ \eta_{i}^{n} \Delta_{i}^{n} H \mid \mathcal{F}_{ \frac{i - 1}{n}} \right] = 0,
\end{equation*}
because $\left[ W, H \right] = 0$. Finally, stable convergence in law follows by Theorem IX 7.28 in \citet*{jacod-shiryaev:03a}:
\begin{equation*}
\bar{U}_{n} \overset{d_{s}}{\to} MN \left( 0, \Theta (m,\overline{\lambda}) IQ \right),
\end{equation*}
which completes the proof of part I. \qedi
\end{proof}
\begin{proof}[Proof of part II]
We begin by setting
\begin{equation*}
\delta_{i}^{n} = \nu_{1, m}^{-1} ( \lambda) \sqrt{ \frac{m}{n}} \left( q_{i} \left(m, \lambda \right) - w_{i}^{ \left( n, m \right)} \left( \lambda \right) \right),
\end{equation*}
and obtain the identity:
\begin{equation*}
U_{n}' - \bar{U}_{n} = \sum_{i = 1}^{n} \left( \delta _{i}^{n} - \E \left[ \delta _{i}^{n} \mid \mathcal{F}_{ \frac{i - 1}{n}} \right] \right).
\end{equation*}
To complete the second step, it suffices that
\begin{equation*}
\sum_{i = 1}^{n} \E \left[ | \delta_{i}^{n} |^{2} \right] \to 0.
\end{equation*}
We omit the proof of this result, as it be shown by using exactly the same methods behind the proof of the convergence in \eqref {Eqn:pn1} in Theorem \ref{Thm:QRVconsistency}. \qedi
\end{proof}
\begin{proof}[Proof of part III]
It holds that:
\begin{equation*}
\sqrt{ \frac{m}{n}} \left( \sum_{i = 1}^{n} \E \left[ \xi_{i}^{n} \mid \mathcal{F}_{ \frac{i - 1}{n}} \right] - \int_{0}^{1} \sigma_{u}^{2} \text{d}u \right) = \sqrt{mn} \sum_{i = 1}^{n} \int_{ \frac{i - 1}{n}}^{ \frac{i}{n}} \left( \sigma_{ \frac{i - 1}{n}}^{2} - \sigma_{u}^{2} \right) \text{d}u.
\end{equation*}
Exploiting the results of Section 8 (Part 2) in \citet*{barndorff-nielsen-graversen-jacod-podolskij-shephard:06a} (recall that $m$ is a fixed
number), we find that, under condition (V), the convergence
\begin{equation} \label{convsigma}
\sqrt{ \frac{m}{n}} \left( \sum_{i = 1}^{n} \E \left[ \xi_{i}^{n} \mid \mathcal{F}_{ \frac{i - 1}{n}} \right] - \int_{0}^{1} \sigma_{u}^{2} \text{d}u \right) \overset{p}{ \to} 0,
\end{equation}
holds. Now, we prove the first convergence of part III stated above. Under condition (V), we introduce the decomposition
\begin{equation*}
\sqrt{N} D_i^m X - \beta_{i}^{n} = \mu_{i}^{n} \left( 1 \right) + \mu_{i}^{n} \left( 2 \right),
\end{equation*}
where $\mu_{i}^{n} \left( 1 \right)$ and $\mu_{i}^{n} \left( 2 \right)$ are $m$-dimensional vectors with components defined by
\begin{align}
\left( \mu_{i}^{n} \left( 1 \right) \right)_{k} &= \sqrt{N} \int_{ \frac{i - 1}{n} + \frac{k - 1}{N}}^{ \frac{i - 1}{n} + \frac{k}{N}} \left( a_{u} - a_{ \frac{i - 1}{n}} \right) \text{d}u + \sqrt{N} \int_{ \frac{i - 1}{n} + \frac{k - 1}{N}}^{ \frac{i - 1}{n} + \frac{k}{N}} \Bigl( \int_{ \frac{i - 1}{n} + \frac{k - 1}{N}}^{u} a_{s}' \text{d}s \nonumber \\[0.25cm]
&+ \int_{ \frac{i - 1}{n} + \frac{k - 1}{N}}^{u} \left( \sigma_{s}' - \sigma_{ \frac{i - 1}{n}}' \right) \text{d}W_{s} + \int_{ \frac{i - 1}{n} + \frac{k - 1}{N}}^{u} \left( v_{s}' - v_{ \frac{i - 1}{n}}' \right) \text{d}B_{s}' \Bigr) \text{d}W_{u}, \nonumber \\[0.25cm]
\left( \mu_{i}^{n} \left( 2 \right) \right)_{k} &= \sqrt{N} \Big( \frac{1}{N} a_{ \frac{i - 1}{n}} + \sigma_{ \frac{i - 1}{n}}' \int_{ \frac{i - 1}{n} + \frac{k - 1}{N}}^{ \frac{i - 1}{n} + \frac{k}{N}} \left( W_{u} - W_{ \frac{i - 1}{n}} \right) \text{d}W_{u} \nonumber \\[0.25cm]
\label{Eqn:mu} &+ v_{ \frac{i - 1}{n}}' \int_{ \frac{i - 1}{n} + \frac{k}{N}}^{ \frac{i - 1}{n} + \frac{k}{N}} \left( B_{u}' - B_{ \frac{i - 1}{n}}' \right) \text{d}W_{u} \Big),
\end{align}
for $k = 1, \ldots, m$. Moreover, we decompose
\begin{equation*}
\nu_{1, m}^{-1} ( \lambda) q_{i} \left(m, \lambda \right) - \xi_{i}^{n} = \theta_{i}^{n} \left( 1 \right) + \theta_{i}^{n} \left( 2 \right),
\end{equation*}
where
\begin{align}
\theta_{i}^{n} \left( 1 \right) &= \nu_{1, m}^{-1} ( \lambda) \Big( g_{ \lambda m}^{2} \left( \sqrt{N} D_i^m X \right) - g_{ \lambda m}^{2} \left( \beta_{i}^{n} + \mu_{i}^{n} \left( 2 \right) \right)
\nonumber \\[0.25cm]
&+ g_{m - \lambda m + 1}^{2} \left( \sqrt{N} D_i^m X \right) - g_{m - \lambda m + 1}^{2} \left( \beta_{i}^{n} + \mu_{i}^{n} \left( 2 \right) \right) \Big), \nonumber \\[0.25cm]
\theta_{i}^{n} \left( 2 \right) &= \nu_{1, m}^{-1} ( \lambda) \Big( g_{ \lambda m}^{2} \left( \beta_{i}^{n} + \mu_{i}^{n} \left( 2 \right) \right) - g_{ \lambda m}^{2} \left( \beta_{i}^{n} \right) \nonumber \\[0.25cm]
\label{Eqn:theta} &+ g_{m - \lambda m + 1}^{2} \left( \beta_{i}^{n} + \mu_{i}^{n} \left( 2 \right) \right) - g_{m - \lambda m + 1}^{2} \left( \beta_{i}^{n} \right) \Big).
\end{align}
Using the same methods as for the proof of \eqref{Eqn:pn1} in Theorem \ref{Thm:QRVconsistency}, we obtain
\begin{equation*}
\sqrt{ \frac{m}{n}} \sum_{i = 1}^{n} \E \left[ | \theta_{i}^{n} \left( 1 \right) | \right] \to 0,
\end{equation*}
which implies
\begin{equation*}
\sqrt{ \frac{m}{n}} \sum_{i = 1}^{n} \E \left[ \theta_{i}^{n} \left( 1 \right) \mid \mathcal{F}_{ \frac{i - 1}{n}} \right] \overset{p}{ \to} 0.
\end{equation*}
Thus, we are left to prove that
\begin{equation*}
\sqrt{ \frac{m}{n}} \sum_{i = 1}^{n} \E \left[ \theta_{i}^{n} \left( 2 \right) \mid \mathcal{F}_{ \frac{i - 1}{n}} \right] \overset{p}{ \to} 0.
\end{equation*}
Now we apply Lemma \ref{Lem:gdiff} to the term $\theta_{i}^{n} \left( 2 \right)$ with $\epsilon = N^{- 1 / 2}$ and
\begin{align*}
x &= \beta_{i}^{n}, \\[0.25cm]
y &= \sqrt{N} \mu_{i}^{n} \left( 2 \right).
\end{align*}
Notice that as $\sigma$ does not vanish, we have $\left( \beta_{i}^{n} \right)_{k} \not = \left( \beta_{i}^{n} \right)_{l}$ for all $1 \leq k < l \leq m$ almost surely, and consequently the assumptions of Lemma \ref{Lem:gdiff} are satisfied. Finally, we obtain the stochastic expansion
\begin{align*}
\sqrt{ \frac{m}{n}} \sum_{i = 1}^{n} \E \left[ \theta_{i}^{n} \left( 2 \right) \mid \mathcal{F}_{ \frac{i - 1}{n}} \right] &= 2 \nu_{1, m}^{-1} ( \lambda) \sqrt{ \frac{m}{n}} \sum_{i = 1}^{n} \E \Big[ g_{ \lambda m} \left( \beta_{i}^{n} \right) \Big( \mu_{i}^{n} \left( 2 \right) \Big)_{ \left( \lambda m
\right)*} \\[0.25cm]
&+ g_{m - \lambda m + 1} \left( \beta_{i}^{n} \right) \Big( \mu_{i}^{n} \left( 2 \right) \Big)_{ \left( m - \lambda m + 1 \right)*} \mid \mathcal{F}_{ \frac{i - 1}{n}} \Big] + o_{p} \left( 1 \right),
\end{align*}
where we recall that $\left( \lambda m \right)*$ is defined by
\begin{equation*}
\left( \lambda m \right)* = k \quad \text{with} \quad \left( \beta_{i}^{n} \right)_{k} = \left( \beta_{i}^{n} \right)_{ \left( \lambda m \right)}.
\end{equation*}
Now $\left( W, V \right) \overset{d}{=} - \left( W, V \right)$ and
\begin{equation*}
g_{ \lambda m} \left( \beta_{i}^{n} \right) \Big( \mu_{i}^{n} \left( 2 \right) \Big)_{ \left( \lambda m \right)*} + g_{m - \lambda m + 1} \left( \beta_{i}^{n} \right) \Big( \mu_{i}^{n} \left( 2 \right) \Big)_{ \left( m - \lambda m + 1 \right)*}
\end{equation*}
is odd in $\left( W, V \right)$, which implies that
\begin{equation} \label{nullidentity}
\E \Big[ g_{ \lambda m} \left( \beta_{i}^{n} \right) \Big( \mu_{i}^{n} \left( 2 \right) \Big)_{ \left( \lambda m \right)*} + g_{m - \lambda m + 1} \left( \beta_{i}^{n} \right) \Big( \mu_{i}^{n} \left( 2 \right) \Big)_{ \left( m - \lambda m + 1 \right)*} \mid \mathcal{F}_{ \frac{i - 1}{n}} \Big] = 0.
\end{equation}
Consequently,
\begin{equation*}
\sqrt{ \frac{m}{n}} \sum_{i = 1}^{n} \E \left[ \theta_{i}^{n} \left( 2 \right) \mid \mathcal{F}_{ \frac{i - 1}{n}} \right] \overset{p}{ \to} 0,
\end{equation*}
which completes the proof of part III and, hence, Theorem \ref{Thm:QRVcentrallimit} holds. \qed
\end{proof}
\end{proof}

\begin{proof}[Proof of Proposition \ref{prop:QRVjumprobust}]

Here we show that the CLT in Theorem \ref{Thm:QRVcentrallimit} is robust to the presence of finite activity jumps. We consider a process $X$ of the type \eqref{Eqn:Xjump}. Let $X^c$ denote the continuous part of $X$ and set $J_n= \{1\leq i\leq n| ~\mbox{X jumps on the interval } [\frac{i-1}{n}, \frac{i}{n}] \}$. Again it is sufficient to assume that $k=1$ and $\overline \lambda =\lambda \in (1/2, 1)$. Now we use the decomposition
\begin{equation*}
\sqrt N \left( QRV_{N}(m,\lambda) - IV \right) = \sqrt N \left( \frac{1}{\nu_1(m,\lambda)} \frac{m}{  N} \sum_{i \in J_n^c} q_{i} (m,\lambda) - IV \right) + \frac{1}{\nu_1(m,\lambda)} \frac{m}{ \sqrt N} \sum_{i \in J_n} q_{i} (m,\lambda).
\end{equation*}
Recall that $X$ has only finitely many jumps on compact intervals (a.s.), so the second sum on the right-hand side  is finite. Therefore, the first term on the right-hand side converges stably to the limit described in Theorem \ref{Thm:QRVcentrallimit} (since this is a statistic based on $X^c$). We need to show that
\begin{equation*}
\frac{1}{ \sqrt N} q_{i} (m,\lambda)\overset{p}{ \to} 0
\end{equation*}
for any $i \in J_n$. It is well-known that the probability of having two jumps within the interval $[\frac{i-1}{n}, \frac{i}{n}]$ is negligible, so we assume that $X$ jumps one time at $s\in [\frac{i-1}{n}, \frac{i}{n}]$ for some $i \in J_n$. Recall that $q_{i} (m,\lambda)=g^{2}_{ \lambda m} ( \sqrt{N} D_i^m X ) + g^{2}_{m - \lambda m + 1} ( \sqrt{N} D_i^m X )$. Due to \eqref{incrineq} we have that ($\Delta X_s= X_s - X_{s-}$)
\begin{equation*}
\frac{1}{ \sqrt N} q_{i} (m,\lambda)1_{\{ |\Delta X_s| \geq 2 ||\sqrt{N} D_i^m X^c||\}}\leq \frac{C} { \sqrt N} ||\sqrt{N} D_i^m X^c||^2 \overset{p}{ \to} 0,
\end{equation*}
since on $\{ |\Delta X_s| \geq 2 ||\sqrt{N} D_i^m X^c||\}$ the jump is contained in $(D_i^m X)_{(0)}$ or $(D_i^m X)_{(m)}$, which both do not appear in $q_{i} (m,\lambda)$ because $\lambda \in (1/2,1)$. On the other hand, we have that
\begin{equation*}
\frac{1}{ \sqrt N} q_{i} (m,\lambda)1_{\{ |\Delta X_s| < 2 ||\sqrt{N} D_i^m X^c||\}}\leq \frac{C \left(||\sqrt{N} D_i^m X^c||^2 + |\Delta X_s|^2\right)
|| D_i^m X^c||^2} { \sqrt N |\Delta X_s|^2} \overset{p}{ \to} 0
\end{equation*}
again due to \eqref{incrineq}. This completes the proof of Proposition \ref{prop:QRVjumprobust}. \qed
\end{proof}

\begin{proof}[Proof of Theorem \ref{Thm:QRVsubcentrallimit}]

In the following we assume that the process $X$ is continuous (the robustness to finite activity jumps is shown as in Theorem \ref{Thm:QRVcentrallimit}). Here we show the CLT for the subsampled statistic $QRV_N^{sub}(m,\overline{\lambda},\alpha)$. Note that the summands $q_i^{sub}(m, \lambda_j)$ ($1\leq j\leq k$) in the definition of $QRV_N^{sub}(m,\overline{\lambda},\alpha)$ are m-dependent (i.e. $q_i^{sub}(m, \lambda_j)$ and $q_l^{sub}(m, \lambda_j)$ are correlated for $|i-l|<m$), which makes the proof more complicated.

We will apply ``big blocks \& small blocks''-technique to break this dependence. More precisely, we will build big blocks of size $pm$, which will be separated by a small block of size $m$. This procedure ensures the (conditional) independence of big blocks, whereas the small blocks become
asymptotically negligible when we later let $p$ converge to infinity.

For this purpose we require some additional notations. First, set
$$a_i(p)=i(p+1)m,~\qquad b_i(p)=i(p+1)m + pm~,$$
and let $A_i(p)$ denote the set of integers $l$ with $a_i(p)\leq l<b_i(p)$ and $B_i(p)$  the set of integers $l$ with $b_i(p)\leq l<a_{i+1}(p)$. Furthermore, let $j_N(p)$ denote the largest integer $j$ with $b_j(p)\leq N$. Notice that $j_N(p)=O( N / p)$. Finally, we set $i_N(p)= (j_N(p) + 1)
(p+1)m$.

Next, as in the proof of Theorem \ref{Thm:QRVcentrallimit}, we define an approximation of $D_{i,m} X$ by
\begin{equation*}
D_{i,m}^l  = \left( \sigma_{ \frac{l}{N}} \Delta_j^N W \right)_{i\leq j \leq i+m-1}
\end{equation*}
with $l\leq i$, and we set
\begin{equation*}
q_{i,l}^{sub}(m,\lambda )= g^2_{\lambda m} (\sqrt N D_{i,m}^l) + g^2_{m-\lambda m+1} (\sqrt N D_{i,m}^l).
\end{equation*}
We further set
\begin{equation*}
\Upsilon_{i,l}^N=  q^{sub}_{i,l}(m,\lambda ) - \mathbb{E} [q^{sub}_{i,l}(m,\lambda )| \mathcal F_{ \frac{l}{N}} ],
\end{equation*}
and
\begin{equation*}
\tilde \Upsilon_{j}^N= \left \{ \begin{array} {cc}
\frac{1}{\nu_1(m,\lambda)} \frac{1}{N} \Upsilon_{j,a_i(p)}^N, & j \in A_i(p) \\
\frac{1}{\nu_1(m,\lambda)} \frac{1}{N} \Upsilon_{j,b_i(p)}^N, & j \in B_i(p) \\
\frac{1}{\nu_1(m,\lambda)} \frac{1}{N} \Upsilon_{j,i_N(p)}^N, & j \geq i_N(p)
\end{array}
\right.
\end{equation*}
Finally, we define
\begin{equation*}
\zeta(p,1)_j^N = \sum_{l=a_j(p)}^{b_j(p) -1} \tilde \Upsilon_{l}^N~, \qquad \zeta(p,2)_j^N = \sum_{l=b_j(p)}^{a_{j+1}(p)-1} \tilde \Upsilon_{l}^N~,
\end{equation*}
and
\begin{equation*}
M(p)^N =  \sum_{j=0}^{j_N(p)} \zeta(p,1)_j^N~, \qquad  N(p)^N = \sum_{j=0}^{j_N(p)} \zeta(p,2)_j^N~, \qquad C(p)^N =  \sum_{j=i_n(p)}^{N} \tilde{\Upsilon}_j^N
\end{equation*}
Notice that the big blocks are collected in $M(p)^N$, the small blocks are contained in $N(p)^N$ and $C(p)^N$ is the sum of the border terms.

Recall that the quantities $M(p)^N, N(p)^N, C(p)^N$ are constructed from the approximations of the true returns. Consequently, we obtain the decomposition
\begin{equation} \label{identitydec}
\sqrt N \left( QRV^{sub}_{N} (m,\lambda)- IV \right) = \sqrt N \left( M(p)^N + N(p)^N + C(p)^N\right) + \gamma_N(p)~,
\end{equation}
where $\gamma_N(p)$ stands for the approximation error when $D_{i,m} X$ is replaced by $D_{i,m}^l$ (plus the error that is due to the replacement of $1/(N-m+1)$ by $1/N$ in the definition of $\tilde \Upsilon_{j}^N$, which is obviously asymptotically negligible). Exactly as in Part II and III of the proof of Theorem \ref{Thm:QRVcentrallimit} we deduce that
$$\lim_{p\rightarrow \infty } \limsup_{N\rightarrow \infty } P(|\gamma_N(p)|>\epsilon )=0$$
for any $\epsilon >0$. Since $C(p)^N$ contains a a fixed number of summands (for fixed $p$ and $m$) and each summand is of order $1/N$, we also obtain
\begin{equation*}
\lim_{p\rightarrow \infty } \limsup_{N\rightarrow \infty } P(\sqrt N|C(p)^N|>\epsilon )=0.
\end{equation*}
Now notice that the summands in the definition of $N(p)^N$ are uncorrelated. Consequently, we obtain
\begin{equation*}
\mathbb{E} \Big[|N(p)^N|^2\Big]\leq \frac{C}{p N},
\end{equation*}
 which implies that
\begin{equation*}
\lim_{p\rightarrow \infty } \limsup_{N\rightarrow \infty } P(\sqrt N|N(p)^N|>\epsilon )=0.
\end{equation*}
We are left to prove the CLT for $\sqrt N  M(p)^N$, where $M(p)^N = (M(p)^N_1, \ldots, M(p)^N_k)$ and each $M(p)^N_l$ is associated with $\lambda_l\in (1/2, 1)$. Set
\begin{equation*}
M(p)^N_l= \sum_{j=0}^{j_N(p)} \zeta(p,1)_{j,l}^N
\end{equation*}
to emphasize the dependence of $\zeta(p,1)_{j,l}^N$ on $\lambda_l$. As in Part I of the proof of Theorem \ref{Thm:QRVcentrallimit} we obtain
\begin{eqnarray*}
&& N^{1 / 2} \sum_{j=0}^{j_N(p)} \mathbb{E} [\zeta(p,1)_{j,l}^N \Delta W(p)_j^N|\mathcal F_{\frac{a_j(p)}{N}}]  = 0  \\[1.5 ex]
&& N^{1 / 2} \sum_{j=0}^{j_N(p)} \mathbb{E} [\zeta(p,1)_{j,l}^N \Delta H(p)_j^N|\mathcal
F_{\frac{a_j(p)}{N}}] = 0~,
\end{eqnarray*}
where $\Delta Y(p)_j^N = Y_{ \frac{b_j(p)}{N}} - Y_{ \frac{a_j(p)}{N}}$ for any process $Y$, and $H$ is a bounded martingale that is orthogonal to $W$. A straightforward computation shows that, for any fixed $p$,
\begin{equation*}
N \sum_{j=0}^{j_N(p)} \mathbb{E} [\zeta(p,1)_{j,l}^N \zeta(p,1)_{j,i}^N|\mathcal F_{\frac{a_j(p)}{N}}]  \overset{p}{ \to} \Theta^{sub}(m,\overline{\lambda})_{il} (p) IQ
\end{equation*}
with
\begin{eqnarray*}
&& \Theta^{sub}(m,\overline{\lambda})_{il} (p) = \frac{p}{p+1} \left( \frac{1}{m} \Theta(m,\overline{\lambda})_{il} \right. \\[1.5 ex]
&& \left. +\frac{2}{ \nu_1(m,\lambda_{i}) \nu_1 (m, \lambda_{l} )} \sum_{k=1}^{m-1} (1- \frac{k}{pm}) \mbox{cov}\left( | U_{ \left( m\lambda_i \right)}^{(0)} |^{2} + | U_{ \left( m - m\lambda_i + 1 \right)}^{(0)} |^{2} ,  | U_{ \left( m\lambda_l \right)}^{(k)} |^{2} + | U_{ \left( m -
m\lambda_l + 1 \right)}^{(k)} |^{2} \right) \right),
\end{eqnarray*}
where $U^{(0)}$ and $U^{(k)}$ are defined in Theorem \ref{Thm:QRVsubcentrallimit}. Now, we deduce by Theorem IX 7.28 in \citet*{jacod-shiryaev:03a}:
\begin{equation*}
\sqrt N  M(p)^N \overset{d_{s}}{\to} Y_p = MN \left( 0, \Theta^{sub}(m,\overline{\lambda})(p) IQ \right)
\end{equation*}
for any fixed $p$. On the other hand, we have that $Y_p \overset{p}{ \to} Y=MN \left( 0,\Theta^{sub}(m,\overline{\lambda}) IQ \right)$ when $p\rightarrow \infty $. This completes the proof of Theorem \ref{Thm:QRVsubcentrallimit}. \qed
\end{proof}

\begin{proof}[Proof of Proposition \ref{prop:QRVmtoinfty}]
First, recall that $U_{(\lambda m)} \overset{p}{ \to} c_{\lambda }$ as $m\rightarrow \infty $. Thus, we immediately obtain the identities
\begin{equation*}
\nu_1(\lambda) = 2c_{ \lambda}^{2}~, \qquad
\nu_1(\lambda_i,\lambda_j)= 4c_{\lambda_i}^{2}c_{\lambda_j}^{2}.
\end{equation*}
To prove the other identities we derive a joint CLT for $(U_{(\lambda_1 m)}^{(0)}, U_{(\lambda_2 m)}^{(k)})$ for $\lambda_1\geq \lambda_2$ (and both are in the interval $(1/2,1)$). Let $\Phi$ denote the $N(0,1)$-distribution and $\phi $ its density. The crucial step is the so-called Bahadur representation that says:
\begin{equation*}
U_{(\lambda_1 m)}^{(0)} - c_{\lambda_1}= \frac{\lambda_1 - \hat F_n(c_{\lambda_1})}{\phi (c_{\lambda_1})} + O(m^{-3/4} \log m) ~~ a.s.,
\end{equation*}
where
\begin{equation*}
\hat F_n(c_{\lambda_1})= \frac 1m \sum_{i=1}^m 1_{\{U_i\leq c_{\lambda_1}\}}
\end{equation*}
is the empirical distribution function. Clearly, we also have
\begin{equation*}
U_{(\lambda_2 m)}^{(k)} - c_{\lambda_2}= \frac{\lambda_2 - \hat F_n(c_{\lambda_2})^{(k)}}{\phi (c_{\lambda_2})} + O(m^{-3/4} \log m) ~~ a.s.,
\end{equation*}
where
\begin{equation*}
\hat F_n(c_{\lambda_2})^{(k)}= \frac 1m \sum_{i=1+k}^{m+k} 1_{\{U_i\leq c_{\lambda_2}\}}.
\end{equation*}
Now we apply the CLT for empirical distribution functions. Assume that $k/m = \mu + o(m^{-1/2})$ ($\mu\in [0,1]$). We obtain the following CLT:
\begin{equation*}
\sqrt m \left( \begin{array} {c} \frac 1m \sum_{i=1}^k 1_{\{U_i\leq c_{\lambda_1}\}} - \mu \lambda_1 \\
\frac 1m \sum_{i=k+1}^{m} 1_{\{U_i\leq c_{\lambda_1}\}} - (1-\mu) \lambda_1 \\
\frac 1m \sum_{i=k+1}^{m} 1_{\{U_i\leq c_{\lambda_2}\}} - (1-\mu) \lambda_2 \\
\frac 1m \sum_{i=m}^{m+k} 1_{\{U_i\leq c_{\lambda_2}\}} - \mu \lambda_2
\end{array}\right) \overset{d}{\to} N(0, \Sigma )
\end{equation*}
with
\begin{equation*}
\Sigma = \left( \begin{array} {cccc}
\mu \lambda_1 (1-\lambda_1) & 0 & 0 & 0 \\
0& (1-\mu) \lambda_1 (1-\lambda_1) & (1-\mu) \lambda_2 (1-\lambda_1)& 0 \\
0& (1-\mu) \lambda_2 (1-\lambda_1) & (1-\mu) \lambda_2 (1-\lambda_2)& 0 \\
0 & 0 & 0 & \mu \lambda_2 (1-\lambda_2)
\end{array}\right)
\end{equation*}
Next, we apply the $\Delta$-method for the function:
\begin{equation*}
f(x,y,z,w)= \Big(\frac{1}{\phi (c_{\lambda_1})} (x+y) , \frac{1}{\phi (c_{\lambda_2})} (z+w) \Big).
\end{equation*}
Clearly,
\begin{equation*}
Df(x,y,z,w)= \left( \begin{array} {cccc}
\frac{1}{\phi (c_{\lambda_1})}   & \frac{1}{\phi (c_{\lambda_1})} & 0 & 0 \\
0& 0 & \frac{1}{\phi (c_{\lambda_2})} & \frac{1}{\phi (c_{\lambda_2})}
\end{array}\right)
\end{equation*}
and we obtain the CLT (set $Df=Df(\mu \lambda_1, (1-\mu) \lambda_1, \mu \lambda_2, (1-\mu) \lambda_2)$):
\begin{equation*}
\sqrt m \left( \begin{array} {c} U_{(\lambda_1 m)}^{(0)} - c_{\lambda_1} \\
U_{(\lambda_2 m)}^{(k)} - c_{\lambda_2}
\end{array}\right) \overset{d}{\to} N(0, Df~ \Sigma~ (Df)')
\end{equation*}
with
\begin{equation*}
Df~ \Sigma~ (Df)') = \left( \begin{array} {cc}
\frac{\lambda_1(1-\lambda_1)}{\phi^2 (c_{\lambda_1})}   & \frac{(1-\mu)\lambda_2(1-\lambda_1)}{\phi (c_{\lambda_1}) \phi (c_{\lambda_2})} \\[1.5 ex]
\frac{(1-\mu)\lambda_2(1-\lambda_1)}{\phi (c_{\lambda_1}) \phi (c_{\lambda_2})} & \frac{\lambda_2(1-\lambda_2)}{\phi^2 (c_{\lambda_2})}
\end{array}\right)
\end{equation*}
Now apply again the $\Delta$-method  for the function $f(x,y)= (x^2,y^2)$:
\begin{equation*}
\sqrt m \left( \begin{array} {c} | U_{(\lambda_1 m)}^{(0)}|^2 - c_{\lambda_1}^2 \\
|U_{(\lambda_2 m)}^{(k)}|^2 - c_{\lambda_2}^2
\end{array}\right) \overset{d}{\to} N\left(0, 4\left(
\begin{array} {cc}
\frac{c_{\lambda_1}^2 \lambda_1(1-\lambda_1)}{\phi^2 (c_{\lambda_1})}   & \frac{c_{\lambda_1} c_{\lambda_2}(1-\mu)\lambda_2(1-\lambda_1)}{\phi (c_{\lambda_1}) \phi (c_{\lambda_2})} \\[1.5 ex]
\frac{c_{\lambda_1} c_{\lambda_2}(1-\mu)\lambda_2(1-\lambda_1)}{\phi (c_{\lambda_1}) \phi (c_{\lambda_2})} & \frac{c_{\lambda_2}^2 \lambda_2(1-\lambda_2)}{\phi^2 (c_{\lambda_2})}
\end{array}\right)  \right)
\end{equation*}
From the latter CLT we deduce that
\begin{equation*}
\Theta(\overline{\lambda})_{12} = 2\frac{ (1 - \lambda_1 ) ( 2 \lambda_2 - 1)}{ \phi \bigl( c_{ \lambda_1} \bigr) \phi \bigl( c_{
\lambda_2} \bigr)  c_{ \lambda_1} c_{ \lambda_2} }.
\end{equation*}
Finally, recall that
\begin{align*}
\Theta^{sub}(m,\overline{\lambda})_{12} &= \frac{1}{m} \Theta(m,\overline{\lambda})_{12} \\
&+ \frac{2}{ \nu_1(m,\lambda_{1}) \nu_1 (m, \lambda_{2} )} \sum_{k=1}^{m-1} \mbox{cov}\left( | U_{ \left( m\lambda_1 \right)}^{(0)} |^{2} + | U_{ \left(
m - m\lambda_1 + 1 \right)}^{(0)} |^{2} ,  | U_{ \left( m\lambda_2 \right)}^{(k)} |^{2} + | U_{ \left( m - m\lambda_2 + 1 \right)}^{(k)} |^{2} \right).
\end{align*}
Clearly
\begin{equation*}
\frac{1}{m} \Theta(m,\overline{\lambda})_{12} \rightarrow 0, \quad \nu_1(m,\lambda_{i}) \rightarrow 2c_{\lambda_{i}}^2, \quad \frac{1}{m} \sum_{k=1}^{m-1} (1-\frac{k}{m}) \rightarrow 1/2
\end{equation*}
as $m\rightarrow \infty $. Hence, for the term $\Theta^{sub} (m, \overline \lambda )_{12}$ we asymptotically obtain (by replacing again $\mu $ by $k/m$):
\begin{equation*}
\Theta^{sub} (m, \overline \lambda )_{12} \rightarrow 2\frac{ (1 - \lambda_1 ) ( 2 \lambda_2 - 1)}{ \phi \bigl( c_{ \lambda_1} \bigr) \phi \bigl( c_{
\lambda_2} \bigr)  c_{ \lambda_1} c_{ \lambda_2} }.
\end{equation*}\qed
\end{proof}

\begin{proof}[Proof of Theorem \ref{Thm:QRV*consistency}]
As in Theorem \ref{Thm:QRVconsistency} it is sufficient to consider the 1-dimensional case $k=1$, $\overline{\lambda} = \lambda \in (1/2, 1)$. For $l\leq i$ we set
$$\beta_{i}^{*N} = N^{1 / 4} \{ \sigma_{ \frac{i}{N}} \overline{W}_{i+(j-1)K}^{N} + \overline{u}_{i+(j-1)K}^{N}\}_{j = 1}^m~,$$
and define
\begin{equation} \label{wprime}
w_{i}^{*(n,m)} = g^2_{\lambda m}(\beta_{i}^{*N}) + g^2_{m-\lambda m +1}(\beta_{i}^{*N}).
\end{equation}
First, we state the following Lemma, which is crucial for our proofs (especially for the CLT).
\begin{lemma} \label{noisemom}
Assume that $\mathbb{E} \left( u^{4}_i \right)<\infty$. Then it holds that
\begin{eqnarray*}
\mathbb{E} [w_{i}^{*(n,m)}| \mathcal{F}_{ \frac{i}{N}} ] = \nu_1(m,\lambda) ( c\psi_2 \sigma_{ \frac{i}{N}}^2 + \frac{c}{\psi_1} \omega^2) + o_p(N^{-1/4})
\end{eqnarray*}
uniformly in $i$.
\end{lemma}
\begin{proof} [Proof of Lemma \ref{noisemom}]
By representation \eqref{Eqn:Ybaerep}  we obtain the (conditional) convergence in distribution
\begin{equation*}
\beta_i^{*N} |\mathcal{F}_{ \frac{i}{N}} \overset{d}{ \to} N_{m} \biggl(0, \mbox{diag} \Bigl( c\psi_2 \sigma_{ \frac{i}{N}}^2 + \frac{c}{\psi_1} \omega^2, \ldots, c\psi_2 \sigma_{ \frac{i}{N}}^2 + \frac{c}{\psi_1} \omega^2 \Bigr) \biggr).
\end{equation*}
Denote by $P_m^N$ the distribution of the left-hand side and by $\Phi_m$ the distribution of the right-hand side (by $\phi_m$ we denote the Lebesgue density of $\Phi_m$). Now we apply the Edgeworth-expansion result presented in Lahiri (see Theorem 6.1 and 6.2 therein):
\begin{equation*}
\left| \int (g^2_{\lambda m} + g^2_{m-\lambda m +1}) d(P_m^N - \Phi_m -N^{-1/4}P_{m,1} ) \right | = o_p(N^{-1/4}),
\end{equation*}
which holds under the condition $\mathbb{E} \left( u^{4}_i \right)<\infty$. The quantity $P_{m,1}$ is the second order term of the Edgeworth-expansion that has an odd density $p_{m,1}$ that is given by
\begin{equation*}
p_{m,1}(x)= -\sum_{|\nu| = 3} \frac{\chi_\nu}{\nu!} D^\nu \phi_m (x)~,
\end{equation*}
where $\nu = (\nu_1, \ldots, \nu_m)$ is a nonnegative integer vector ($|\nu| = \nu_1+  \cdots + \nu_m$, $\nu!=\nu_1!  \cdots  \nu_m!$), $\chi_\nu$ are constants that depend on the cumulants of the marginal distribution of the process $u$ and $D^\nu = \partial / \partial x_1^{\nu_1} \cdots \partial / \partial x_m^{\nu_m}$. Since $g^2_{\lambda m} + g^2_{m-\lambda m +1}$ is an even function, we deduce that
\begin{equation*}
 \int (g^2_{\lambda m} + g^2_{m-\lambda m +1}) dP_m (1) =0
\end{equation*}
(because $p_{m,1}$ is odd). Thus
\begin{eqnarray*}
\mathbb{E} [w_i^{*(n,m)}| \mathcal{F}_{ \frac{i}{N}} ] = \nu_1(m,\lambda) ( c\psi_2 \sigma_{ \frac{i}{N}}^2 + \frac{c}{\psi_1} \omega^2) + o_p(N^{-1/4}),
\end{eqnarray*}
which holds uniformly in $i$ because $\sigma $ is bounded. This completes the proof of Lemma \ref{noisemom}. \qed
\end{proof}
By the same methods as presented in the proof of Theorem \ref{Thm:QRVconsistency} (see e.g. \eqref{Eqn:prob1}) we conclude that
\begin{equation} \label{eqest}
QRV^\ast_N (m,\lambda ) - \frac{1}{\nu_1(m,\lambda)} \frac{1}{c\psi_2 (N-m(K-1)+1)} \sum_{i = 0}^{N - m(K-1) } w_{i}^{*(n,m)} \overset{p}{ \to} 0.
\end{equation}
Moreover, by Lemma \ref{noisemom}, the convergence
\begin{equation} \label{eqapp}
\frac{1}{\nu_1(m,\lambda)} \frac{1}{c\psi_2 (N-m(K-1)+1)} \sum_{i = 0}^{N - m(K-1)} \mathbb{E} [w_i^{*(n,m)} | \mathcal{F}_{ \frac{i}{N}} ] \overset{p}{ \to} IV + \frac{\psi_1 \omega^{2}}{c^{2} \psi_2}
\end{equation}
holds. In view of \eqref{eqest} and  \eqref{eqapp} we are left to proving
\begin{equation} \label{pr}
\frac{1}{\nu_1(m,\lambda)} \frac{1}{c\psi_2 (N-m(K-1)+1)} \sum_{i = 0}^{N - m(K-1)} \eta_i^{*N} \overset{p}{ \to} 0, \qquad \eta_i^{*N}= w_i^{*(n,m)} - \mathbb{E} [w_i^{*(n,m)}| \mathcal{F}_{ \frac{i}{N}}].
\end{equation}
Observe that due to the construction of $\beta_i^{*N}$, the boundedness of $\sigma$ and $\mathbb{E} \left( u_{i}^{4} \right)<\infty $ we obtain the estimate
\begin{equation} \label{eqcovest}
\mathbb{E} [\eta_i^{*N} \eta_j^{*N} ]\leq C~, \qquad |i-j|<m(K-1)~,
\end{equation}
whereas $\mathbb{E} [\eta_i^{*N} \eta_j^{*N} ]=0$ for $|i-j|\geq m(K-1)$. Since $m$ is fixed, we deduce the estimate
\begin{equation}
\mathbb{E} \biggl[ \Big|\frac{1}{\nu_1(m,\lambda)} \frac{1}{c\psi_2 (N-m(K-1)+1)} \sum_{i = 0}^{N - m(K-1)} \eta_i^{*N}  \Big|^2 \biggr] \leq \frac{C}{K},
\end{equation}
which completes the proof of Theorem \ref{Thm:QRV*consistency}. \qed
\end{proof}

\begin{proof}[Proof of Theorem \ref{Thm:QRV*centrallimit}]

We assume that the process $X$ is continuous (the robustness to finite activity jumps is shown as in Theorem \ref{Thm:QRVcentrallimit}) and show the CLT for the noise robust statistic $QRV_N^{\ast}(m,\overline{\lambda},\alpha)$. The summands $q_i^{\ast}(m, \lambda_j)$ ($1\leq j\leq k$) in the definition of $QRV_N^{sub}(m,\overline{\lambda},\alpha)$ are now $m(K-1)$-dependent, so we have to apply ``big blocks \& small blocks''-technique once again to break this dependence. More precisely, we will build big blocks of size $pm(K-1)$, which will be separated by a small block of size $m(K-1)$ (again this procedure ensures the (conditional) independence of big blocks, whereas the small blocks become asymptotically negligible when we later let $p$ converge to infinity). Quite often we will use the same notations as in the proof of Theorem \ref{Thm:QRVsubcentrallimit} to emphasize the strong parallels between these proofs.

First, set
$$a_i(p)=i(p+1)m(K-1),~\qquad b_i(p)=i(p+1)m(K-1) + pm(K-1)~,$$
and let $A_i(p)$ denote the set of integers $l$ with $a_i(p)\leq l<b_i(p)$ and $B_i(p)$  the set of integers $l$ with $b_i(p)\leq l<a_{i+1}(p)$. Furthermore, let $j_N(p)$ denote the largest integer $j$ with $b_j(p)\leq N$. Notice that $j_N(p)=O(\sqrt N / p)$. Finally, we set $i_N(p)= (j_N(p) + 1)(p+1)m(K-1)$.

Next, we define an approximation of $\overline{D}_i^N Y$ by
\begin{equation*}
\overline{D}_{i,l}^N  = \{ \sigma_{ \frac{l}{N}} \overline{W}_{i+(j-1)(K-1)}^{N} + \overline{u}_{i+(j-1)(K-1)}^{N}\}_{j = 1}^m
\end{equation*}
with $l\leq i$, and we set
\begin{equation*}
q^\ast_{i,l}(m,\lambda )= g^2_{\lambda m} (N^{1/4} \overline{D}_{i,l}^N ) + g^2_{m-\lambda m+1} (N^{1/4} \overline{D}_{i,l}^N ).
\end{equation*}
We further set
\begin{equation*}
\Upsilon_{i,l}^N=  q^\ast_{i,l}(m,\lambda ) - \mathbb{E} [q^\ast_{i,l}(m,\lambda )| \mathcal F_{ \frac{l}{N}} ],
\end{equation*}
and
\begin{equation*}
\tilde \Upsilon_{j}^N= \left \{ \begin{array} {cc}
\frac{1}{\nu_1(m,\lambda)} \frac{1}{c\psi_2 N} \Upsilon_{j,a_i(p)}^N, & j \in A_i(p) \\
\frac{1}{\nu_1(m,\lambda)} \frac{1}{c\psi_2 N} \Upsilon_{j,b_i(p)}^N, & j \in B_i(p) \\
\frac{1}{\nu_1(m,\lambda)} \frac{1}{c\psi_2 N} \Upsilon_{j,i_N(p)}^N, & j \geq i_N(p)
\end{array}
\right.
\end{equation*}
Finally, we define
\begin{equation*}
\zeta(p,1)_j^N = \sum_{l=a_j(p)}^{b_j(p) -1} \tilde \Upsilon_{l}^N~, \qquad \zeta(p,2)_j^N = \sum_{l=b_j(p)}^{a_{j+1}(p)-1} \tilde \Upsilon_{l}^N~,
\end{equation*}
and
\begin{equation*}
M(p)^N =  \sum_{j=0}^{j_N(p)} \zeta(p,1)_j^N~, \qquad  N(p)^N = \sum_{j=0}^{j_N(p)} \zeta(p,2)_j^N~, \qquad C(p)^N =  \sum_{j=i_n(p)}^{N} \tilde{\Upsilon}_j^N
\end{equation*}
Notice that the big blocks are collected in $M(p)^N$, the small blocks are contained in $N(p)^N$ and $C(p)^N$ is the sum of the border terms.

Recall that the terms $M(p)^N, N(p)^N, C(p)^N$ are constructed from the approximations $\overline{D}_{i,l}^N$. Thus
\begin{equation} \label{eqneg1}
N^{1 / 4} \left( QRV^{ \ast}_{N} (m,\overline{\lambda},\alpha)- \frac{\psi_1}{c^2 \psi_2} \widehat{\omega}^2 - IV \right) = N^{1 / 4} \left( M(p)^N + N(p)^N + C(p)^N\right) + \gamma_N(p)~,
\end{equation}
where $\gamma_N(p)$ stands for the approximation error when $\overline{D}_i^N Y$ is replaced by $\overline{D}_{i,l}^N$ (plus the error that is due to the replacement of $1/(N-m(K-1)+1)$ by $1/N$ in the definition of $\tilde \Upsilon_{j}^N$, which is obviously asymptotically negligible). Now, the convergence
$$\lim_{p\rightarrow \infty } \limsup_{N\rightarrow \infty } P(|\gamma_N(p)|>\epsilon )=0,$$
for any $\epsilon >0$, follows through the lines of Part II and III of the proof of Theorem \ref{Thm:QRVcentrallimit}. In fact, there are two additional difficulties that has to be shown in a different way. First of all we have to prove (under assumption (V)) that
\begin{equation*}
\frac{1}{\nu_1(m,\lambda)} \frac{1}{c\psi_2 (N-m(K-1)+1)} \sum_{i = 0}^{N - m(K-1)} \mathbb{E} [w_i^{*(n,m)} | \mathcal{F}_{ \frac{i}{N}} ] - \left( IV + \frac{\psi_1 \omega^{2}}{c^{2} \psi_2} \right) = o_p(N^{-1/4})
\end{equation*}
(which is the counterpart of \eqref{convsigma}). But the afore-mentioned estimate follows immediately from Lemma \ref{noisemom} and \eqref{convsigma}. The other difficulty arises when we have to deal with the counterpart of the identity \eqref{nullidentity}, which in this case would involve a noise term (and have a slightly different form). However, this type of identity remains true, because  $\left( W, V, u \right) \overset{d}{=} - \left( W, V, u \right)$ since the marginal distribution of $u$ is assumed to be symmetric around $0$.

Since $C(p)^N$ contains at most $(p+1)m(K-1)$ summands and each summand is of order $1/N$, we also obtain
\begin{equation*}
\lim_{p\rightarrow \infty } \limsup_{N\rightarrow \infty } P( N^{1/4} |C(p)^N|>\epsilon )=0.
\end{equation*}
Now notice that the summands in the definition of $N(p)^N$ are uncorrelated and each summand is of order $K/N$. Consequently, we obtain
\begin{equation*}
\mathbb{E} \Big[|N(p)^N|^2\Big]\leq \frac{C}{p \sqrt N},
\end{equation*}
 which implies that
\begin{equation*}
\lim_{p\rightarrow \infty } \limsup_{N\rightarrow \infty } P(N^{1/4}|N(p)^N|>\epsilon )=0.
\end{equation*}
We are left to proving the CLT for $N^{1/4}  M(p)^N$, where $M(p)^N = (M(p)^N_1, \ldots, M(p)^N_k)$ and each $M(p)^N_l$ is associated with $\lambda_l\in (1/2, 1)$. Set
\begin{equation*}
M(p)^N_l= \sum_{j=0}^{j_N(p)} \zeta(p,1)_{j,l}^N
\end{equation*}
to emphasize the dependence of $\zeta(p,1)_{j,l}^N$ on $\lambda_l$ (we also associate $\Upsilon_{j,i}^N (l)$ with $\lambda_l$).

As in Part I of the proof of Theorem \ref{Thm:QRVcentrallimit} we obtain
\begin{eqnarray*}
N^{1 / 4} \sum_{j=0}^{j_N(p)} \mathbb{E} [\zeta(p,1)_{j,l}^N \Delta W(p)_j^N|\mathcal F_{\frac{a_j(p)}{N}}]  = 0,
\end{eqnarray*}
where $\Delta Y(p)_j^N = Y_{ \frac{b_j(p)}{N}} - Y_{ \frac{a_j(p)}{N}}$ for any process $Y$. As in \cite{jacod-li-mykland-podolskij-vetter:09a} (see (5.58)) we have that
\begin{eqnarray*}
N^{1 / 4} \sum_{j=0}^{j_N(p)} \mathbb{E} [\zeta(p,1)_{j,l}^N \Delta H(p)_j^N|\mathcal F_{\frac{a_j(p)}{N}}] \overset{p}{ \to}  0
\end{eqnarray*}
for any bounded martingale $H$ that is orthogonal to $W$.

Now notice that \eqref{Eqn:Ybaerep} implies the identities
\begin{equation} \label{eqcor}
N^{1 / 2} \mathbb{E} [\overline{W}_j^{N} \overline{W}_i^{N}] = c w_{h}\Big( \frac{|j - i|}{K} \Big)  + O(K^{-1}), \quad N^{1 / 2} \mathbb{E} [\overline{u}_j^{N} \overline{u}_i^{N}] = \frac{1}{c} w_{h'} \Big( \frac{|j-i|}{K} \Big) \omega^2 + O(K^{-1}), \quad |j-i|<K-1,
\end{equation}
where the functional $w_{f} (u)$ is given in Definition \ref{def2}. When $|j-i|\geq K-1$ the above covariances are $0$. Assume now that $j>i$ with $j=i +(l-1)(K-1) +d$ for some $1\leq l\leq m$ and $0\leq d\leq K-2$, and $a_z(p)\leq i,j\leq b_z(p) -1 $. Due to the identities in \eqref{eqcor} we obtain
$$\mathbb{E} [\Upsilon_{j,a_z(p)}^N (l_1) \Upsilon_{i,a_z(p)}^N (l_2)| \mathcal F_{\frac{a_z(p)}{N}} ]= f_{m,l,\sigma_{a_z(p)}, \frac{d}{K}} (\lambda_{l_1},  \lambda_{l_2}) + O_p(K^{-1})~,$$
where $f_{m,l,x,u}(\lambda_{l_1},  \lambda_{l_2})$ is given in Definition \ref{def2} ($1\leq l_1,l_2\leq k$). This implies that
\begin{eqnarray*}
&&N^{1 / 2} \sum_{j=0}^{j_N(p)} \mathbb{E} [\zeta(p,1)_{j, l_1}^N \zeta(p,1)_{j, l_2}^N|\mathcal F_{\frac{a_j(p)}{N}}] \\[1.5 ex]
&&\overset{p}{ \to} \frac{p}{p + 1} \frac{2 \nu_{1, m}^{-1} \left( \lambda_{l_1} \right) \nu_{1, m}^{-1} \left( \lambda_{l_2} \right) } {  c \psi_2^2} \sum_{l =1}^{m}  \Big(1- \frac{l - 1}{pm} \Big)\int_{0}^{1} \int_{0}^{1} f_{m, l, \sigma_{t}, u} \left( \lambda_{l_1},  \lambda_{l_2} \right) \text{\upshape{d}}t \text{\upshape{d}}u = \Gamma(p)_{l_1, l_2}.
\end{eqnarray*}
Now, we deduce by Theorem IX 7.28 in \citet*{jacod-shiryaev:03a}:
\begin{equation*}
\sqrt N  M(p)^N \overset{d_{s}}{\to} Z_p = MN \left( 0, \Gamma(p) \right)
\end{equation*}
for any fixed $p$. On the other hand, we have that $Z_p \overset{p}{ \to} Z=MN \left( 0, \frac{2}{c\psi_2^2} \Sigma_m (\lambda_1, \ldots, \lambda_k) \right)$ when $p\rightarrow \infty $. This completes the proof of Theorem \ref{Thm:QRV*centrallimit}. \qed
\end{proof}

\begin{proof}[Proof of Proposition \ref{prop:QRV*feasible}]

Due to the polarization identity it is sufficient to prove Proposition \ref{prop:QRV*feasible} for $k=1$, $\overline \lambda = \lambda \in (1/2,1)$. Recall the definition of $w_i^{*(n,m)}$ in (\ref{wprime}). As in the proof of Theorem \ref{Thm:QRV*consistency} we have that
\begin{eqnarray*}
&&\frac{\nu_1^{-2}(m,\lambda)} {c \psi_2^2 (K-1) (N-3m(K-1)+3)} \sum_{i = m(K-1)-1}^{N - 2m(K-1) + 1} \left( q^\ast_i(m,\lambda) \left(\sum_{j=i-m(K-1)+1}^{i+m(K-1)-1} \{q^\ast_j(m,\lambda) - q^\ast_{i+m(K-1)}(m,\lambda)\} \right) \right. \\[1.5 ex]
&&\left.- w_i^{*(n,m)} \left(\sum_{j=i-m(K-1)+1}^{i+m(K-1)-1} \{w_j^{*(n,m)} - w_{i+m(K-1)}^{*(n,m)}\} \right) \right) \overset{p}{ \to} 0.
\end{eqnarray*}
Next, a straightforward (but somewhat tedious) calculation shows that
\begin{align*}
&\frac{\nu_1^{-2}(m,\lambda)} {c \psi_2^2 (K-1) (N-3m(K-1)+3)} \times \\[1.5 ex]
& \sum_{i = m(K-1)-1}^{N - 2m(K-1) + 1} \mathbb{E} \left[ w_i^{*(n,m)} \left(\sum_{j=i-m(K-1)+1}^{i+m(K-1)-1} \{w_j^{*(n,m)} - w_{i+m(K-1)}^{*(n,m)}\} \right) \Big|\mathcal F_{\frac{i-m(K-1)+1}{N}} \right] \\[1.5 ex]
&\overset{p}{ \to} \frac{2 \nu_{1, m}^{-2} ( \lambda )}{  c\psi_2^2} \sum_{l = 1}^{m} \int_{0}^{1} \int_{0}^{1} f_{m, l, \sigma_{t}, u} \left( \lambda \right) \text{\upshape{d}}t \text{\upshape{d}}u.
\end{align*}
On the other hand, we deduce that
\begin{eqnarray*}
&&\frac{\nu_1^{-2}(m,\lambda)} {c \psi_2^2 (K-1) (N-3m(K-1)+3)} \sum_{i = m(K-1)-1}^{N - 2m(K-1) + 1} \left(w_i^{*(n,m)} \left(\sum_{j=i-m(K-1)+1}^{i+m(K-1)-1} \{w_j^{*(n,m)} - w_{i+m(K-1)}^{*(n,m)}\} \right) \right. \\[1.5 ex]
&&\left.- \mathbb{E} \left[w_i^{*(n,m)} \left(\sum_{j=i-m(K-1)+1}^{i+m(K-1)-1} \{w_j^{*(n,m)} - w_{i+m(K-1)}^{*(n,m)}\} \right) \Big|\mathcal F_{\frac{i-m(K-1)+1}{N}} \right] \right) \overset{p}{ \to} 0
\end{eqnarray*}
as in (\ref{pr}). This completes the proof of Proposition \ref{prop:QRV*feasible}. \qed
\end{proof}

\pagebreak

\renewcommand{\baselinestretch}{1.0}
\small
\bibliographystyle{rfs}
\bibliography{userref}

\end{document}